%% file: bodyCadArxiv.tex
\documentclass{article} 
\pdfoutput=1
\usepackage{graphicx,amssymb,amsmath,amsthm,subfig,url,wrapfig}
\usepackage{multirow} 
\usepackage{colortbl}

 \newtheorem{thm}{Theorem}[section]

 \numberwithin{equation}{section}
	
\newtheorem{fact}{Fact}

\definecolor{lightgray}{gray}{0.8}

\newcounter{algorithmCounter}
\newtheorem{myAlg}[algorithmCounter]{Algorithm}

\renewcommand{\Re}{\mathbb{R}}
\newcommand{\vc}[1]{\ensuremath{\bm {#1}}} 
\renewcommand{\vec}[1]{\ensuremath{\mathbf {#1}}}
\newcommand{\mtx}[1]{\ensuremath{\mathbf {#1}}}
\newcommand{\SE}{\text{SE}} 
\newcommand{\se}{\text{se}}

\newcommand{\so}{\text{so}}
\newcommand{\smallcdots}{\cdot\hspace{-1mm}\cdot\hspace{-1mm}\cdot}
\newcommand{\rmfill}{\hspace{-1mm}\smallcdots \vec 0 \hspace{-1mm}\cdot\hspace{-1mm}\cdot\hspace{-.2mm}\cdot\hspace{-.8mm}}
\usepackage{bm}

\newcommand{\out}[1]{}

\newcommand{\comment}[1]{}

\long\def\symbolfootnote[#1]#2{\begingroup%
\def\thefootnote{\fnsymbol{footnote}}\footnote[#1]{#2}\endgroup}

\begin{document}	
	\title{Body-and-cad geometric constraint systems\footnote{
		An abbreviated version appeared in: 24th Annual ACM Symposium on Applied Computing, 
		Technical Track on Geometric Constraints and Reasoning GCR'09, Honolulu, HI, 2009. The work in this paper is based on Lee-St.John's Ph.D. dissertation \cite{lee:PhDThesis:2008}.}}
		\author{
		%
		 Kirk Haller\\
		       SolidWorks Corporation\\
		       300 Baker Avenue\\
		       Concord, MA 01742\\
		       \texttt{khaller@solidworks.com}
		\and Audrey Lee-St. John\footnote{Partially supported by NSF grant CCF-0728783 of Ileana Streinu, an NSF Graduate Fellowship and SolidWorks Corporation.}\\
		       Mount Holyoke College\\
		       South Hadley, MA 01075\\
		       \texttt{astjohn@mtholyoke.edu}
		\and Meera Sitharam \footnote{Partially supported by a Research Grant from SolidWorks 2007.}\\
		       University of Florida\\
		       Gainesville, FL 32611\\
		       \texttt{sitharam@cise.ufl.edu}
		\and Ileana Streinu\footnote{Partially supported by NSF grant CCF-0728783 and a DARPA ``23 Mathematical Challenges'' grant.}\\
		       Smith College\\
		       Northampton, MA 01063\\
		       \texttt{istreinu@smith.edu}
		\and Neil White \\ 
			       University of Florida\\
			       Gainesville, FL 32611\\
			       \texttt{white@math.ufl.edu}
		}
\maketitle
\begin{abstract}
Motivated by
constraint-based CAD
 software, we
develop the foundation for the rigidity theory of
a very general model: 
the {\em body-and-cad structure}, composed of rigid bodies
in 3D constrained by pairwise {\bf c}oincidence, {\bf a}ngular
and {\bf d}istance constraints. We identify 21 relevant 
geometric constraints and develop the corresponding
infinitesimal rigidity theory for these structures. The 
classical body-and-bar rigidity model can be viewed as a body-and-cad 
structure that uses only one constraint from this new class.

As a consequence, we identify a new, necessary, but not sufficient,
counting condition for {\em minimal rigidity} of body-and-cad structures: {\em nested sparsity}. 
This is a slight generalization of the well-known sparsity condition of Maxwell.

\end{abstract}

\input{intro.tex}

\input{prelims.tex}
\input{bodyCad.tex}

\input{allConstraints.tex}

\input{combinatorics.tex}

\input{algs.tex}
\input{conclusions.tex}

\bibliographystyle{abbrv}
 \bibliography{geomConstraints}
\end{document}

%% file: intro.tex
\section{Introduction}
This paper initiates the study of and sets up the foundation for the 
rigidity theory of a large class of 3D geometric constraint systems. 
These systems are composed of rigid bodies with specific 
{\bf c}oincidence, {\bf a}ngular and {\bf d}istance constraints
and are called {\em body-and-{\bf cad} structures}. To the best of our 
knowledge, these constraints have not been systematically studied 
before from this perspective.

\smallskip\noindent{\bf Motivation.} 
Popular computer aided design (CAD) software applications
based on geometric constraint solvers 
allow users to design complex 3D systems by placing geometric constraints
among sets of rigid body building blocks. The constraints are specified by identifying
{\em geometric elements} (points, lines, planes, or splines) on
participating rigid bodies. Detecting when a user has created a fully-defined
sub-system or has added a redundant (or inconsistent) constraint 
are important problems for providing informative feedback.
However, analyzing all constraints simultaneously is a very 
difficult problem. In this paper, we focus on a subset of these constraints
that are amenable to a rigidity-theoretical investigation.

Underlying classical rigidity theory results is a general proof pattern, spanning algebraic geometry (for {\em rigidity}), linear algebra (for {\em infinitesimal rigidity}) and graph theory (for {\em combinatorial rigidity}). The ultimate goal is a full combinatorial characterization of {\em generically  minimally rigid structures}, but
such results are extremely rare: 3D
bar-and-joint rigidity remains a conspicuously open 
problem \cite{graver:servatius:rigidityBook:1993}, while the 2D version is fully understood \cite{laman}.
An important step along the way is identifying a pattern in the rigidity
matrix developed as part of the {\em infinitesimal rigidity theory} for the structures. 
While this is straightforward for the well-known bar-and-joint model, it
is more complicated in the body-and-bar model. In this paper, we formulate 
the even more involved rigidity matrix for the {\em body-and-cad} model.

\smallskip\noindent{\bf Results.}
We define a {\em body-and-cad structure} to be composed of rigid bodies
connected by {\em pairwise} {\em {\bf c}oincidence, {\bf a}ngular} (parallel, perpendicular, 
or arbitrary fixed angular) and {\em {\bf d}istance} constraints. The
constraints occur between specified points, lines
or planes (called {\em geometric elements}). 
Besides the well-studied distance constraint between points (as in body-and-bar structures), 
we identify 20 new 
pairwise constraints.
We label constraints by the geometric elements involved, 
e.g., a {\em line-plane} perpendicular constraint between bodies $A$ and $B$ indicates that
a {\em line} on $A$ is perpendicular to a {\em plane}
on $B$. The complete set of body-and-cad constraints that we study is
further subdivided into six categories:
\begin{itemize}
	\item {\bf Point-point constraints:} coincidence, distance.
	\item {\bf Point-line constraints:} coincidence, distance.
	\item {\bf Point-plane constraints:} coincidence, distance.
	\item {\bf Line-line constraints:} parallel, perpendicular, fixed angular,
	coincidence, distance.
	\item {\bf Line-plane constraints:} parallel, perpendicular, fixed angular,
	coincidence, distance.
	\item {\bf Plane-plane constraints:} parallel, perpendicular, fixed angular,
	coincidence, distance.
\end{itemize}

We develop the pattern of the {\bf rigidity matrix} and identify a
necessary combinatorial counting property called
{\em nested sparsity}, which is the counterpart of the well-known
Maxwell condition \cite{MaxwellEquil1864} for fixed length rigidity. We also show that
this condition is {\em not} sufficient. However, it can be used 
as a filter for finding candidate rigid components. Finally, we present
an efficient algorithm for nested sparsity, based on pebble 
game algorithms previously developed
for sparse graphs.

\smallskip\noindent{\bf Related work.}
Classical rigidity theory \cite{graver:servatius:rigidityBook:1993}
focuses on distance constraints between
points \cite{laman} or rigid bodies \cite{tay:rigidityMultigraphs-I:1984,white:whiteley:algebraicGeometryFrameworks:1987}. 
{\em Direction}
constraints (where 2 points are required to define a fixed
direction, with respect to a global coordinate system) are
well-understood and arise from parallel redrawing applications
 \cite{whiteley:Hypergraph:1989}. 
Motivated by CAD systems, Servatius and Whiteley present a
characterization, which can be viewed as a generalized Laman
counting property, for 2D systems with both length and direction
constraints  \cite{ServatiusWhiteleyCAD1999}.

Work on angular constraints has also focused on combinatorial characterization results.
Zhou and Sitharam \cite{zhouThesis06} characterize a large
class of 2D angular constraint systems along with a set of
combinatorial construction rules that maintain generic independence.
Saliola and Whiteley \cite{SaliolaWhiteleyCAD2004} prove that, even in the
plane, the complexity of determining the independence of a set of circle intersection angles
is the same as that of generic bar-and-joint rigidity in 3D.
A full characterization for angular constraints of the nature that
 appear in this paper is further described 
 in \cite{stjohn:streinu:angularRigidityCCCG:2009,lee:PhDThesis:2008}.

Combinatorial {\em sparsity} conditions 
 are intimately tied with rigidity
theory, appearing often as necessary conditions 
(as for bar-and-joint rigidity) and sometimes even as
complete characterizations (as for
2D bar-and-joint and body-and-bar frameworks in arbitrary dimension)
\cite{whiteley:Matroids:1996,laman,tay:rigidityMultigraphs-I:1984}. 
Pebble game algorithms have been developed for 
solving sparsity problems \cite{streinu:lee:pebbleGames:2008,streinu:theran:hypergraphsPebbleGames:2009,graded}. 
These algorithms do not apply, however, to
the so-called $(3,6)$-counting conditions known to be a necessary, but not 
sufficient, condition for 3D bar-and-joint rigidity. In fact, no efficient
algorithm is known for these counts. 

Related work on the constraints studied in this paper 
has appeared in the CAD research community, usually
within the context of decomposition approaches; a survey may be found
in \cite{sitharamSurvey,decompSurvey}. 
In this setting, a geometric constraint
system (GCS) is formulated as an algebraic system of equations. Due to the complexity
of solving such a system, it is traditionally decomposed into structured
sub-systems that can be solved and later recombined to obtain a solution
to the original GCS. In the process of decomposition, approximate notions of
combinatorial rigidity have been used \cite{sitharamDR1,sitharamDR2}.

Results in the CAD literature have observed that
angular constraints exhibit special behavior. 
For the so-called generalized Stewart
platform, \cite{GaoLeiLiaoZhang2005} gives explicit equations that highlight
this distinction.
Gao et al. \cite{gao:lin:zhang:CTreeDecomposition:2006} present a method for analyzing
2D and 3D systems with a restricted set of coincidence, angular and distance constraints.
Both \cite{GaoLeiLiaoZhang2005}
and \cite{gao:lin:zhang:CTreeDecomposition:2006}
treat angular constraints separately, implicitly using  
natural necessary counting conditions to do so. We consider analogous systems 
from the rigidity theory perspective, expressing them infinitesimally
using Grassmann-Cayley algebra;
the shape of the rigidity matrix described in Section \ref{sec:infTheory}
explicitly reveals the distinct treatment of angular constraints.
Grassmann-Cayley, Clifford algebras and geometric algebras often
appears in the context of CAD or geometric theorem proving; see, e.g.,
\cite{LiWu03,serre09}.
Recent work of \cite{riviere09} expresses constraints from a 
similar perspective when providing a foundation for software 
to build a GCS.

Incidence constraints have been studied previously in 
connection with Geometric Theorem Proving \cite{LiWu03,MichelucciS06} 
for projective incidence theorems.
Sitharam et al. \cite{DBLP:journals/ijcga/Sitharam06,sitharam2010JSC,sitharam2010IJCGA} 
formalize the question of obtaining a well-formed
and optimal system of algebraic equations to resolve a collection
 of incident rigid bodies. 
\cite{DBLP:journals/ijcga/Sitharam06} studies ``well-formedness,'' a condition necessary  to avoid
dependent equations, and
a new, underlying matroid whose independent sets capture this.
A combinatorial measure of
algebraic complexity of the system of equations is described in \cite{sitharam2010JSC}, and another underlying
matroid is used to optimize this measure. In \cite{sitharam2010IJCGA}, it is shown how
to reconcile the independent sets of the prvious two matroids to obtain
an optimal, well-formed system.

\smallskip\noindent{\bf Structure.}
Section \ref{sec:prelims} gives a brief overview of the required mathematical
background. Section \ref{sec:infTheory} develops the foundations
for the infinitesimal rigidity theory, providing the basic building
blocks used for each new constraint. 
Each of the full set of constraints is then expressed using these building blocks in 
Section \ref{sec:allConstraints}, resulting in the complete derivation
of the rigidity matrix.
Section \ref{sec:combinatorics} identifies a new combinatorial
property resulting from the structure of the rigidity matrix; this
{\em nested sparsity} condition, while necessary, is shown
not to be sufficient with a counterexample. Section \ref{sec:algs}
presents algorithms for nested sparsity using pebble games as oracles. 
Finally, Section \ref{sec:conclusions} discusses extensions, applications
and future directions.

%% file: prelims.tex
\section{Preliminaries}
\label{sec:prelims}

Our results rely on the same mathematical background as 
 the work on body-and-bar rigidity by 
Tay \cite{tay:rigidityMultigraphs-I:1984} and 
White and Whiteley \cite{white:whiteley:algebraicGeometryFrameworks:1987}.
We use Grassmann-Cayley algebra, Pl\"{u}cker coordinates and instantaneous screw
theory (see, e.g., \cite{white94grassmanncayley,white97geometric} and \cite{selig:Robotics:1996}). 
For self-containment, we briefly introduce notation and basic concepts
from the Grassmann-Cayley algebra and its correspondence with instantaneous screws.

\subsection{Terminology and notation.}
\label{sec:termNot}
We restrict ourselves to dimension 3 in this paper; 
{\em 2-tensors} in the Grassmann-Cayley algebra (see, e.g., \cite{white94grassmanncayley,white97geometric})
are identified with vectors in $\Re^6$. 
The Grassmann-Cayley {\em join} operator is represented with $\vee$. 
The join $\vec p \vee \vec q$ of two vectors $\vec p, \vec q \in \Re^4$
is the collection of all 6 minors of the matrix $M$ obtained with $\vec p$
and $\vec q$ as its rows. 
We fix a convention at this point to 
order the minors in a 6-vector as
$(|M_{14}|, |M_{24}|, |M_{34}|, |M_{23}|,-|M_{13}|, |M_{12}|)$\footnote{We remark that other papers
(e.g., \cite{tay:rigidityMultigraphs-I:1984} 
and \cite{white:whiteley:algebraicGeometryFrameworks:1987}), use a different convention
by fixing the order as $(|M_{12}|, |M_{13}|, |M_{14}|, |M_{23}|, |M_{24}|, |M_{34}|)$}.
The dot product of two vectors $u$ and $v$ is denoted $\langle u, v\rangle$.
The {\em star operator $^*$} 
swaps the first and last 3 coordinates of a 6-vector.
If $\vec p \in \Re^3$ and $c \in \Re$, we denote by $(\vec p: c)$ 
the vector of length 4 obtained by appending $c$ to $\vec p$.

\smallskip\noindent{\bf Rigid body motions.} 
The theory of screws was introduced by Ball \cite{BallTheoryScrews00}
as a way of expressing rigid body motion.
Rigid body transformations are associated with
elements of the special Euclidean group $\SE(3)$. 
By Chasles' Theorem from 1830 (see \cite{selig:Robotics:1996}), 
they can also be expressed as {\em screw motions} (see Figure \ref{fig.screw}). It follows
that every instantaneous rigid body motion can be expressed as an {\em instantaneous screw motion}
(see Figure \ref{fig.instantaneousScrew}); for further
details, we refer the reader to a standard text, e.g., page 24 of \cite{selig:Robotics:1996}.
\begin{figure}[h]
\centering \subfloat[Every rigid body motion can be expressed
as a screw motion that includes rotation and translation along
the screw axis.] {\label{fig.screw}
\begin{minipage}[b]{0.25\linewidth}
\centering
\includegraphics[scale=.25]{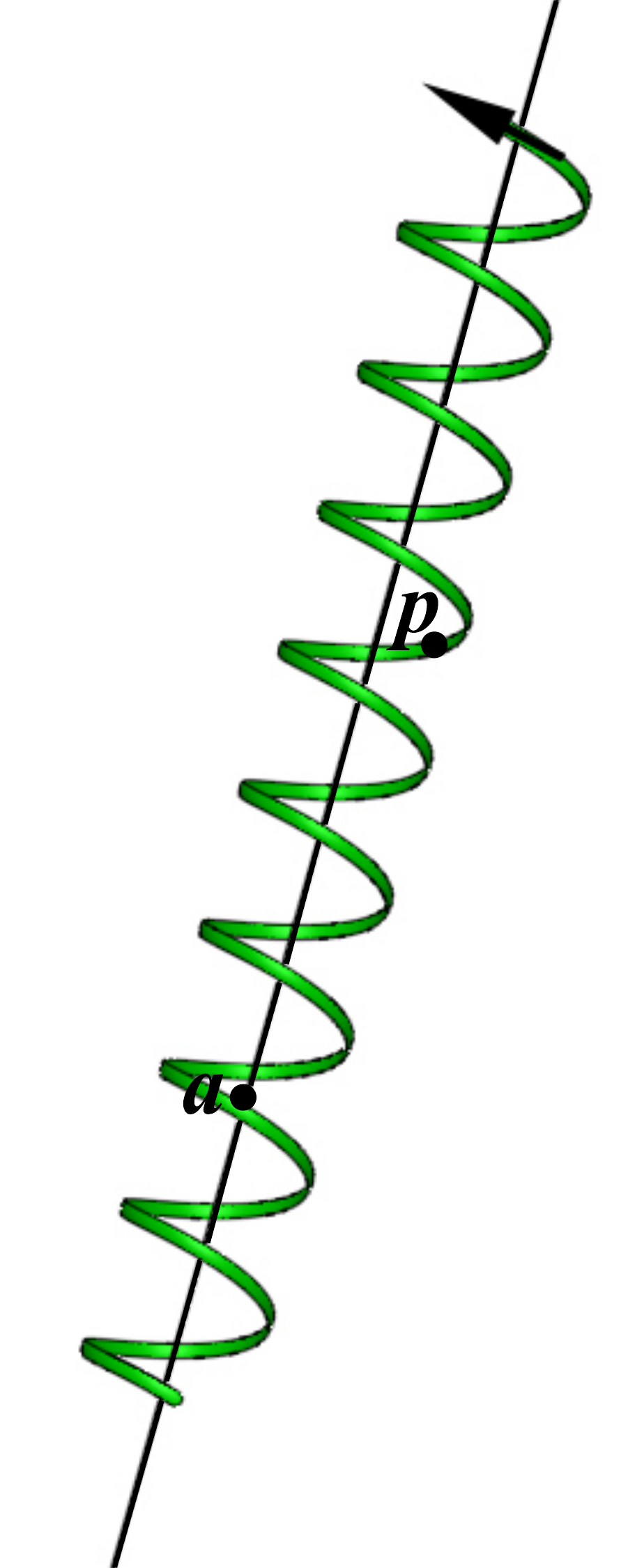}
\end{minipage}}%
\hspace{15mm}
\centering \subfloat[An instantaneous screw has both rotational
and translational components.]{\label{fig.instantaneousScrew}
\begin{minipage}[b]{0.25\linewidth}
\centering\includegraphics[scale=.25]{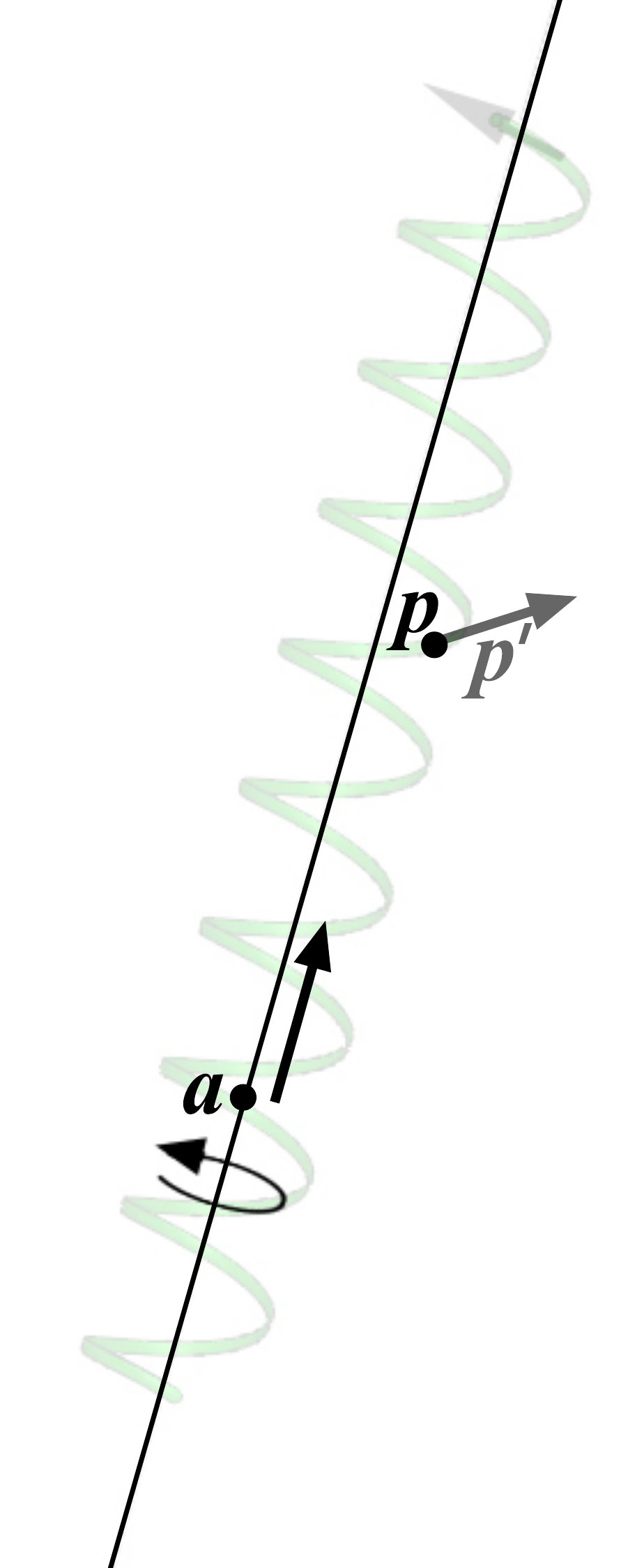}
\end{minipage}}
\caption{A screw motion and its associated instantaneous screw motion. }
\label{fig.screwMotion}
\end{figure}
Both screw motions and instantaneous screw motions are defined with respect to a 
{\em screw axis} along with a rotation about the axis and a translation along it.

In this paper, we are concerned only with instantaneous screw motions, 
which, for brevity, will be referred to as {\em instantaneous screws}.
An instantaneous screw is represented by a 6-vector
$\vec s = (-\vc \omega,  \vec v)$, where $\vc \omega, \vec v \in \Re^3$;
the minus sign in front of $\vc \omega$ is a convenient, technical convention. 
The first component $\vc \omega$ encodes the angular velocity; as a vector,
$\vc \omega$ gives the direction of the screw axis, and its magnitude encodes
the angular speed. The translational velocity can be
computed from $\vc \omega$ and $\vec v$, but we skip the details as they are
not relevant for the rest of the paper. Note that the star operator applied to a screw
 $\vec s = (-\vc \omega, \vec v)$ gives $\vec s^* = (\vec v, -\vc \omega)$.
There is an exact correspondence between 2-tensors
and instantaneous screws. This correspondence is the key to describing
the rigidity matrix of the body-and-cad structures.

\subsection{Body-and-cad structures}
A body-and-cad structure in 3D is composed of $n$ bodies
interconnected by pairwise constraints. Each body $i$ is represented
by a frame of reference, specified by a transformation matrix $T_i$ from
the special Euclidean group $SE(3)$. Each body $i$ additionally has
a set of {\em geometric elements} (points, lines or planes) identified
as attachments for the constraints. 

\smallskip
\noindent{\bf Representation of geometric elements.} 
Each geometric element is rigidly
affixed to a body $i$ and is described with coordinates that
are local with respect to the frame of reference for body $i$.
For ease of analysis, we represent a plane in point-normal form
as the pair $(\vec p, \vec d)$, with $\vec p, \vec d \in \Re^3$, where
$\vec p$ is a point on the plane and $\vec d$ is the normal to the plane. We represent
a line in parametric form, given by the pair $(\vec p, \vec d)$, 
with $\vec p, \vec d \in \Re^3$, where
$\vec p$ is a point on the line and $\vec d$ is its direction.

\subsubsection{Cad graphs}
\label{section:cadgraphs}
We now introduce the {\em cad graph}, our main combinatorial object for body-and-cad rigidity.

To illustrate this concept, consider the following example, depicted in Figure \ref{fig.dice}.
Let $A$ and $B$ be two dice rigidly stacked with
the following constraints: 
(i) ({\bf plane-plane parallel})
 $A$'s Face 1 is parallel to $B$'s Face 1,
(ii) ({\bf plane-plane perpendicular})
 $A$'s Face 2 is perpendicular to $B$'s Face 3,
(iii) ({\bf line-plane distance})
The distance between $A$'s Line 12 (intersection of Faces 1 and 2)
and $B$'s Face 1 is 1, and
(iv) ({\bf point-point coincidence})
$A$'s Corner 236 (the point defined by Faces 2, 3 and 6) is
coincident to $B$'s Corner 123.
\begin{figure}[htb]
  \begin{center}
    \includegraphics[width=.25\linewidth]{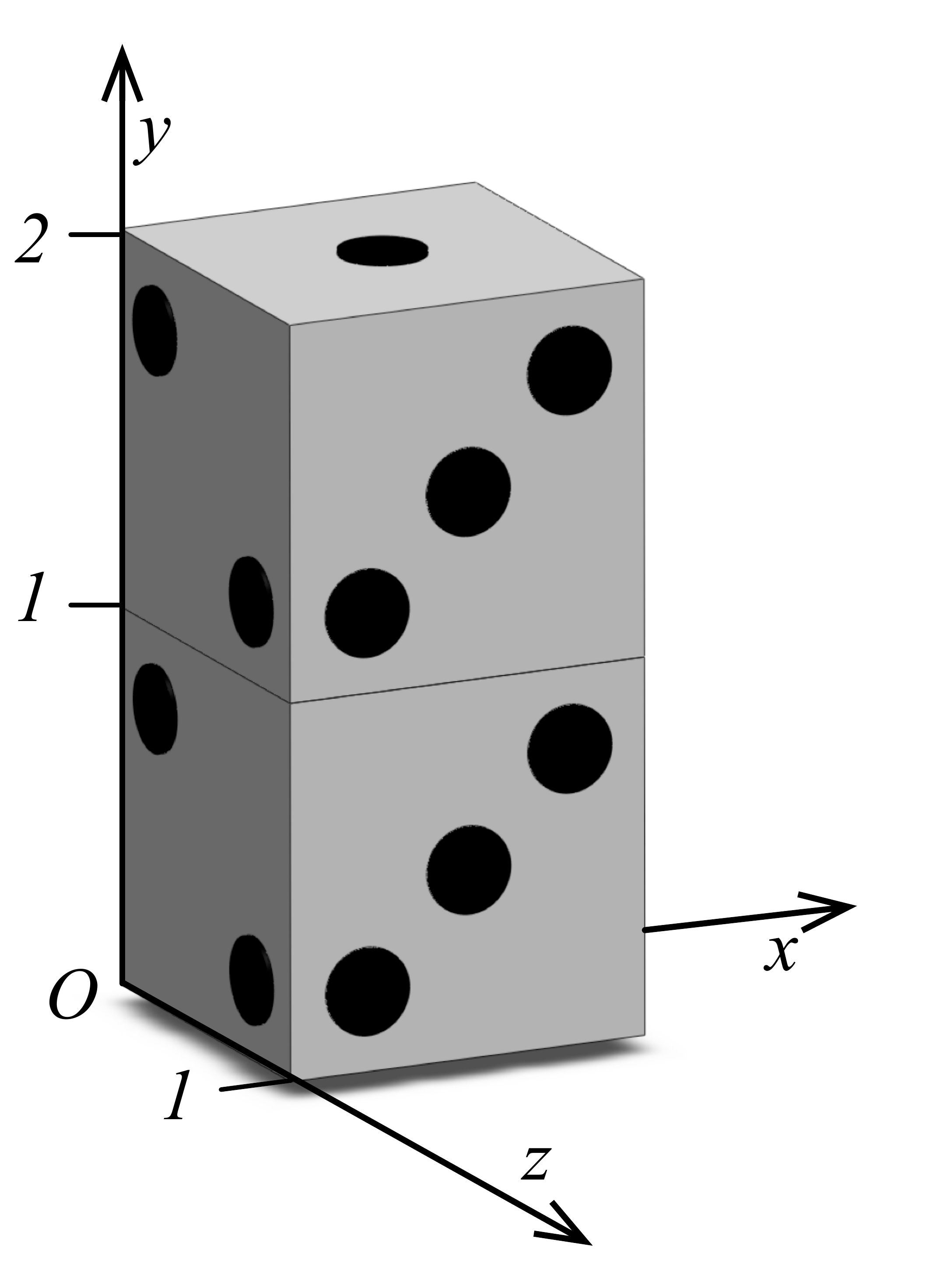}
  \end{center}
  \caption{Two dice rigidly stacked; die $A$ is above $B$. Faces
are labeled by the number of dots, and
face 6 lies at the bottom (opposite 1). The length of an edge is 1.}
  \label{fig.dice}
\end{figure}
These constraints are captured by a graph with two nodes connected by
annotated edges called the {\em cad graph}; see Figure \ref{fig.diceCadGraph}.
\begin{figure}[htb]
  \begin{center}
    \includegraphics[width=.75\linewidth]{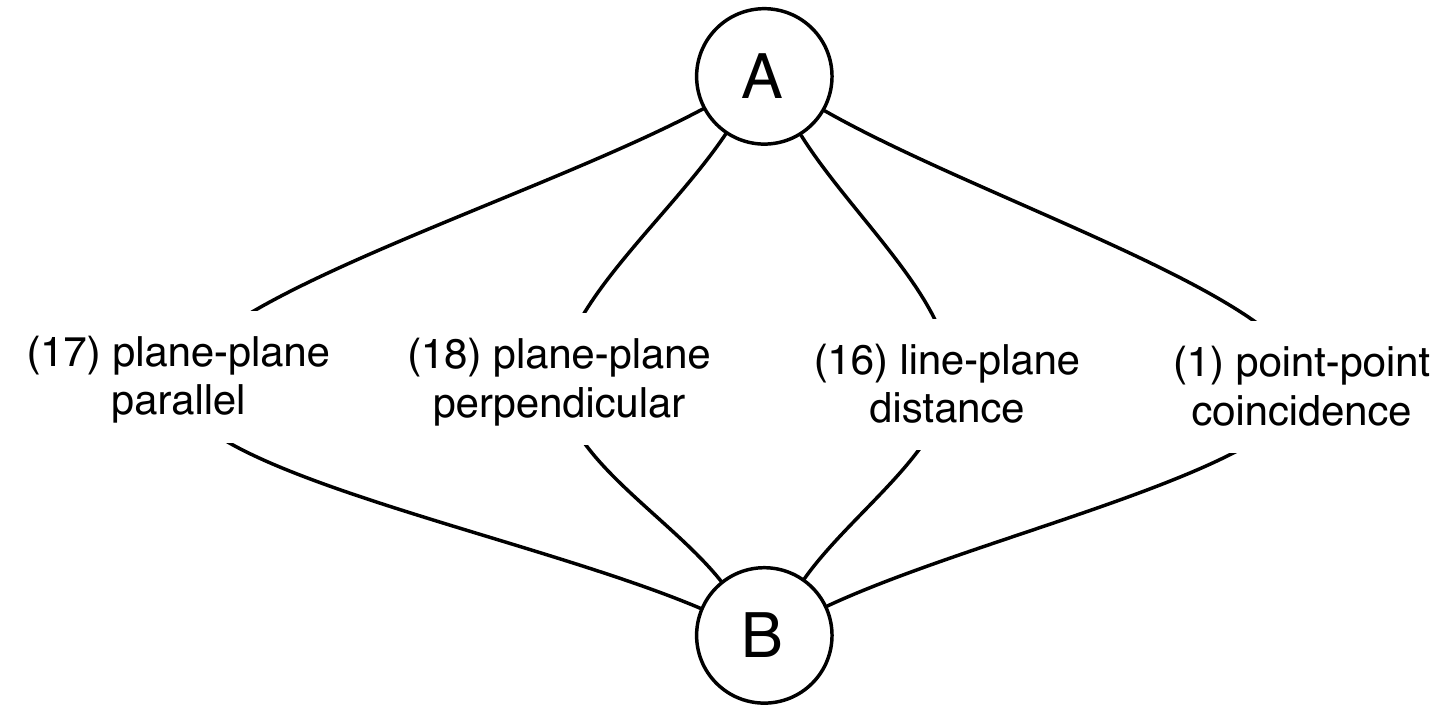}
  \end{center}
  \caption{The cad graph for the example depicted in Figure \ref{fig.dice}.}
  \label{fig.diceCadGraph}
\end{figure}

Formally, a {\em cad graph} $(G,c)$ is a multigraph $G = (V,E)$ together with
an edge coloring function $c: E \rightarrow C$, where $C = \{c_1, c_2, \ldots, c_{21}\}$ consists
of 21 colors corresponding to the full set of cad constraints:
\begin{enumerate}
	\item point-point coincidence
	\item point-point distance
	
	\item point-line coincidence
	\item point-line distance
	
	\item point-plane coincidence
	\item point-plane distance
	
	\item line-line parallel
	\item line-line perpendicular
	\item line-line fixed angular
	\item line-line coincidence
	\item line-line distance
	
	\item line-plane parallel
	\item line-plane perpendicular
	\item line-plane fixed angular
	\item line-plane coincidence
	\item line-plane distance
	
	\item plane-plane parallel
	\item plane-plane perpendicular
	\item plane-plane fixed angular
	\item plane-plane coincidence
	\item plane-plane distance
	
\end{enumerate}
The geometric meaning of these colors will be described in the next section;
the cad graph only captures the type of constraint imposed.

\subsubsection{Body-and-cad frameworks}
A {\em body-and-cad framework} $(G,c,L_1,\ldots, L_{21})$ is a cad graph $(G,c)$ along with 
a family of functions $L_1, \ldots, L_{21}$ describing the geometry of the structure, where
the function $L_i$ captures the constraints corresponding to edges with the $i$th color. 
Let $E_i = \{e \in E | c(e) = c_i\}$ be the set of $c_i$-colored edges.

For example, the function for plane-plane fixed angular constraints 
$L_{19}:E_{19} \rightarrow (\Re^3 \times \Re^3) \times (\Re^3 \times \Re^3) \times \Re$
maps an edge $e=ij$ to a triple $((\vec p_i, \vec d_i), (\vec p_j, \vec d_j), \alpha)$ so that
the planes $(\vec p_i, \vec d_i)$ and $(\vec p_j, \vec d_j)$ affixed to bodies $i$, respectively $j$, 
are constrained to have the angle $\alpha$ between them.
The function for point-point coincidence constraints $L_1:E_1 \rightarrow \Re^3 \times \Re^3$
maps an edge $e=ij$ to a pair of points $(\vec p_i, \vec p_j)$ affixed to bodies 
$i$, respectively $j$, that are constrained
to be coincident.
We will define the complete family of functions when analyzing them individually in Section \ref{sec:allConstraints}.

A {\em realization} $G({\cal T})$ of a body-and-cad framework $(G,c,L_1, \ldots,L_{21})$ 
assigns a tuple of frames ${\cal T} = (T_1, \ldots, T_n)$ for each vertex, satisfying the specified
constraints.
In this paper, we are {\em not} concerned with realization questions.
We will always assume that a body-and-cad framework is given by a concrete realization,
from which the family of functions $L_1,\ldots,L_{21}$ is computed.

\smallskip
\noindent{\bf Body-and-cad rigidity.}
Intuitively, a body-and-cad framework is {\em rigid} if the only 
motions 
respecting the constraints are the {\em trivial} 3D motions
(rotations and translations); otherwise, it is {\em flexible}.
We omit the technical definition, as it falls outside the scope of this paper.

\smallskip
\noindent{\bf Body-and-cad minimal rigidity.}
For classical distance constraints, the concept of {\em minimal rigidity} is defined 
as follows: a structure is minimally rigid if the removal of any constraint
results in a flexible structure. However, in our case, geometric constraints 
may correspond to more than one ``primitive'' constraint. 
Intuitively, a {\em primitive} constraint yields only one row in the
rigidity matrix (formally defined in Section \ref{sec:infTheory}), 
while the body-and-cad constraints may yield several rows.
In our setting, we define minimal rigidity as above, but referring
to the removal of primitive constraints only:  a rigid body-and-cad
structure is {\em minimally rigid} if the removal of any primitive constraint
results in a flexible structure.

We return to the example from Figure \ref{fig.dice} 
to illustrate the subtleties of this concept.
The structure depicted is {\em rigid}. We say the structure
is {\em overconstrained} since it remains rigid even after the removal of
constraint (iii).
The resulting structure is now minimally rigid.
As we will see in Section \ref{sec:infTheory}, constraints (i), (ii) and
(iv) correspond to 6 primitive constraints. Thus, the removal of any
primitive constraint results in a flexible structure.

Now consider stacking
the dice with the following two
constraints: 
(i) ({\bf line-line coincidence}) 
$A$'s Line 26 is coincident to $B$'s Line 12 and
(ii) ({\bf line-line coincidence})
$A$'s Line 36 is coincident to $B$'s Line 13.
This structure is still rigid. While it
becomes flexible after the removal
of either constraint (i) or (ii), it is {\em not} minimally rigid.
As we will see in Section \ref{sec:infTheory}, a line-line coincidence
constraint corresponds to 4 primitive constraints. Thus, this
structure has 8 primitive constraints and is overconstrained.
To give some intuition, note that
a structure composed of 2 rigid bodies has 12 degrees of freedom. Of these, 6 
are trivial, so we may fix body $A$ to factor them out. 
Now consider constraint (i); 
the structure is left with 2 degrees of freedom, as $B$ may slide
along the line and rotate about it. This line-line
coincidence constraint is ``eliminating'' 4 degrees of freedom,
formalized by the 4 rows of the rigidity matrix
developed in Section \ref{sec:infTheory} for the line-line coincidence constraint.

\smallskip
\noindent{\bf Body-and-cad infinitesimal rigidity.}
Infinitesimal rigidity is the linearized version 
of rigidity and is the only type we study in this paper.
Let
 $\vec s =  (\vec s_1, \ldots, \vec s_n) \in (\Re^6)^n$
assign an instantaneous screw $\vec s_i$ to each body $i$ and
let $\vec s^* = (\vec s_1^*, \ldots, \vec s_n^*)$.
The vector $\vec s$ is an {\em infinitesimal motion} of a body-and-cad structure if
it infinitesimally respects the constraints. This can be expressed with the
help of the {\em rigidity matrix}, fully described in 
Section \ref{sec:infTheory}.
An infinitesimal motion is
a vector in the kernel of the rigidity matrix.
The kernel always contains the {\em trivial infinitesimal motions}, defined
as those $\vec s$ with $\vec s_i = \vec s_j$ for all $i$ and $j$.

A body-and-cad framework is {\em infinitesimally rigid} if the only infinitesimal
motions are trivial; otherwise, it is {\em infinitesimally flexible}.

\smallskip
\noindent{\bf Remarks.}
To develop the rigidity theory for a new model, three steps must be accomplished. 
\begin{enumerate}
\item 
{\em Algebraic theory.} 
Formulate the rigidity concept in algebraic terms,
resulting in an algebraic variety. 
\item 
{\em Infinitesimal theory.} 
Analyze the local behavior at some point on the
algebraic variety. This reduces to the study of a {\em rigidity matrix}.
\item 
{\em Combinatorial rigidity.} Seek a combinatorial characterization 
of minimal rigidity in terms of properties of an underlying graph structure. This is usually
derived from properties of the rigidity matrix at a {\em generic} point on the algebraic variety.
\end{enumerate}
In this paper, we 
directly formulate the
infinitesimal rigidity theory for body-and-cad structures
and identify combinatorial properties for the generic case.

To summarize, a {\em cad graph} is an edge-colored multigraph that captures the 
body-and-cad combinatorics, and a {\em body-and-cad framework} captures 
the geometry of the structure.
As we develop the analysis of these concepts, an 
additional combinatorial object called the {\em primitive cad graph} will
be associated to the cad graph. This is a multigraph with red or black
edges, which captures certain combinatorial properties of infinitesimal
body-and-cad rigidity.

%% file: bodyCad.tex
\section{Foundations of infinitesimal theory}
\label{sec:infTheory}
The example from the previous section exposes some of the subtleties
encountered with body-and-cad constraints that are not found
when considering classical distance constraints.
We introduce two new concepts to simplify the analysis: {\em primitive angular}
and {\em blind} constraints. We then define, as building blocks, 4 {\em basic}
angular and blind constraints and develop their infinitesimal theory.
All 21 body-and-cad constraints can be studied using
these {\bf building blocks}, leading to the body-and-cad rigidity matrix.

\subsection{Primitive constraints}
A {\em primitive}
constraint is one that may restrict at most one degree of freedom.
For example, a {\bf point-point distance} (bar) constraint is a primitive constraint,
while the {\bf line-line coincidence} constraint from the example in 
the preceding section is not.
We classify primitive constraints into two types:
{\em angular} and {\em blind}; 
as the theory is developed, it will become more clear why these classifications
are appropriate, as they correspond to constraints demonstrating different algebraic
behaviors.
 
A rigid body in 3D has 6 degrees of freedom,
3 of which are rotational and 3 of which are translational. A
 {\em primitive angular}
constraint may restrict only a rotational degree of freedom, whereas a 
{\em primitive blind}
constraint may restrict either a rotational or a translational degree of freedom.
For instance, a {\bf line-line perpendicular} constraint is a primitive
angular constraint as it may restrict at most one rotational degree of freedom.
A {\bf point-point distance} (bar) constraint is a primitive blind constraint
as it may restrict at most one rotational or translational degree of freedom.
We will associate a set of primitive angular and a set of primitive blind
constraints with each body-and-cad constraint.

\subsection{Rigidity matrix}
The rigidity matrix $R$ for a body-and-cad structure
has 6 columns for each body $i$, corresponding to 
the components of the instantaneous screw $\vec s_i$, as was
done for the original body-and-bar rigidity matrix\footnote{The starred version $\vec s_i^*$ 
(see Section \ref{sec:termNot}) will be used to conveniently order
the columns of the rigidity matrix.}.
There
is a row for each primitive constraint associated to the
original body-and-cad structure. A primitive angular constraint results in 
a row containing
zero entries in the first 3 columns for each body, while
a primitive blind constraint may 
have non-zero entries in any of the 6 columns for each body.
In the schematic below, gray cells indicate potentially non-zero entries, and
red cells highlight the zero entries for angular constraints.

Since the trivial motions corresponding
to the 3D rigid motions are necessarily
in the kernel of $R$, the maximum rank of $R$ is $6n-6$. By definition,
a structure is {\em infinitesimally rigid} if its 
rigidity matrix has rank exactly $6n-6$.

\begin{center}
\begin{tabular}{ccccccccc}
 & \multicolumn{2}{c}{$\vec s_1^*$} & & \multicolumn{2}{c}{$\vec s_i^*$} &  &
\multicolumn{2}{c}{$\vec s_n^*$}\\
\cline{2-3} \cline{5-6} \cline{8-9}
 & \multicolumn{1}{c}{$\vec v_1$} &\multicolumn{1}{c}{-$\vc \omega_1$} &  $\cdots$ 
 & \multicolumn{1}{c}{$\vec v_i$} &\multicolumn{1}{c}{-$\vc \omega_i$} & $\cdots$ 
 & \multicolumn{1}{c}{$\vec v_n$} &\multicolumn{1}{c}{-$\vc \omega_n$} \\
\cline{2-9}
	& \multicolumn{1}{|c|}{\cellcolor{red}$\vec 0$} & 
	\multicolumn{1}{c|}{\cellcolor{lightgray}} &  $\cdots$ 
	& \multicolumn{1}{|c|}{\cellcolor{red}$\vec 0$} & 
	\multicolumn{1}{c|}{\cellcolor{lightgray}} & $\cdots$ 
	& \multicolumn{1}{|c|}{\cellcolor{red}$\vec 0$} & 
	\multicolumn{1}{c|}{\cellcolor{lightgray}} \\
Angular
	& \multicolumn{1}{|c|}{\cellcolor{red}$\vdots$} & 
	\multicolumn{1}{c|}{\cellcolor{lightgray}$\vdots$} &  $\cdots$ 
	& \multicolumn{1}{|c|}{\cellcolor{red}$\vdots$} & 
	\multicolumn{1}{c|}{\cellcolor{lightgray}$\vdots$} & $\cdots$ 
	& \multicolumn{1}{|c|}{\cellcolor{red}$\vdots$} & 
	\multicolumn{1}{c|}{\cellcolor{lightgray}$\vdots$} \\
constraints  
	& \multicolumn{1}{|c|}{\cellcolor{red}$\vec 0$} & 
	\multicolumn{1}{c|}{\cellcolor{lightgray}} &  $\cdots$ 
	& \multicolumn{1}{|c|}{\cellcolor{red}$\vec 0$} & 
	\multicolumn{1}{c|}{\cellcolor{lightgray}} & $\cdots$ 
	& \multicolumn{1}{|c|}{\cellcolor{red}$\vec 0$} & 
	\multicolumn{1}{c|}{\cellcolor{lightgray}} \\
	& \multicolumn{1}{|c|}{\cellcolor{red}$\vdots$} & 
	\multicolumn{1}{c|}{\cellcolor{lightgray}$\vdots$} &  $\cdots$ 
	& \multicolumn{1}{|c|}{\cellcolor{red}$\vdots$} & 
	\multicolumn{1}{c|}{\cellcolor{lightgray}$\vdots$} & $\cdots$ 
	& \multicolumn{1}{|c|}{\cellcolor{red}$\vdots$} & 
	\multicolumn{1}{c|}{\cellcolor{lightgray}$\vdots$} \\
	& \multicolumn{1}{|c|}{\cellcolor{red}$\vec 0$} & 
	\multicolumn{1}{c|}{\cellcolor{lightgray}} &  $\cdots$ 
	& \multicolumn{1}{|c|}{\cellcolor{red}$\vec 0$} & 
	\multicolumn{1}{c|}{\cellcolor{lightgray}} & $\cdots$ 
	& \multicolumn{1}{|c|}{\cellcolor{red}$\vec 0$} & 
	\multicolumn{1}{c|}{\cellcolor{lightgray}} \\
	
\cline{2-9}

	& \multicolumn{1}{|c|}{\cellcolor{lightgray}} &
	 	\multicolumn{1}{c|}{\cellcolor{lightgray}} & 
	& \multicolumn{1}{|c|}{\cellcolor{lightgray}} & 
		\multicolumn{1}{c|}{\cellcolor{lightgray}} & 
	& \multicolumn{1}{|c|}{\cellcolor{lightgray}} & 
		\multicolumn{1}{c|}{\cellcolor{lightgray}} \\
		& \multicolumn{1}{|c|}{\cellcolor{lightgray}} &
		 	\multicolumn{1}{c|}{\cellcolor{lightgray}} & 
		& \multicolumn{1}{|c|}{\cellcolor{lightgray}} & 
			\multicolumn{1}{c|}{\cellcolor{lightgray}} & 
		& \multicolumn{1}{|c|}{\cellcolor{lightgray}} & 
			\multicolumn{1}{c|}{\cellcolor{lightgray}} \\
Blind  
	& \multicolumn{1}{|c|}{\cellcolor{lightgray}$\vdots$} &
	 	\multicolumn{1}{c|}{\cellcolor{lightgray}$\vdots$} & 
	$\cdots$ 
	& \multicolumn{1}{|c|}{\cellcolor{lightgray}$\vdots$} & 
		\multicolumn{1}{c|}{\cellcolor{lightgray}$\vdots$} &
		$\cdots$ 
	& \multicolumn{1}{|c|}{\cellcolor{lightgray}$\vdots$} & 
		\multicolumn{1}{c|}{\cellcolor{lightgray}$\vdots$} \\
constraints
		& \multicolumn{1}{|c|}{\cellcolor{lightgray}} &
		 	\multicolumn{1}{c|}{\cellcolor{lightgray}} &  
		& \multicolumn{1}{|c|}{\cellcolor{lightgray}} & 
			\multicolumn{1}{c|}{\cellcolor{lightgray}} & 
		& \multicolumn{1}{|c|}{\cellcolor{lightgray}} & 
			\multicolumn{1}{c|}{\cellcolor{lightgray}} \\

		& \multicolumn{1}{|c|}{\cellcolor{lightgray}} &
		 	\multicolumn{1}{c|}{\cellcolor{lightgray}} &  
		& \multicolumn{1}{|c|}{\cellcolor{lightgray}} & 
			\multicolumn{1}{c|}{\cellcolor{lightgray}} &  
		& \multicolumn{1}{|c|}{\cellcolor{lightgray}} & 
			\multicolumn{1}{c|}{\cellcolor{lightgray}} \\

\cline{2-9}
\end{tabular}
\end{center}

\subsection{Building blocks}
\label{sec:buildingBlocks}
We now define 4 very specific {\bf basic} angular and blind constraints 
(2 of each) and develop the infinitesimal theory for them. 
Everything in
Section \ref{sec:allConstraints}
is derived from 
these basic {\bf building blocks}. The material
presented here is the most technical part of our paper. 

\noindent{\bf Angular building blocks} \\
All body-and-cad angular constraints can be 
reduced to the following basic constraints 
between pairs of lines: \\
\noindent {\em (i) basic line-line non-parallel fixed angular}, and \\ 
\noindent {\em (ii) basic line-line parallel}.

\smallskip

\noindent{\em (i) Basic line-line non-parallel fixed angular.} 
A line-line non-parallel angular constraint between bodies $i$ and $j$ is
defined by identifying a pair of non-parallel lines, each rigidly affixed to one body,
and fixing the angle between them.
Let $\vec d_i$ 
and $\vec d_j$ 
be the directions of the lines affixed to bodies
$i$ and $j$, respectively. Then the constraint is
infinitesimally maintained if the axis of the relative screw $\vec s_i - \vec s_j$
is in a direction lying in the plane determined by $\vec d_i$ and $\vec d_j$, i.e., 
\begin{equation*}
	\left\langle (\vc \omega_i - \vc \omega_j), \vec d_i \times\vec  d_j \right\rangle = 0
\end{equation*}
Since $-\vc \omega_i$ is composed of the last three coordinates of $\vec s_i^*$, this is equivalent to
\begin{equation}
	\label{eq:nonParallelAngleInf}
	\left\langle (\vec s_i^* - \vec s_j^*), ((0,0,0), \vec d_j \times\vec  d_i) \right\rangle = 0
\end{equation}
This corresponds to one row in the rigidity matrix:
\begin{center}
\begin{tabular}{ccccccc}
 & \multicolumn{2}{c}{$\vec s_i^*$} & & \multicolumn{2}{c}{$\vec s_j^*$} & \\
\cline{2-3} \cline{5-6}
$\cdots$ & $\vec v_i$& $-\vc\omega_i$ &
$\cdots$ & $\vec v_j$& $-\vc\omega_j$ & $\cdots$ \\
\hline
\multicolumn{1}{|c|}{$\rmfill$} & 
	\multicolumn{1}{c|}{\cellcolor{red}$\vec 0$} & 
	\multicolumn{1}{c|}{\cellcolor{lightgray}$\vec d_j \times \vec d_i$}&
	\multicolumn{1}{c|}{$\rmfill$} &
	\multicolumn{1}{c|}{\cellcolor{red}$\vec 0$} & 
	\multicolumn{1}{c|}{\cellcolor{lightgray}$\vec d_i \times \vec d_j$}&	
	\multicolumn{1}{c|}{$\rmfill$} \\
\hline
\end{tabular}
\end{center}

\noindent{\em (ii) Basic line-line parallel constraint.}  
A line-line parallel constraint between bodies $i$ and $j$ is
defined by identifying a pair of parallel lines, each rigidly affixed to one body,
and restricting them to remain parallel.
Let $\vec d = (a, b, c)$ be 
the direction of the parallel lines. Then the constraint
is infinitesimally maintained if the axis of the relative screw $\vec s_i - \vec s_j$ 
is in the same direction as $\vec d$, i.e., $(\vc \omega_i - \vc \omega_j) = \alpha \vec d$, 
for some scalar $\alpha$. This can be expressed by the following two
linear equations, where $\vc \omega = \vc \omega_i - \vc \omega_j = (\omega^x, \omega^y, \omega^z)$:
\begin{eqnarray*}
	\omega^x b - \omega^y a &=&0 \\
	\omega^y c - \omega^z b &=&0
\end{eqnarray*}
Since $-\vc \omega_i$ is composed of the last three coordinates of $\vec s_i^*$, these are equivalent to
\begin{eqnarray}
	\label{eq:parallelAngleInf1}
	\left\langle (\vec s_i^* - \vec s_j^*), (0,0,0,-b, a, 0)\right\rangle \\
	\label{eq:parallelAngleInf2}
	\left\langle (\vec s_i^* - \vec s_j^*), (0,0,0,0,-c,b)\right\rangle
\end{eqnarray}
and correspond to {\bf two rows} in the rigidity matrix:
\begin{center}
\begin{tabular}{ccccccc}
 & \multicolumn{2}{c}{$\vec s_i^*$} & & \multicolumn{2}{c}{$\vec s_j^*$} & \\
\cline{2-3} \cline{5-6}
$\cdots$ & $\vec v_i$& $-\vc\omega_i$ &
$\cdots$ & $\vec v_j$& $-\vc\omega_j$ & $\cdots$ \\
\hline
\multicolumn{1}{|c|}{$\rmfill$} & 
\multicolumn{1}{c|}{\cellcolor{red}$\vec 0$} & 
\multicolumn{1}{c|}{\cellcolor{lightgray}$(-b, a, 0)$}&
	\multicolumn{1}{c|}{$\rmfill$} &
	\multicolumn{1}{c|}{\cellcolor{red}$\vec 0$} & 
	\multicolumn{1}{c|}{\cellcolor{lightgray}$(b, -a, 0)$}&	
	\multicolumn{1}{c|}{$\rmfill$} \\
\hline
\multicolumn{1}{|c|}{$\rmfill$} & 
\multicolumn{1}{c|}{\cellcolor{red}$\vec 0$} & 
	\multicolumn{1}{c|}{\cellcolor{lightgray}$(0, -c, b)$}&
	\multicolumn{1}{c|}{$\rmfill$} &
	\multicolumn{1}{c|}{\cellcolor{red}$\vec 0$} & 
	\multicolumn{1}{c|}{\cellcolor{lightgray}$(0, c, -b)$}&	
	\multicolumn{1}{c|}{$\rmfill$} \\
\hline
\end{tabular}
\end{center}

\noindent{\bf Blind building blocks} \\
Let $\vec p$ be a point and $\vec p'$ its instantaneous velocity resulting from the
instantaneous screw 
$\vec s \in \Re^6$. Let $\vec c \in \Re^3$ be 
an arbitrary direction vector.
We either constrain the
velocity $\vec p'$ to be orthogonal or parallel to $\vec c$.
This yields the remaining basic constraints:\\
{\em (iii) basic blind orthogonality} (see Figure \ref{fig.pOrthog}), and \\
{\em (iv) basic blind parallel} (see Figure \ref{fig.pSameDir}). \\
Expressing
both of them becomes straightforward using the following fact:
\begin{fact}
	\label{fact:screwJoinPoints}
	Let $\vec s \in \Re^6$ be an instantaneous screw,
	$\vec p \in \Re^3$ a point and 	
	$\vec p'$ the velocity of $\vec p$ under the screw motion $\vec s$.
	Then 	 
	\begin{equation}
		\vec s \vee (\vec p:1) = (\vec p', -\left\langle \vec p,\vec p'\right\rangle)
	\end{equation}
and, for any $\vec q \in \Re^3$ and $q^w \in \Re$,
	\begin{equation}
		\label{eq:screwJoinPoints}
		\vec s \vee (\vec p:1) \vee (\vec q : q^w)=
		 \langle\vec p', \vec q\rangle - q^w \left\langle \vec p,\vec p'\right\rangle
	\end{equation}
\end{fact}
\input{appendix}

%
\begin{figure}[bth]
\centering \subfloat[{\em Orthogonality:} 
constrains the instantaneous velocity $\vec p'$ of the point $\vec p$ to
be orthogonal to a direction $\vec c$. Then $\vec p'$ 
must lie in the plane $(\vec p, \vec c)$.] {\label{fig.pOrthog}
\begin{minipage}[b]{0.45\linewidth}
\centering
\includegraphics[scale=.3]{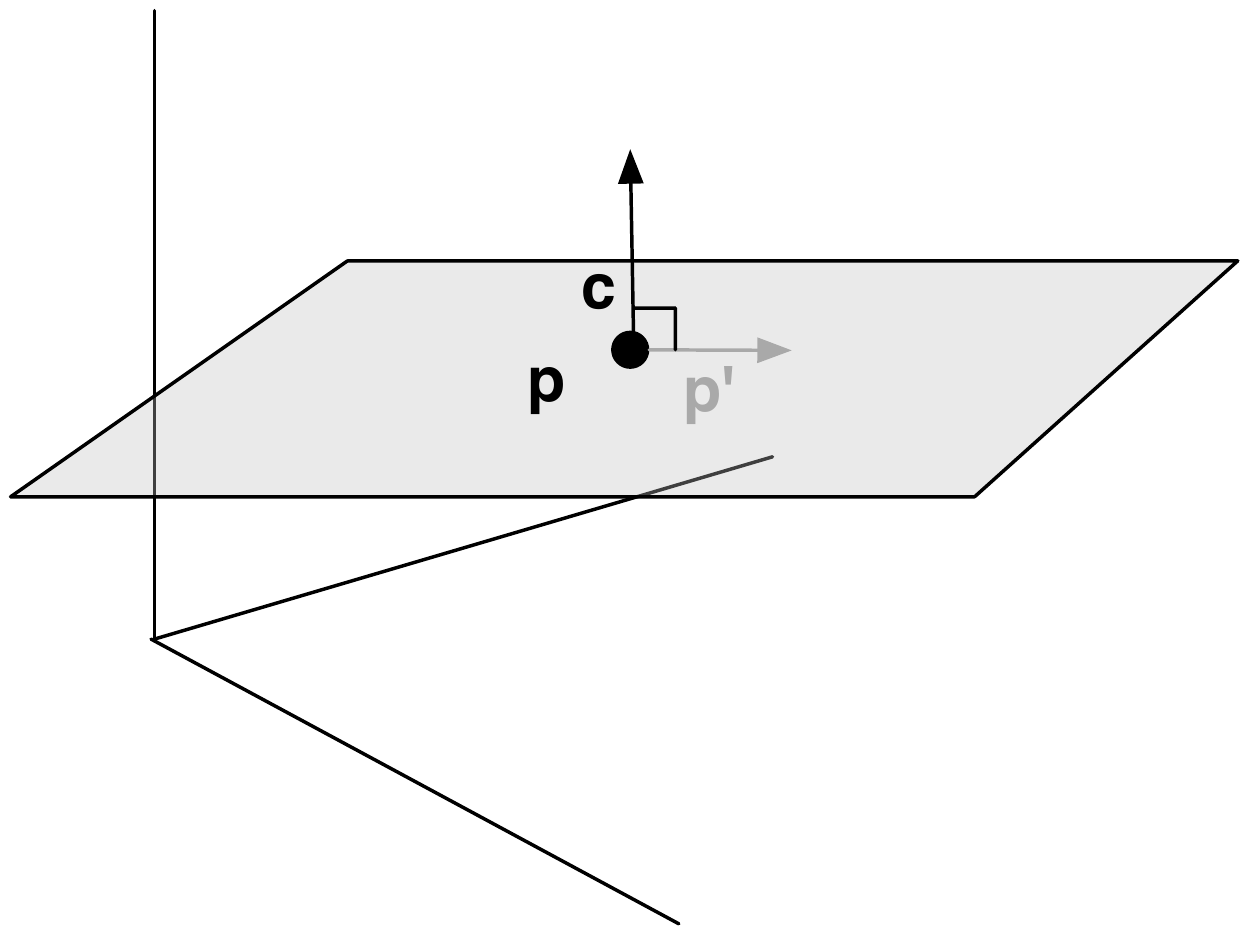}
\end{minipage}}%
\hspace{2mm}
\centering \subfloat[{\em Parallel:} constrains
the instantaneous velocity $\vec p'$ of a point $\vec p$ to
lie in the same direction as a vector $\vec c$.]{\label{fig.pSameDir}
\begin{minipage}[b]{0.45\linewidth}
\centering\includegraphics[scale=.3]{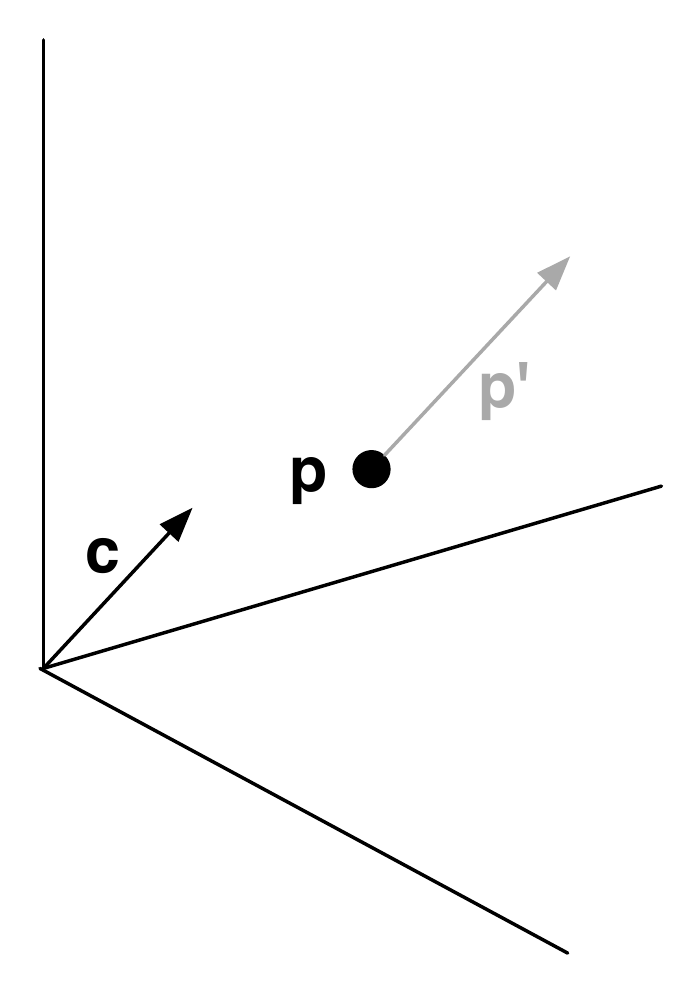}
\end{minipage}}
\caption{Basic blind geometric constraints. }
\label{fig.basicGeomCon}
\end{figure}

\noindent{\em (iii) Basic blind orthogonality constraint} \\
This constrains the instantaneous velocity $\vec p'$ of the point $\vec p$ to
be orthogonal to a direction $\vec c$.
To express this
(see Figure \ref{fig.pOrthog}), we simply substitute $\vec q = \vec c$ and
$q^w = 0$ into Equation \ref{eq:screwJoinPoints}. 
Then $\left\langle \vec p',\vec c\right\rangle = 0$ if and only if
\begin{equation*}
	\vec s \vee (\vec p:1) \vee (\vec c:0) = 0
\end{equation*}
if and only if
\begin{equation}
	\label{eq:velOrthogToVec}
	\left\langle \vec s^*, (\vec p:1) \vee (\vec c:0)\right\rangle = 0
\end{equation}

\noindent{\em (iv) Basic blind parallel constraint} \\
This constrains the instantaneous velocity $\vec p'$ of the point $\vec p$ to lie
in the same direction as a direction $\vec c$.
To express this
(see Figure \ref{fig.pSameDir}), we apply Equation \ref{eq:screwJoinPoints} twice by
substituting
$\vec q = (c^y, -c^x, 0)$ and $q^w = 0$ first, then $\vec q = (0, c^z, -c^y)$
and $q^w=0$.
We obtain that $\vec p' = \alpha \vec c$ for some $\alpha \in \Re$ if and only if
\begin{eqnarray*}
	\vec s \vee (\vec p:1) \vee (c^y, -c^x, 0, 0) &=&0\\
	\vec s \vee (\vec p:1) \vee (0, c^z, -c^y, 0) &=&0
\end{eqnarray*}
if and only if
\begin{eqnarray}
		\label{eq:velSameDirAsVec1}
		\left\langle \vec s^*, (\vec p:1) \vee (c^y, -c^x, 0, 0)\right\rangle &=&0\\
		\label{eq:velSameDirAsVec2}
		\left\langle \vec s^*, (\vec p:1) \vee (0, c^z, -c^y, 0)\right\rangle &=&0
\end{eqnarray}

%% file: appendix.tex
%

\begin{proof}
In the following, 	
superscripts $x, y, z, w$ denote the 
components of a vector in $\Re^4$.
The 
minor of a $3 \times 4$ matrix $A$ 
determined by columns $i$, $j$ and $k$ is denoted $|A_{ijk}|$.

If $\vec s$ is a decomposable 2-tensor (a 2-extensor), then
its components are the minors of a $2 \times 4$ matrix $M$; see, e.g.,
\cite{white94grassmanncayley,white97geometric} for a standard review of 
2-tensors in Grassmann-Cayley algebra.
%
Let $A$ be the $3 \times 4$ matrix obtained by appending $(\vec p:1)$
to the bottom of $M$. The join $\vec s \vee (\vec p : 1)$ is the collection of the four minors of
$A$. We fix the convention that 
$\vec s \vee (\vec p : 1) = (|A_{234}|, -|A_{134}|, |A_{124}|, -|A_{123}|)$.
Then 
$$\vec s \vee (\vec p:1) \vee (\vec q:q^w) = \left|\begin{array}{cccc}
\multicolumn{4}{c}{M} \\
p^x & p^y & p^z & 1 \\
q^x & q^y & q^z & q^w
\end{array}\right|$$
Performing a Laplace expansion along the 4th 
row of the matrix yields
$q^x |A_{234}| - q^y|A_{134}| + q^z|A_{124}| - q^w|A_{123}|
= q^x (\vec s \vee (\vec p:1))^x + q^y (\vec s \vee (\vec p:1))^y 
	+ q^z(\vec s \vee (\vec p:1))^z + q^w (\vec s \vee (\vec p:1))^w$.
	
Crapo and Whiteley \cite{CrapoWhiteley82} derived that
$\vec s \vee (\vec p:1) = (\vec p': -\left\langle \vec p,\vec p'\right\rangle)$.
Applying it, we obtain our desired result.
The derivation when $\vec s$ is indecomposable (the sum of two 2-extensors)
is a simple extension obtained by working with the two 2-extensors simultaneously. 
\end{proof}

%% file: allConstraints.tex
\section{Infinitesimal theory for body-and-cad constraints}
\label{sec:allConstraints}
We use the four basic building blocks just presented
to complete the development of the infinitesimal theory.
In this section, we present the rows of the
rigidity matrix associated with each of the 21 body-and-cad constraints.
In all figures, body $i$ is represented by the green tetrahedron and body
$j$ by the purple cube.
\input{angularConstraints.tex}
\input{blindConstraints.tex}

\input{example.tex}

\subsection{Summary of infinitesimal theory}
We have now completed the development of the infinitesimal theory 
for body-and-cad rigidity. 
Table \ref{table:constraintAssoc} summarizes the associations for each constraint 
to the number of primitive {\bf angular} and {\bf blind} constraints. 
As an example of how to read
the table, the 
last two columns (corresponding to {\bf plane}) of row 3 (corresponding to {\bf coincidence}    
under {\bf line}) indicate
that a {\bf line-plane
coincidence} constraint reduces to 1 angular 
and 1 blind primitive constraint.
In the next section, we identify a combinatorial property based on the
shape of the rigidity matrix.

\begin{table}[bht]
\begin{center}
	\begin{tabular}{|l||c|c|c|c|c|c|}
\hline  & \multicolumn{2}{c|}{\bf point} & \multicolumn{2}{c|}{\bf
line} &
 \multicolumn{2}{c|}{\bf plane} \\
\hline
 & angular & blind & angular & blind & angular & blind \\
\hline  \multicolumn{7}{|l|}{\bf point} \\
\hline  coincidence & 0 & 3 & 0 & 2 & 0 & 1 \\
 distance & 0 & 1 & 0 & 1 & 0 & 1 \\
\hline \multicolumn{7}{|l|}{\bf line} \\
\hline 
 coincidence &  &  & 2 & 2 & 1 & 1 \\
 distance &  &  & 0 & 1 & 1 & 1 \\
 parallel &  &  & 2 & 0 & 1 & 0 \\
 perpendicular &  &  & 1 & 0 & 2 & 0 \\
 fixed angular &  &  & 1 & 0 & 1 & 0 \\
\hline  \multicolumn{7}{|l|}{\bf plane} \\
\hline 
 coincidence & & & & & 2 & 1  \\
 distance & & & & & 2 & 1  \\
parallel & & & & & 2 & 0 \\
 perpendicular  & & & & & 1 & 0 \\
 fixed angular & & & & & 1 & 0\\
 \hline
\end{tabular}
\end{center}
\caption{Association of body-and-cad ({\em coincidence, angular, distance}) constraints 
with the number of blind and angular primitive constraints. } 
\label{table:constraintAssoc}
\end{table}

%% file: angularConstraints.tex
\subsection{Angular constraints}
Angular constraints may be parallel, perpendicular or arbitrary fixed angular
constraints; see Figures \ref{fig.lineLineAngle}--\ref{fig.planePlaneAngle} 
for depictions of line-line, line-plane, and plane-plane angular constraints.
\begin{figure}[h]
\centering \subfloat[Parallel.] {\label{fig.lineLineParallel}
\begin{minipage}[b]{0.3\linewidth}
\centering
\includegraphics[height=1.75in]{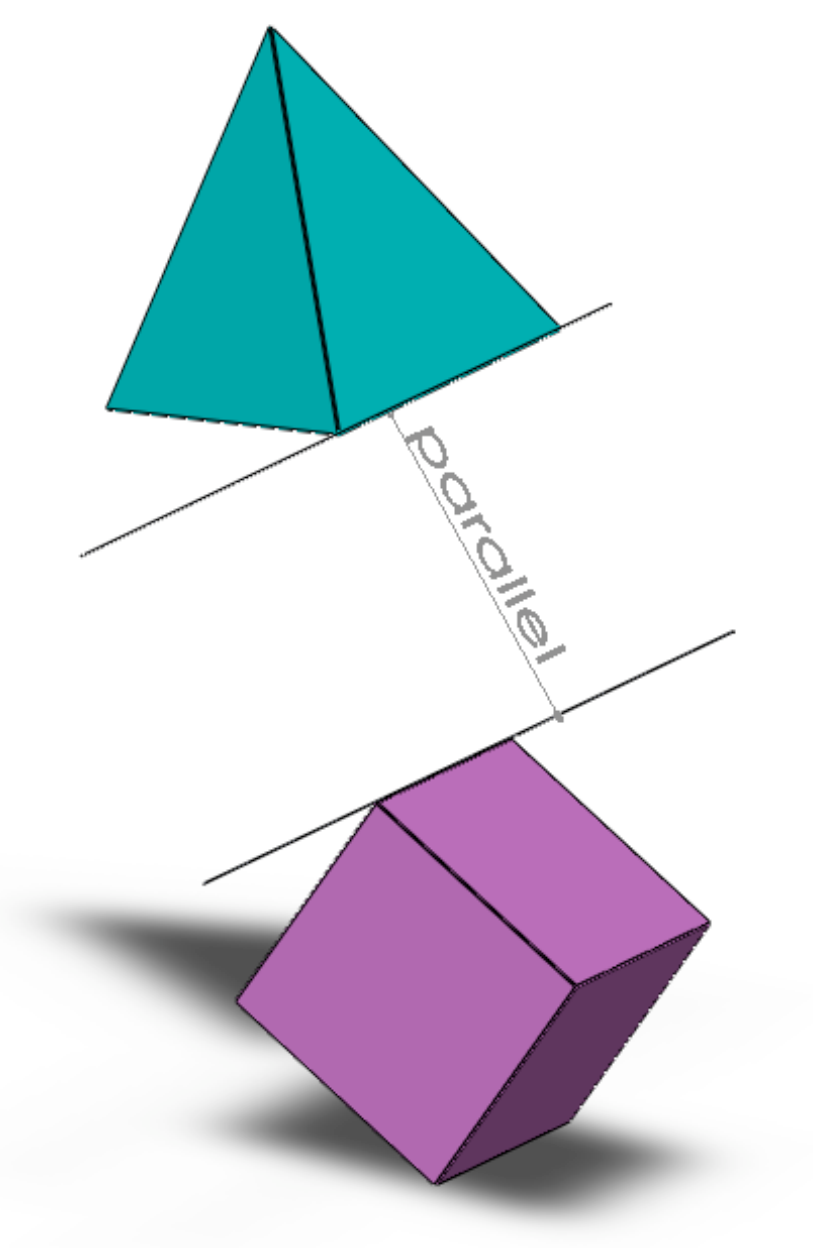}
\end{minipage}}%
\hspace{2mm}
\centering \subfloat[Perpendicular.]{\label{fig.lineLinePerp}
\begin{minipage}[b]{0.3\linewidth}
	\centering\includegraphics[height=1.25in]{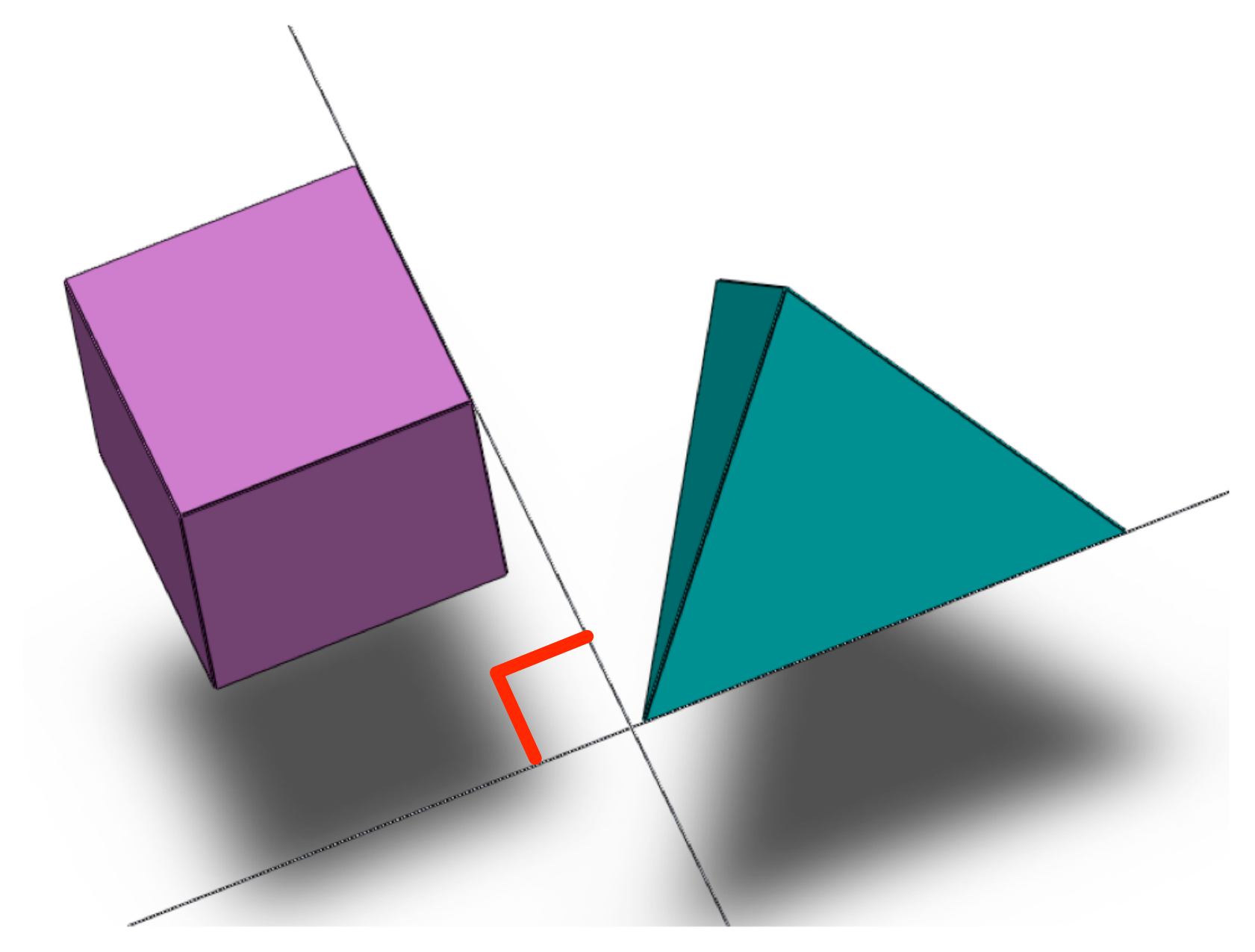}
\end{minipage}}
\hspace{2mm}
\centering \subfloat[Arbitrary fixed angular.]{\label{fig.lineLineArbAngle}
\begin{minipage}[b]{0.3\linewidth}
	\centering\includegraphics[height=1.75in]{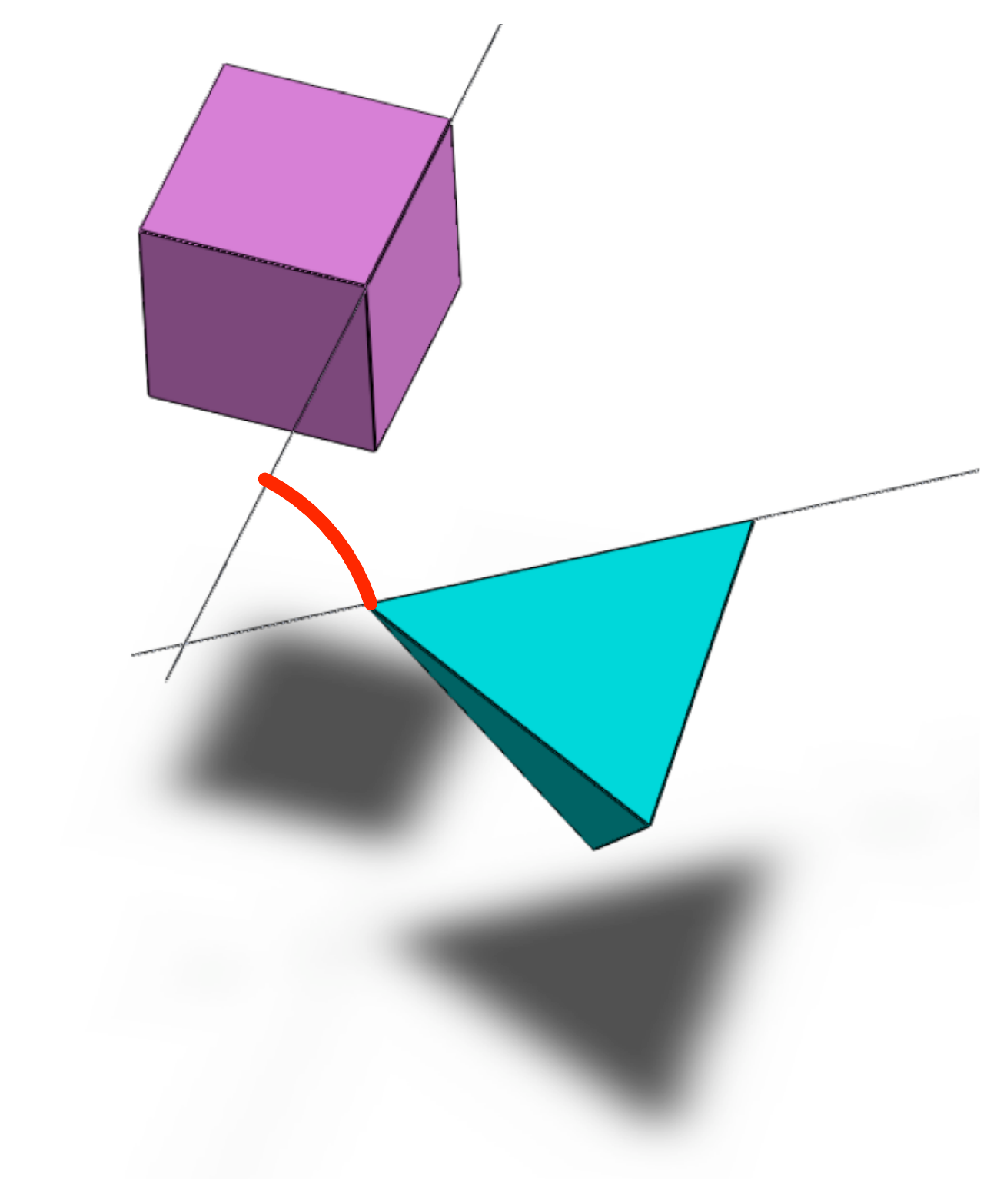}
\end{minipage}}
\caption{Line-line angular constraints. }
\label{fig.lineLineAngle}
\end{figure}
\begin{figure}[h]
\centering \subfloat[Parallel.] {\label{fig.linePlaneParallel}
\begin{minipage}[b]{0.3\linewidth}
\centering
\includegraphics[width=1.2\linewidth]{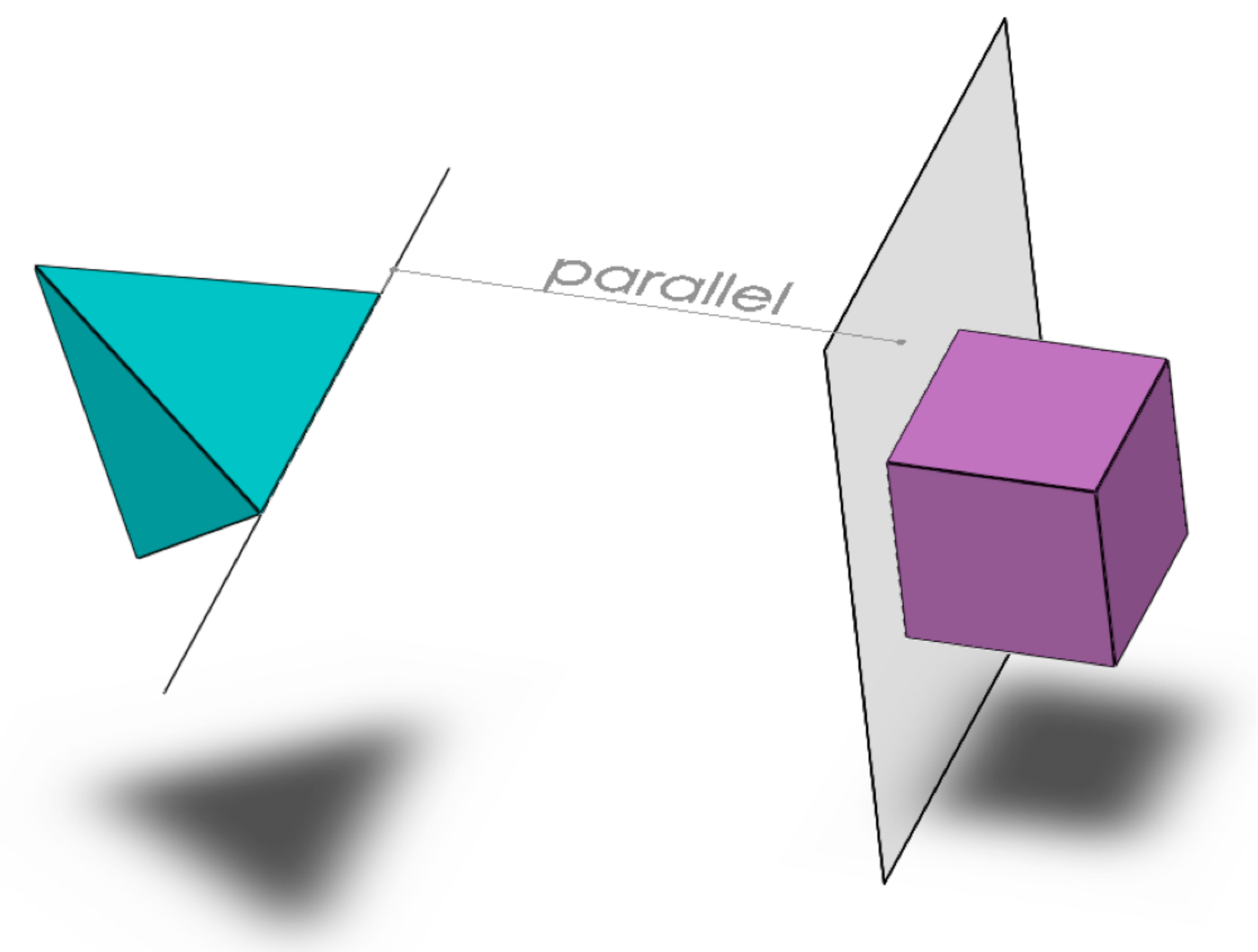}
\end{minipage}}%
\hspace{2mm}
\centering \subfloat[Perpendicular.]{\label{fig.linePlanePerp}
\begin{minipage}[b]{0.3\linewidth}
\centering\includegraphics[width=\linewidth]{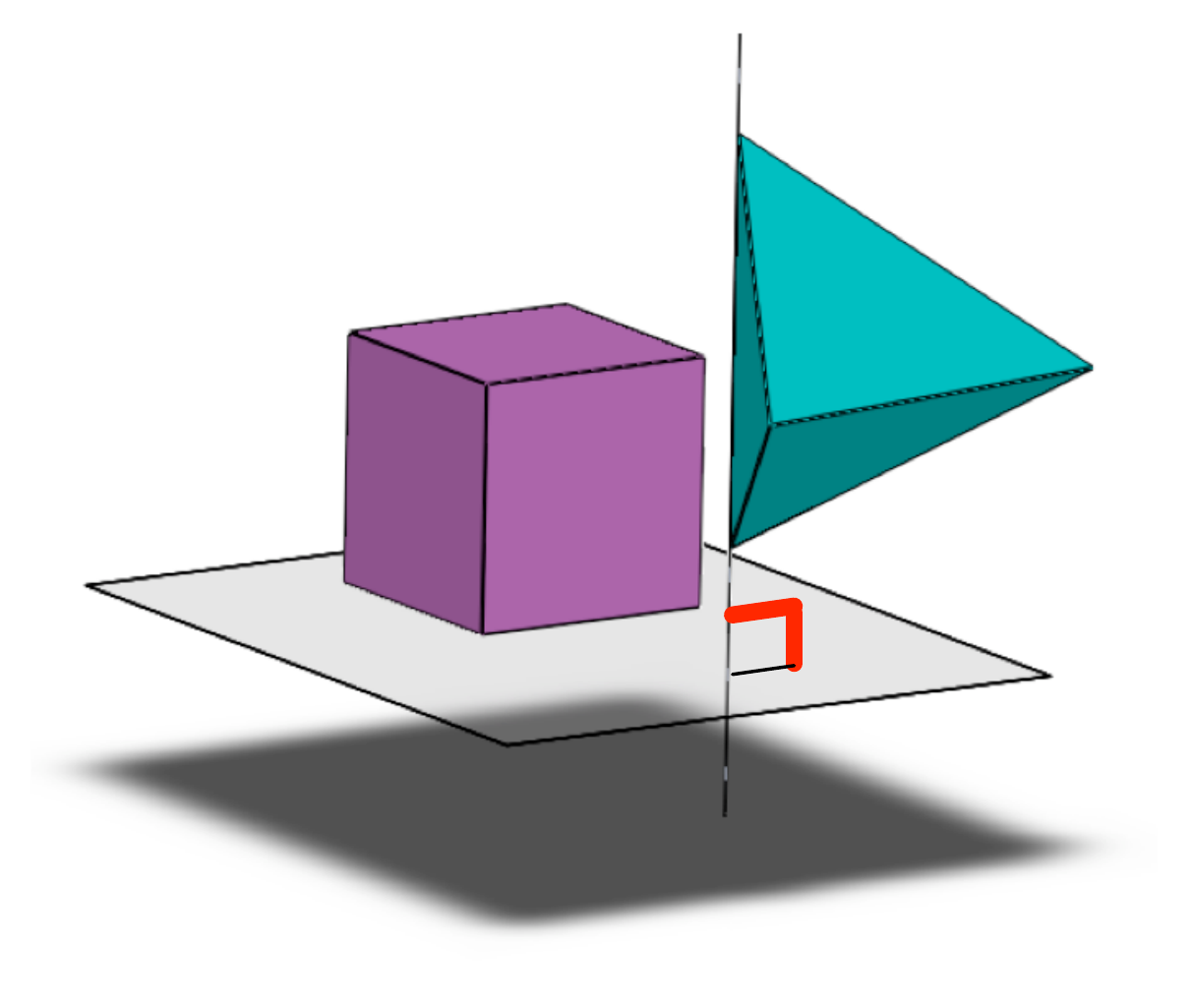}
\end{minipage}}
\centering \subfloat[Arbitrary fixed angular.] {\label{fig.linePlaneArbAngle}
\begin{minipage}[b]{0.3\linewidth}
\centering
\includegraphics[width=\linewidth]{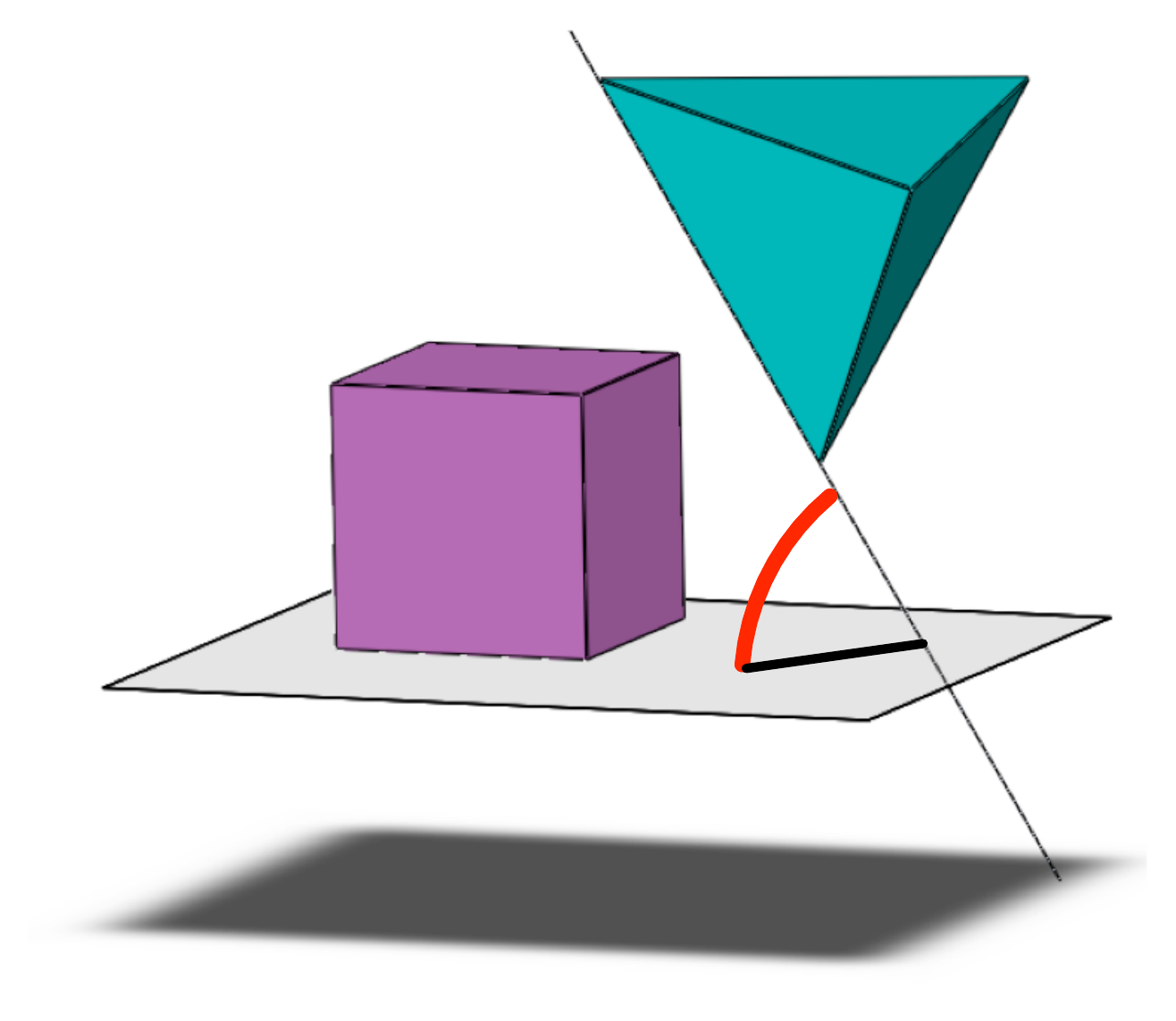}
\end{minipage}}%
\hspace{2mm}

\caption{Line-plane angular constraints. }
\label{fig.linePlaneAngle}
\end{figure}

\begin{figure}[h]
\centering \subfloat[Parallel.] {\label{fig.planePlaneParallel}
\begin{minipage}[b]{0.3\linewidth}
\centering
\includegraphics[width=1.2\linewidth]{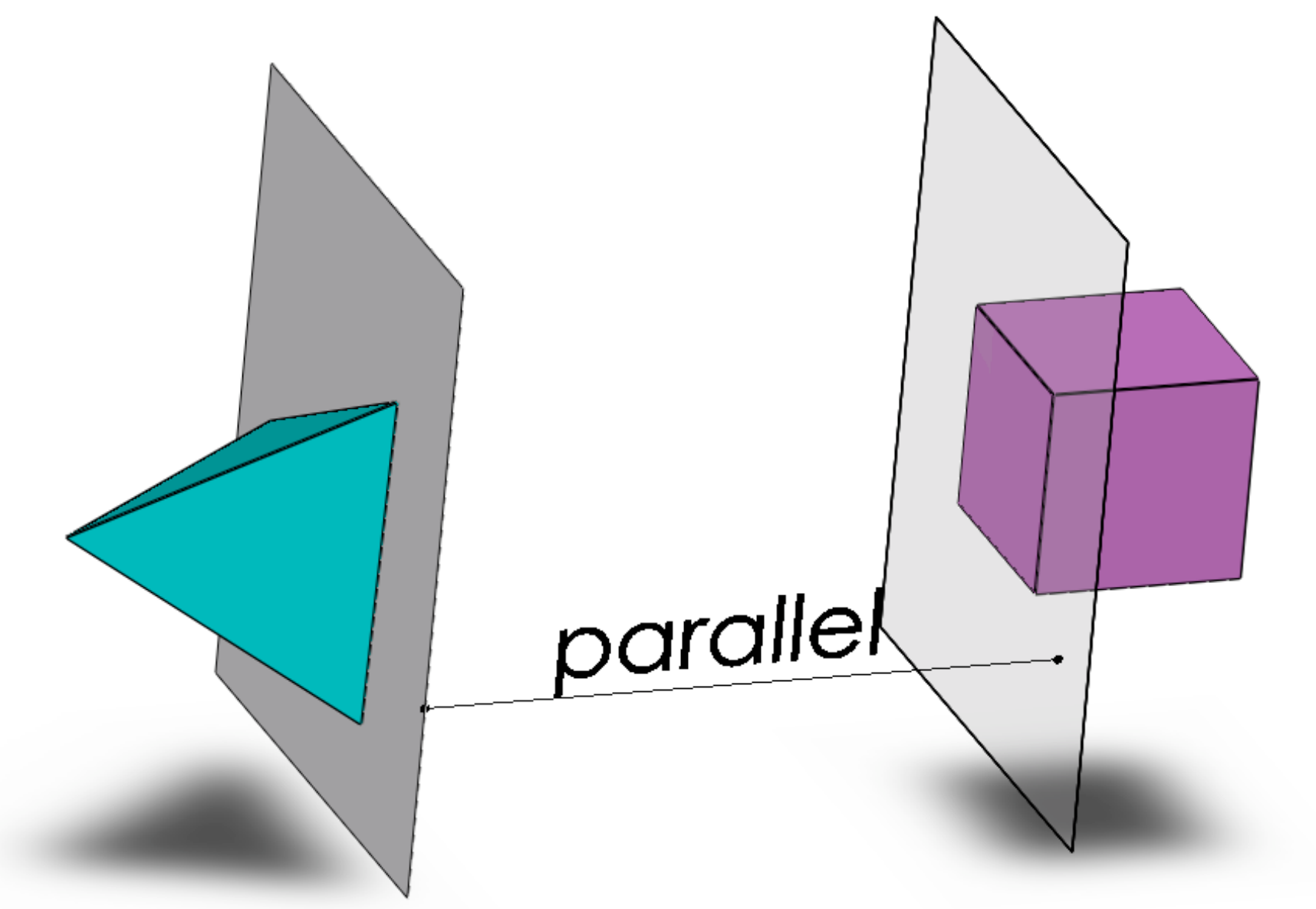}
\end{minipage}}%
\hspace{2mm}
\centering \subfloat[Perpendicular.]{\label{fig.planePlanePerp}
\begin{minipage}[b]{0.3\linewidth}
\centering\includegraphics[width=\linewidth]{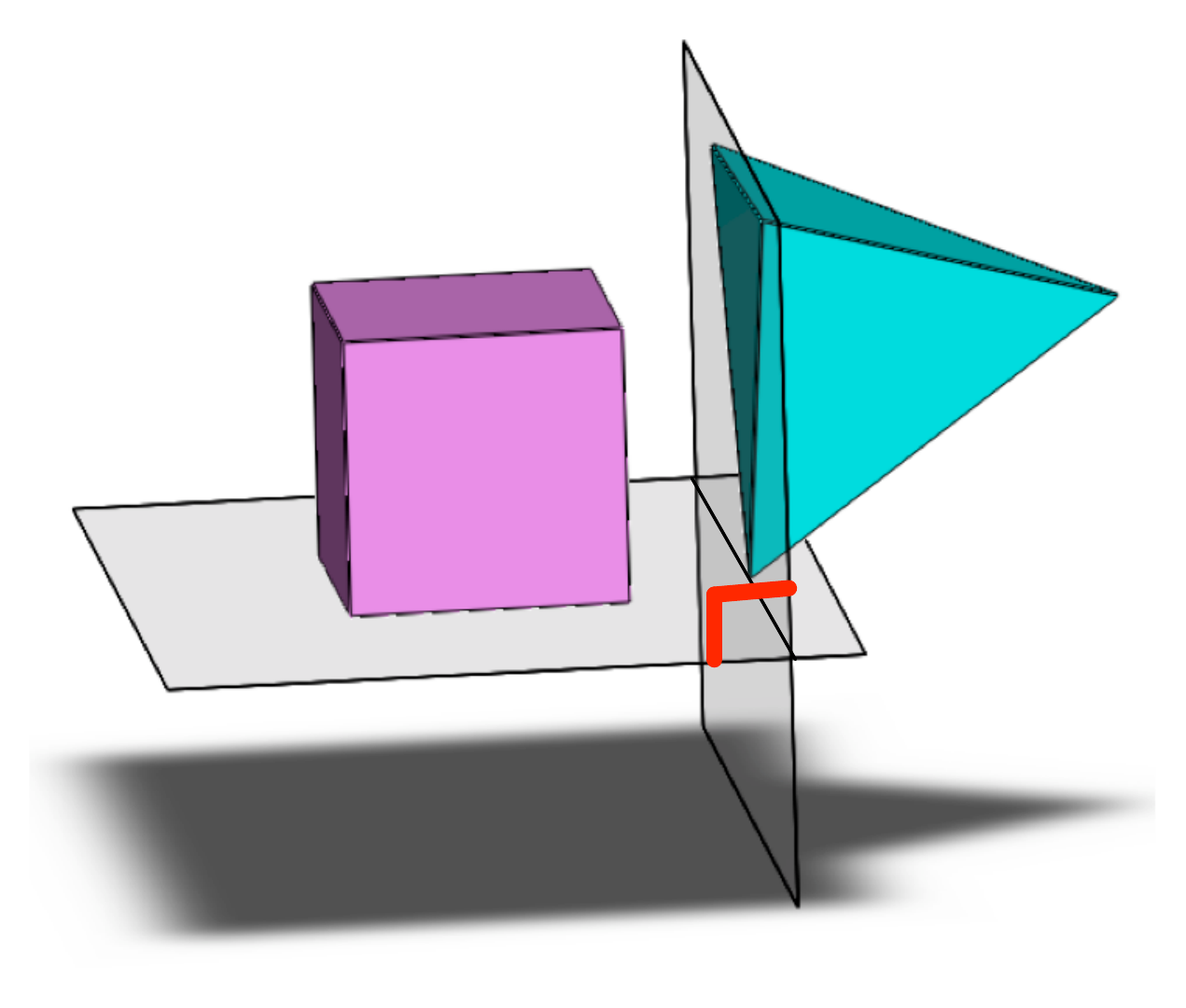}
\end{minipage}}
\hspace{2mm}
\centering \subfloat[Arbitrary fixed angular.]{\label{fig.planePlaneArbAngle}
\begin{minipage}[b]{0.3\linewidth}
\centering\includegraphics[width=\linewidth]{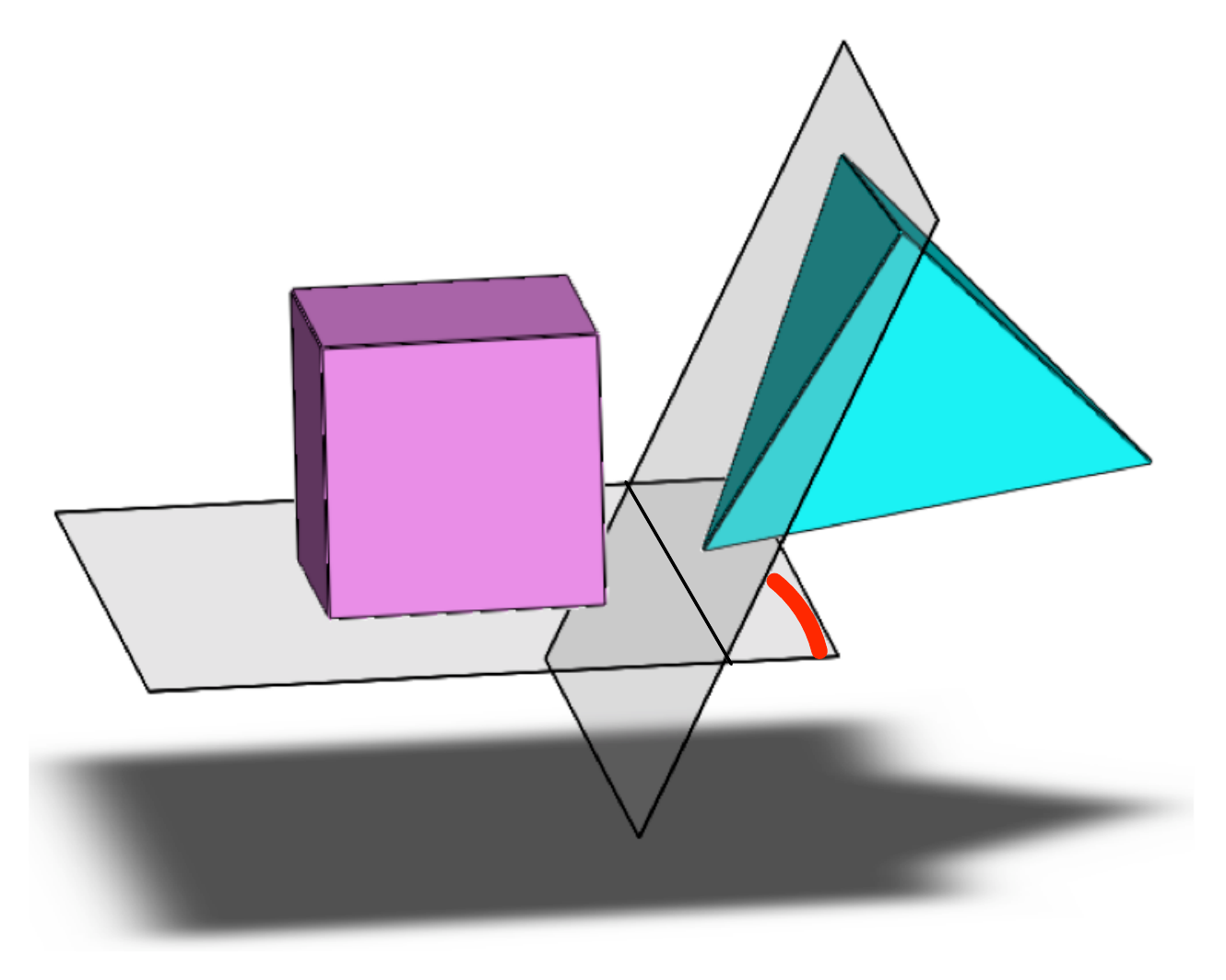}
\end{minipage}}
\caption{Plane-plane angular constraints. }
\label{fig.planePlaneAngle}
\end{figure}

We now formally define the functions for describing angular constraints. 
\begin{itemize}
	\item {\bf Line-line parallel:} 
	$L_{7}:E_{7} \rightarrow \Re^3 \times \Re^3 \times \Re^3$
	maps an edge $e=ij$ to a triple $(\vec p_i, \vec p_j, \vec d)$ so that
	the lines $(\vec p_i, \vec d)$ and $(\vec p_j, \vec d)$ affixed to bodies $i$ and $j$, respectively,
	are constrained to remain parallel to each other.
	
	\item {\bf Line-line perpendicular:} 
	$L_{8}:E_{8} \rightarrow (\Re^3 \times \Re^3) \times (\Re^3 \times \Re^3)$
	maps an edge $e=ij$ to a pair $((\vec p_i, \vec d_i), (\vec p_j, \vec d_j))$ so that
	the lines $(\vec p_i, \vec d_i)$ and $(\vec p_j, \vec d_j)$ affixed to bodies $i$ and $j$, respectively,
	are constrained to remain perpendicular to each other.
	
	\item {\bf Line-line fixed angular:} 
	$L_{9}:E_{9} \rightarrow (\Re^3 \times \Re^3) \times (\Re^3 \times \Re^3) \times \Re$
	maps an edge $e=ij$ to a triple $((\vec p_i, \vec d_i), (\vec p_j, \vec d_j), \alpha)$ so that
	the lines $(\vec p_i, \vec d_i)$ and $(\vec p_j, \vec d_j)$ affixed to bodies $i$ and $j$, respectively,
	are constrained to maintain the angle $\alpha$ between them.
	
	\item {\bf Line-plane parallel:} 
	$L_{12}:E_{12} \rightarrow (\Re^3 \times \Re^3) \times (\Re^3 \times \Re^3)$
	maps an edge $e=ij$ to a pair $((\vec p_i, \vec d_i), (\vec p_j, \vec d_j))$ so that
	the line $(\vec p_i, \vec d_i)$ and plane $(\vec p_j, \vec d_j)$ affixed to bodies $i$ and $j$, respectively,
	are constrained to remain parallel to each other.
	
	\item {\bf Line-plane perpendicular:} 
	$L_{13}:E_{13} \rightarrow (\Re^3 \times \Re^3) \times (\Re^3 \times \Re^3)$
	maps an edge $e=ij$ to a pair $((\vec p_i, \vec d_i), (\vec p_j, \vec d_j))$ so that
	the line $(\vec p_i, \vec d_i)$ and plane $(\vec p_j, \vec d_j)$ affixed to bodies $i$ and $j$, respectively,
	are constrained to remain perpendicular to each other.
	\comment{should this just have one $\vec d$?}

	\item {\bf Line-plane fixed angular:} 
	$L_{14}:E_{14} \rightarrow (\Re^3 \times \Re^3) \times (\Re^3 \times \Re^3) \times \Re$
	maps an edge $e=ij$ to a triple $((\vec p_i, \vec d_i), (\vec p_j, \vec d_j), \alpha)$ so that
	the line $(\vec p_i, \vec d_i)$ and plane $(\vec p_j, \vec d_j)$ affixed to bodies $i$ and $j$, respectively,
	are constrained to maintain the angle $\alpha$ between them.
	
	\item {\bf Plane-plane parallel:} 
	$L_{17}:E_{17} \rightarrow \Re^3 \times \Re^3 \times \Re^3$
	maps an edge $e=ij$ to a triple $(\vec p_i, \vec p_j, \vec d)$ so that
	the planes $(\vec p_i, \vec d)$ and $(\vec p_j, \vec d)$ affixed to bodies $i$ and $j$, respectively,
	are constrained to remain parallel to each other.
	
	\item {\bf Plane-plane perpendicular:} 
	$L_{18}:E_{18} \rightarrow (\Re^3 \times \Re^3) \times (\Re^3 \times \Re^3)$
	maps an edge $e=ij$ to a pair $((\vec p_i, \vec d_i), (\vec p_j, \vec d_j))$ so that
	the planes $(\vec p_i, \vec d_i)$ and $(\vec p_j, \vec d_j)$ affixed to bodies $i$ and $j$, respectively,
	are constrained to remain perpendicular to each other.
	
	\item {\bf Plane-plane fixed angular:} 
	$L_{19}:E_{19} \rightarrow (\Re^3 \times \Re^3) \times (\Re^3 \times \Re^3) \times \Re$
	maps an edge $e=ij$ to a triple $((\vec p_i, \vec d_i), (\vec p_j, \vec d_j), \alpha)$ so that
	the planes $(\vec p_i, \vec d_i)$ and $(\vec p_j, \vec d_j)$ affixed to bodies $i$ and $j$, respectively,
	are constrained to maintain the angle $\alpha$ between them.

\end{itemize}

It is straightforward that the line-line angular constraints (perpendicular, fixed angular
and parallel) are expressed by the two basic angular building blocks.
We observe that the line-plane and plane-plane angular constraints 
reduce to them as follows. 
\label{sec:extensionAngle}
\begin{itemize}
	\item {\bf Line-plane parallel:} Reduces to {\bf line-line non-parallel
	fixed angular} using
	the normal to the plane. 
	\item {\bf Line-plane perpendicular:} Reduces to {\bf line-line parallel} using
	the normal to the plane.
	\item {\bf Line-plane fixed angular:} Reduces to {\bf line-line non-parallel
	fixed angular} using
	the normal to the plane. 
	\item {\bf Plane-plane parallel:} Reduces to {\bf line-line parallel} using
	the planes' normal.
	\item {\bf Plane-plane perpendicular:} Reduces to {\bf line-line non-parallel
	fixed angular} using
	the planes' normals.
	\item {\bf Plane-plane fixed angular:} Reduces to {\bf line-line non-parallel
	fixed angular} using
	the planes' normals.
\end{itemize}

%% file: blindConstraints.tex
\subsection{Blind constraints}
The remaining coincidence and distance constraints reduce to some combination of basic angular  
and basic blind constraints from Section \ref{sec:buildingBlocks}.
We consider them in the following order:
{\bf point-point}, {\bf point-line}, {\bf point-plane}, {\bf line-line}, 
{\bf line-plane} and {\bf plane-plane}.
Since a {\bf point-point distance} constraint 
{\em (Figure \ref{fig.pointPointDistance})}) 
is a bar (see \cite{tay:rigidityMultigraphs-I:1984,white:whiteley:algebraicGeometryFrameworks:1987}), we consider only the {\bf point-point coincidence} constraint.

\input{pointConstraints.tex}

\input{lineConstraints.tex}

\input{planeConstraints.tex}

%% file: pointConstraints.tex
%
\begin{figure}[h]
\centering \subfloat[{\em Coincidence:} bodies $i$ and $j$ must
coincide on the specified point $\vec p$.] {\label{fig.pointPointCoincident}
\begin{minipage}[b]{0.45\linewidth}
\centering
\includegraphics[width=.75\linewidth]{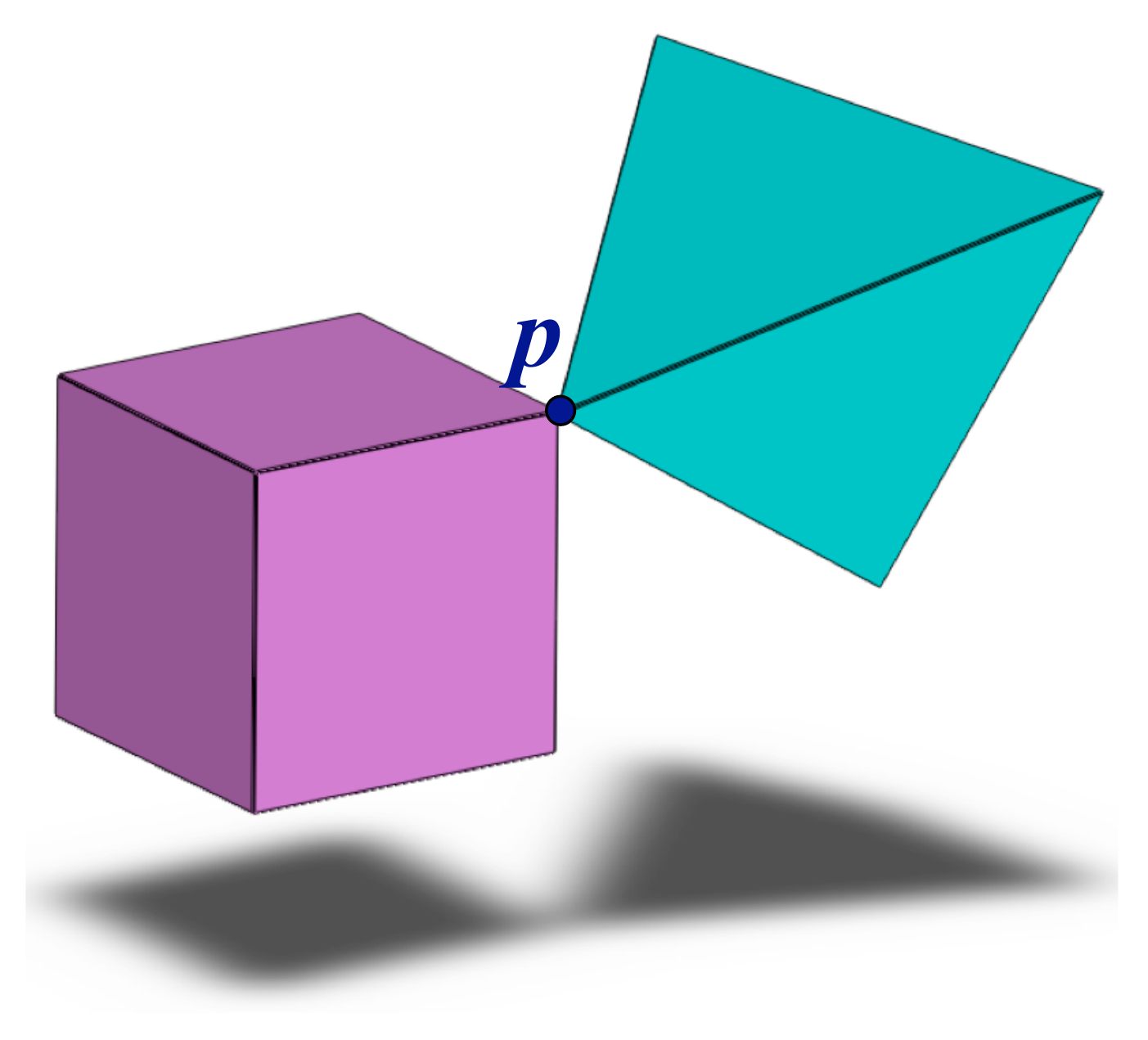}
\end{minipage}}%
\hspace{2mm}
\centering \subfloat[{\em Distance (bar)}: a point affixed to body $i$ 
must be a fixed distance from a point affixed to body $j$.]{\label{fig.pointPointDistance}
\begin{minipage}[b]{0.45\linewidth}
\centering\includegraphics[width=\linewidth]{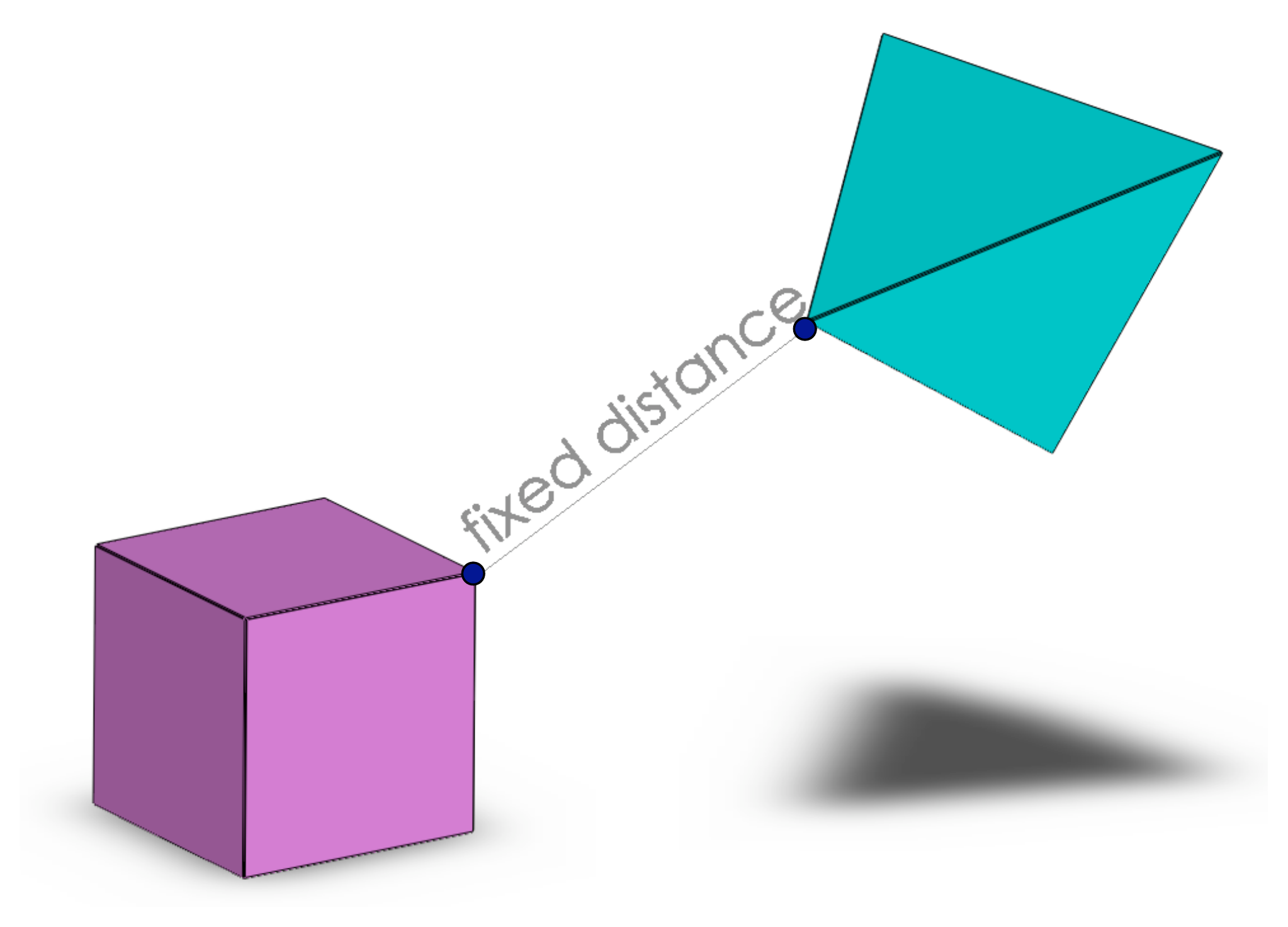}
\end{minipage}}
\caption{Point-point constraints. }
\label{fig.pointPoint}
\end{figure}

The function $L_1:E_1 \rightarrow \Re^3$
maps an edge $e=ij$ to a point $\vec p$ so that it is constrained
to lie on bodies $i$ and $j$ simultaneously.

Then the {\bf point-point coincidence} constraint 
{\em (Figure \ref{fig.pointPointCoincident})}
is infinitesimally maintained if the relative velocity of $\vec p$ is exactly 0. Since 
$\vec s_i = (-\vc \omega_i, \vec v_i)$ and $\vec s_j = (-\vc \omega_j, \vec v_j)$, 
then the relative screw is
defined by $(-\vc \omega,\vec  v)$, where $\vc \omega = (\vc \omega_i - \vc \omega_j)$ 
and $\vec v = (\vec v_i - \vec v_j)$.
Therefore, the constraint is infinitesimally maintained
if and only if $\vec p$'s infinitesimal velocity
$\vec p' = \vc \omega \times \vec p + \vec v = 0$, i.e.,
\begin{eqnarray*}
	\omega^y p^z - \omega^z p^y + v^x &=&0\\
	\omega^z p^x - \omega^x p^z + v^y &=&0\\
	\omega^x p^y - \omega^y p^x + v^z &=&0
\end{eqnarray*}
if and only if
\begin{eqnarray}
	\left\langle \vec (s_i^* - s_j^*),(1, 0, 0, 0, -p^z, p^y)\right\rangle  &=&0\\
	\left\langle \vec (s_i^* - s_j^*),(0, 1, 0, p^z, 0, -p^x)\right\rangle &=&0\\
	\left\langle \vec (s_i^* - s_j^*),(0, 0, 1, -p^y, p^x, 0)\right\rangle &=&0
\end{eqnarray}
Thus, a point-point coincidence constraint corresponds to 3 rows in the rigidity matrix:
\begin{center}
\begin{tabular}{ccccccc}
 & \multicolumn{2}{c}{$\vec s_i^*$} & & \multicolumn{2}{c}{$\vec s_j^*$} & \\
\cline{2-3} \cline{5-6}
$\cdots$ & $\vec v_i$& $-\vc\omega_i$ &
$\cdots$ & $\vec v_j$& $-\vc\omega_j$ & $\cdots$ \\
\hline
\multicolumn{1}{|c|}{$\rmfill$} & 
	\multicolumn{2}{c|}{\cellcolor{lightgray}$(1, 0, 0, 0, -p^z, p^y)$}&
	\multicolumn{1}{c|}{$\rmfill$} &
	\multicolumn{2}{c|}{\cellcolor{lightgray}$(-1, 0, 0, 0, p^z, -p^y)$}&	
	\multicolumn{1}{c|}{$\rmfill$} \\
\hline
\multicolumn{1}{|c|}{$\rmfill$} & 
	\multicolumn{2}{c|}{\cellcolor{lightgray}$(0, 1, 0, p^z, 0, -p^x)$}&
	\multicolumn{1}{c|}{$\rmfill$} &
	\multicolumn{2}{c|}{\cellcolor{lightgray}$(0, -1, 0, -p^z, 0, p^x)$}&	
	\multicolumn{1}{c|}{$\rmfill$} \\
\hline
\multicolumn{1}{|c|}{$\rmfill$} & 
	\multicolumn{2}{c|}{\cellcolor{lightgray}$(0, 0, 1, -p^y, p^x, 0)$}&
	\multicolumn{1}{c|}{$\rmfill$} &
	\multicolumn{2}{c|}{\cellcolor{lightgray}$(0, 0, -1, p^y, -p^x, 0)$}&	
	\multicolumn{1}{c|}{$\rmfill$} \\
\hline
\end{tabular}
\end{center}

%
\begin{figure}[h]
\centering \subfloat[{\em Coincidence:} a point $\vec p_i$ affixed to
body $i$ must lie on a line with direction $\vec d$ affixed to body $j$.
Then the instantaneous velocity $\vec p_i'$ of 
$\vec p_i$ must lie in the same direction as $\vec d$.] {\label{fig.linePointCoincident}
\begin{minipage}[b]{0.45\linewidth}
\centering
\includegraphics[width=.75\linewidth]{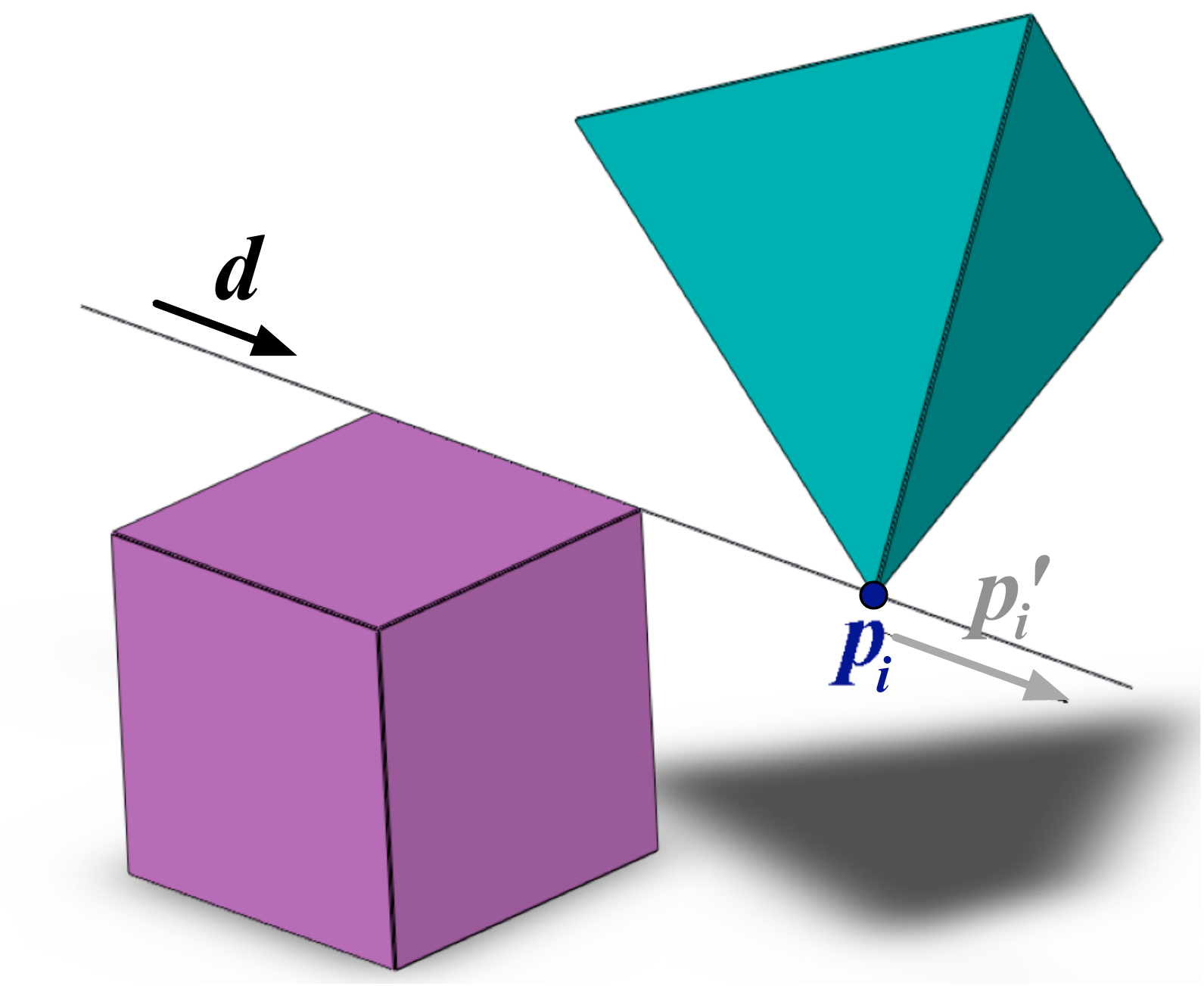}
\end{minipage}}%
\hspace{2mm}
\centering \subfloat[{\em Distance:} a point $\vec p_i$ affixed to
body $i$ must be a fixed distance from a line affixed to body $j$.
If $\hat{\vec d}$ is the perpendicular direction from the line  
to $\vec p_i$, then the instantaneous velocity $\vec p_i'$ of 
$\vec p_i$ must be orthogonal to $\hat{\vec d}$. I.e., $\vec p_i'$ must
lie in the plane $(\vec p_i, \hat{\vec d})$.]{\label{fig.linePointDistance}
\begin{minipage}[b]{0.45\linewidth}
\centering\includegraphics[width=.75\linewidth]{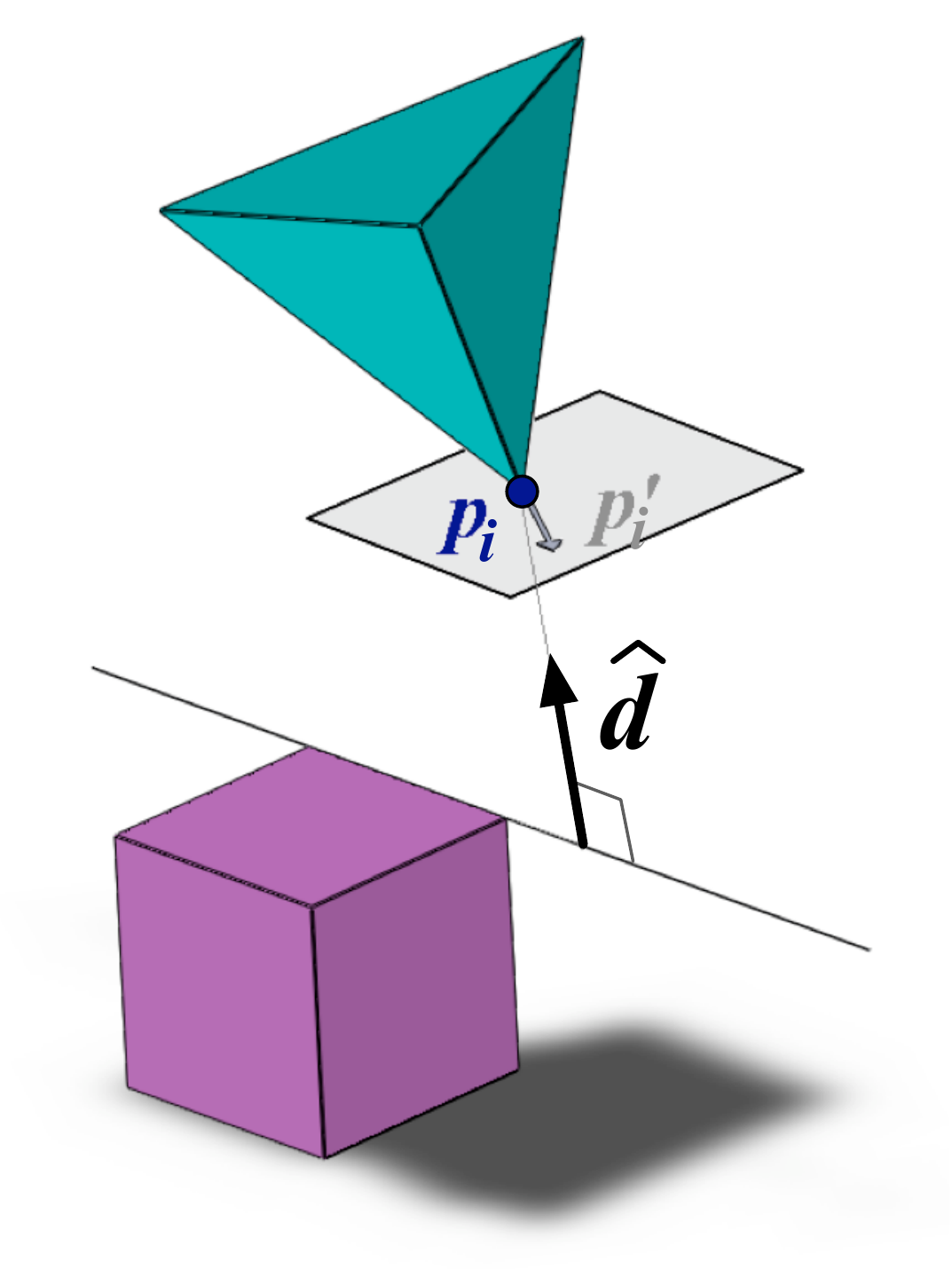}
\end{minipage}}
\caption{Point-line constraints. }
\label{fig.linePoint}
\end{figure}
\noindent{\bf Point-line coincidence} {\em (Figure \ref{fig.linePointCoincident})}:
The function $L_3:E_3 \rightarrow \Re^3 \times (\Re^3 \times \Re^3)$
maps an edge $e=ij$ to a pair $(\vec p_i, (\vec p_j, \vec d))$ so that
point $\vec p_i$ affixed to body $i$
is constrained to lie on the line $(\vec p_j, \vec d)$ affixed to body $j$.

The {\bf point-line coincidence} constraint is infinitesimally
maintained by using 2 primitive {\em blind} constraints from Equations \ref{eq:velSameDirAsVec1}
and \ref{eq:velSameDirAsVec2} to express that
the relative velocity of $\vec p_i$ lies in the same direction as $\vec d$:
\begin{eqnarray}
	\left\langle (\vec s_i - \vec s_j)^*, (\vec p_i:1) \vee (d^y, -d^x, 0, 0)\right\rangle &=&0\\
	\left\langle (\vec s_i - \vec s_j)^*, (\vec p_i:1) \vee (0, d^z, -d^y, 0)\right\rangle &=&0
\end{eqnarray}
Thus, a point-line coincidence constraint corresponds to 2 rows in the rigidity matrix:
\begin{center}
\begin{tabular}{ccccccc}
 & \multicolumn{2}{c}{$\vec s_i^*$} & & \multicolumn{2}{c}{$\vec s_j^*$} & \\
\cline{2-3} \cline{5-6}
$\cdots$ & $\vec v_i$& $-\vc\omega_i$ &
$\cdots$ & $\vec v_j$& $-\vc\omega_j$ & $\cdots$ \\
\hline
\multicolumn{1}{|c|}{$\rmfill$} & 
	\multicolumn{2}{c|}{\cellcolor{lightgray}$(\vec p_i:1) \vee (d^y, -d^x, 0, 0)$}&
	\multicolumn{1}{c|}{$\rmfill$} &
	\multicolumn{2}{c|}{\cellcolor{lightgray}$-((\vec p_i:1) \vee (d^y, -d^x, 0, 0))$}&	
	\multicolumn{1}{c|}{$\rmfill$} \\
\hline
\multicolumn{1}{|c|}{$\rmfill$} & 
	\multicolumn{2}{c|}{\cellcolor{lightgray}$(\vec p_i:1) \vee (0, d^z, -d^y, 0)$}&
	\multicolumn{1}{c|}{$\rmfill$} &
	\multicolumn{2}{c|}{\cellcolor{lightgray}$-((\vec p_i:1) \vee (0, d^z, -d^y, 0))$}&	
	\multicolumn{1}{c|}{$\rmfill$} \\
\hline
\end{tabular}
\end{center}

\noindent{\bf Point-line distance} {\em (Figure \ref{fig.linePointDistance})}:
The function $L_4:E_4 \rightarrow \Re^3 \times (\Re^3 \times \Re^3) \times \Re$
maps an edge $e=ij$ to a triple $(\vec p_i, (\vec p_j, \vec d), a)$ so that
point $\vec p_i$ affixed to body $i$
is constrained to lie a distance $a$ from the line $(\vec p_j, \vec d)$ affixed to body $j$.

Let $\hat{\vec d}$ be the perpendicular direction from the line $(\vec p_j, \vec d)$ 
to $\vec p_i$. Then the 
{\bf point-line distance} constraint is infinitesimally maintained using
1 primitive {\em blind} constraint from
Equation \ref{eq:velOrthogToVec} to express that the 
relative velocity of $\vec p_i$ is orthogonal to $\hat{\vec d}$:
\begin{equation}
	\left\langle (\vec s_i -\vec s_j)^*, (\vec p_i:1) \vee (\hat{\vec d}:0)\right\rangle = 0
\end{equation}
Thus, a point-line distance constraint corresponds to one row in the rigidity matrix:
\begin{center}
\begin{tabular}{ccccccc}
 & \multicolumn{2}{c}{$\vec s_i^*$} & & \multicolumn{2}{c}{$\vec s_j^*$} & \\
\cline{2-3} \cline{5-6}
$\cdots$ & $\vec v_i$& $-\vc\omega_i$ &
$\cdots$ & $\vec v_j$& $-\vc\omega_j$ & $\cdots$ \\
\hline
\multicolumn{1}{|c|}{$\rmfill$} & 
	\multicolumn{2}{c|}{\cellcolor{lightgray}$(\vec p_i:1) \vee (\hat{\vec d}:0)$}&
	\multicolumn{1}{c|}{$\rmfill$} &
	\multicolumn{2}{c|}{\cellcolor{lightgray}$-((\vec p_i:1) \vee (\hat{\vec d}:0))$}&	
	\multicolumn{1}{c|}{$\rmfill$} \\
\hline
\end{tabular}
\end{center}

%
\begin{figure}[h]
\centering \subfloat[{\em Coincidence:} a point $\vec p_i$ affixed to 
body $i$ must lie in a plane with normal $\vec d$ affixed to 
body $j$. Then the instantaneous velocity $\vec p_i'$ of 
$\vec p_i$ must remain in the plane $(\vec p_i, \vec d)$.] {\label{fig.planePointCoincident}
\begin{minipage}[b]{0.45\linewidth}
\centering
\includegraphics[width=.85\linewidth]{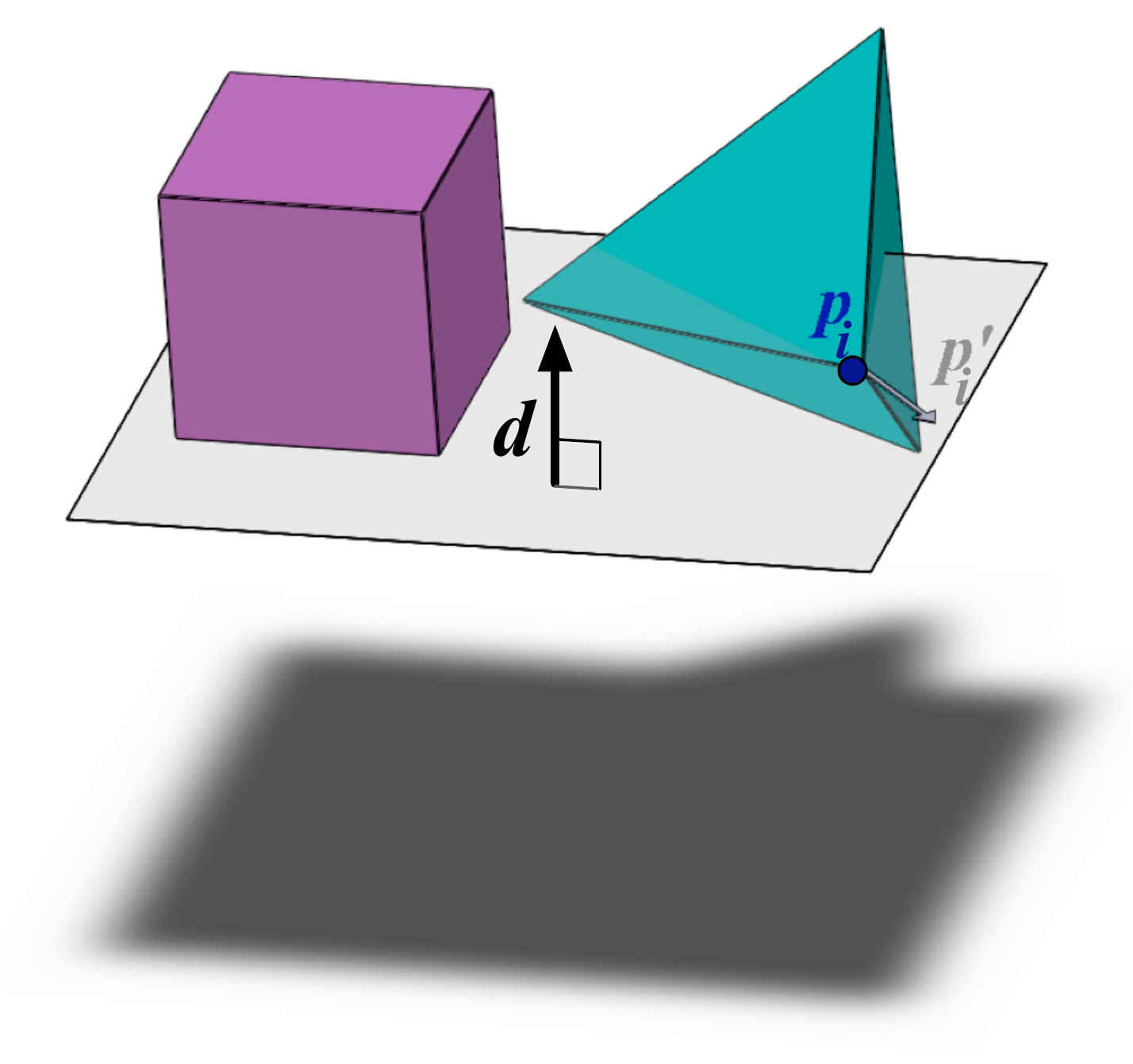}
\end{minipage}}%
\hspace{2mm}
\centering \subfloat[{\em Distance:} a point $\vec p_i$ affixed to 
body $i$ must be a fixed distance from a plane with normal $\vec d$ affixed to 
body $j$. Then the instantaneous velocity $\vec p_i'$ of 
$\vec p_i$ must remain in the plane $(\vec p_i, \vec d)$.]{\label{fig.planePointDistance}
\begin{minipage}[b]{0.45\linewidth}
\centering\includegraphics[width=.85\linewidth]{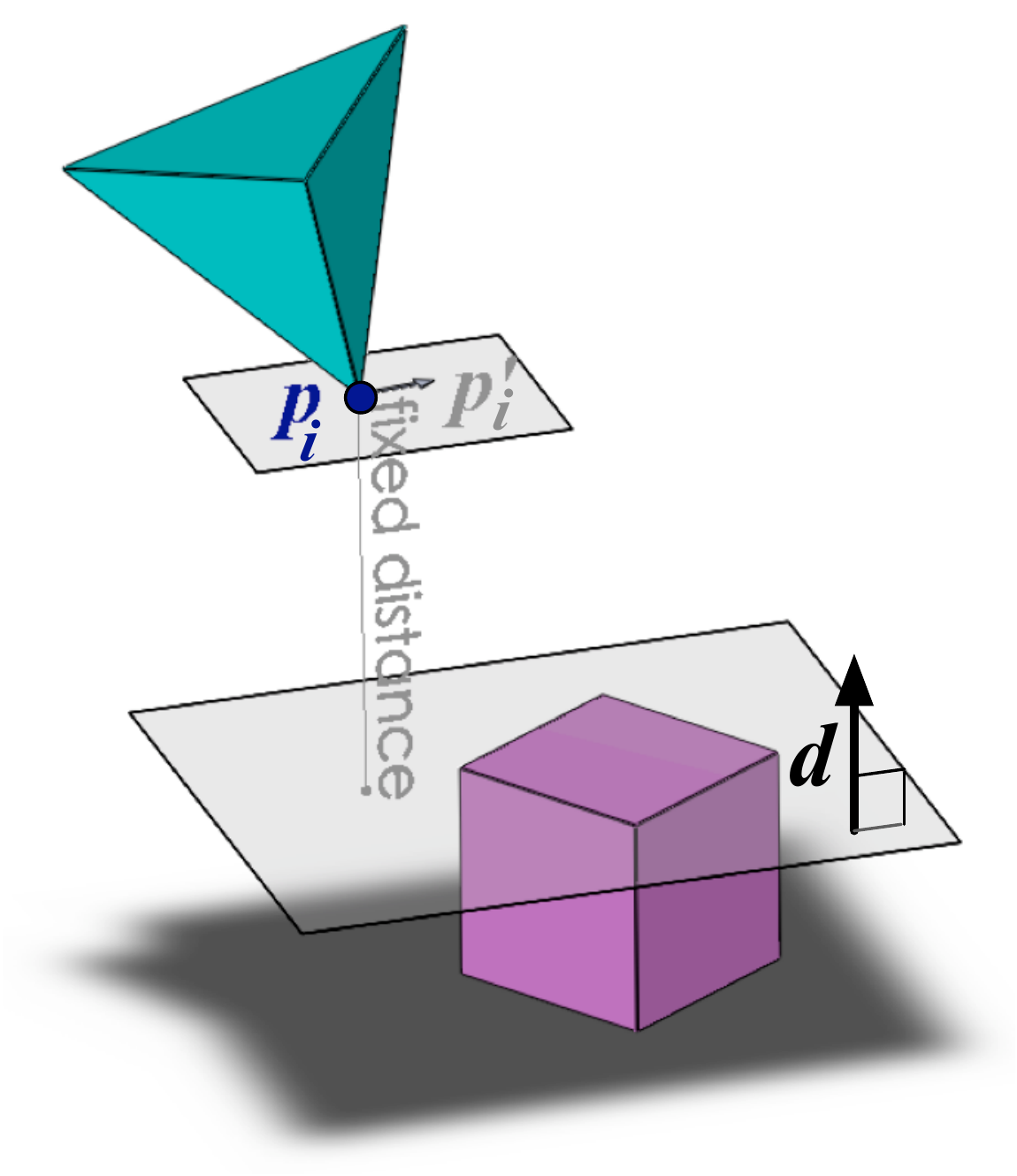}
\end{minipage}}
\caption{Point-plane constraints. }
\label{fig.planePoint}
\end{figure}
\noindent{\bf Point-plane coincidence} {\em (Figure \ref{fig.planePointCoincident})}:
The function $L_5:E_5 \rightarrow \Re^3 \times (\Re^3 \times \Re^3)$
maps an edge $e=ij$ to a pair $(\vec p_i, (\vec p_j, \vec d))$ so that
the point $\vec p_i$ affixed to body $i$ is constrained to lie
in the plane $(\vec p_j, \vec d)$ affixed to body $j$.

The {\bf point-plane coincidence} constraint is
infinitesimally maintained using Equation \ref{eq:velOrthogToVec} to 
express that the relative velocity of $\vec p_i$ remains in the plane:
\begin{equation}
	\left\langle (\vec s_i - \vec s_j)^*, (\vec p_i:1) \vee (\vec d:0)\right\rangle = 0
\end{equation}
Thus, a point-plane coincidence constraint corresponds to one row in the rigidity matrix:
\begin{center}
\begin{tabular}{ccccccc}
 & \multicolumn{2}{c}{$\vec s_i^*$} & & \multicolumn{2}{c}{$\vec s_j^*$} & \\
\cline{2-3} \cline{5-6}
$\cdots$ & $\vec v_i$& $-\vc\omega_i$ &
$\cdots$ & $\vec v_j$& $-\vc\omega_j$ & $\cdots$ \\
\hline
\multicolumn{1}{|c|}{$\rmfill$} & 
	\multicolumn{2}{c|}{\cellcolor{lightgray}$ (\vec p_i:1) \vee (\vec d:0)$}&
	\multicolumn{1}{c|}{$\rmfill$} &
	\multicolumn{2}{c|}{\cellcolor{lightgray}$- ((\vec p_i:1) \vee (\vec d:0))$}&	
	\multicolumn{1}{c|}{$\rmfill$} \\
\hline
\end{tabular}
\end{center}

\noindent{\bf Point-plane distance} {\em (Figure \ref{fig.planePointDistance})}:
The function $L_6:E_6 \rightarrow \Re^3 \times (\Re^3 \times \Re^3) \times \Re$
maps an edge $e=ij$ to a triple $(\vec p_i, (\vec p_j, \vec d), a)$ so that
the point $\vec p_i$ affixed to body $i$ is constrained to lie a distance $a$
from the plane $(\vec p_j, \vec d)$ affixed to body $j$.

The {\bf point-plane distance} constraint 
is infinitesimally maintained by using 
Equation \ref{eq:velOrthogToVec} to 
express that the relative velocity of $\vec p_i$ remains parallel to the plane:
\begin{equation}
	\left\langle (\vec s_i - \vec s_j)^*, (\vec p_i:1) \vee (\vec d:0)\right\rangle = 0
\end{equation}
Thus, a point-plane coincidence constraint corresponds to one row in the rigidity matrix:
\begin{center}
\begin{tabular}{ccccccc}
 & \multicolumn{2}{c}{$\vec s_i^*$} & & \multicolumn{2}{c}{$\vec s_j^*$} & \\
\cline{2-3} \cline{5-6}
$\cdots$ & $\vec v_i$& $-\vc\omega_i$ &
$\cdots$ & $\vec v_j$& $-\vc\omega_j$ & $\cdots$ \\
\hline
\multicolumn{1}{|c|}{$\rmfill$} & 
	\multicolumn{2}{c|}{\cellcolor{lightgray}$ (\vec p_i:1) \vee (\vec d:0)$}&
	\multicolumn{1}{c|}{$\rmfill$} &
	\multicolumn{2}{c|}{\cellcolor{lightgray}$-((\vec p_i:1) \vee (\vec d:0))$}&	
	\multicolumn{1}{c|}{$\rmfill$} \\
\hline
\end{tabular}
\end{center}

%% file: lineConstraints.tex
%
\begin{figure}[h]
\centering \subfloat[{\em Coincidence:} bodies $i$ and $j$ must
coincide on the specified line $(\vec p, \vec d)$. Then, in addition
to a line-line parallel constraint, the instantaneous velocity $\vec p'$
must lie in the same direction as $\vec d$.] {\label{fig.lineLineCoincident}
\begin{minipage}[b]{0.45\linewidth}
\centering
\includegraphics[width=.7\linewidth]{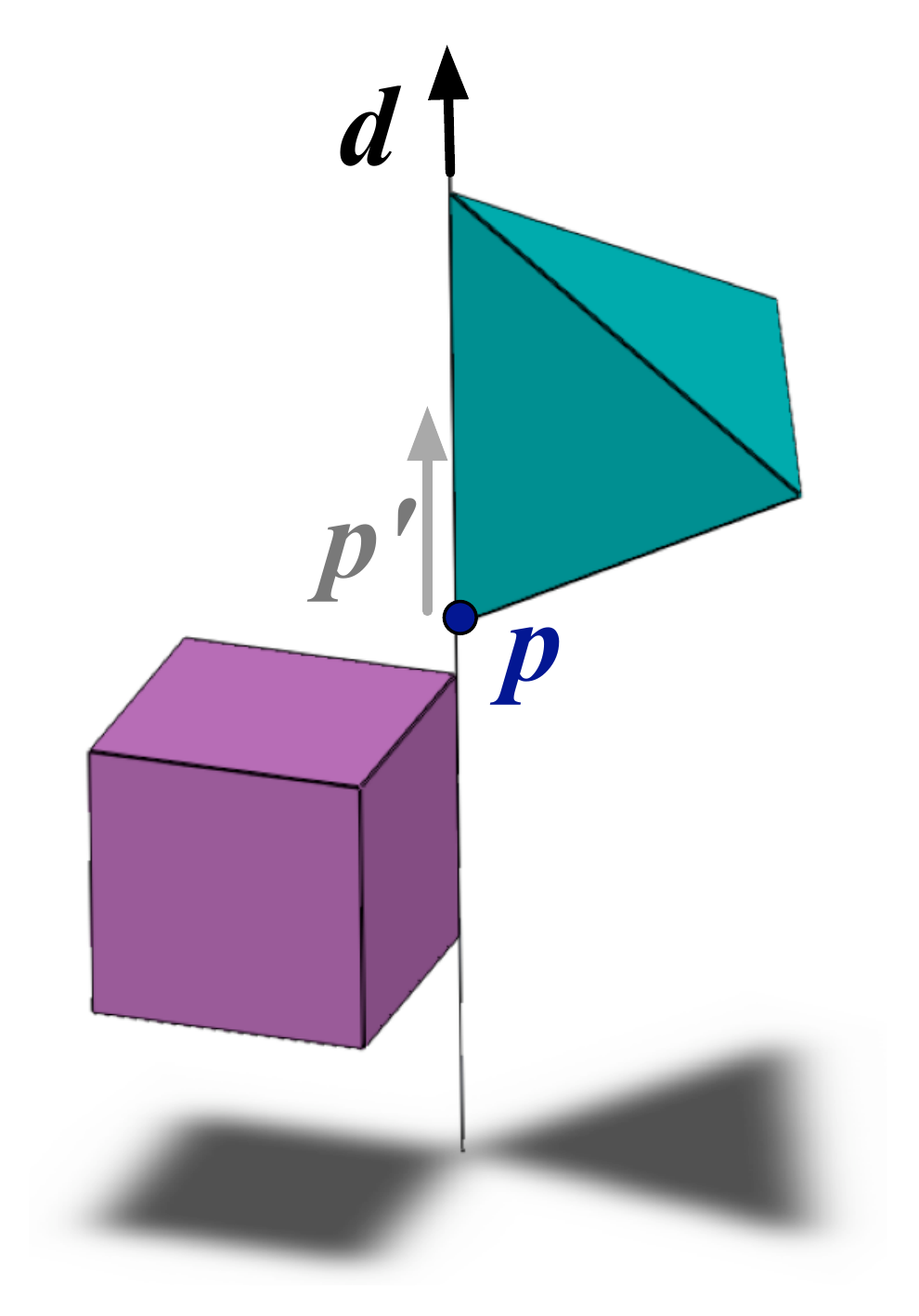}
\end{minipage}}%
\hspace{2mm}
\centering \subfloat[{\em Distance:} a line $(\vec p_i, \vec d_i)$ affixed to 
body $i$ must be a fixed distance from a line $(\vec p_j, \vec d_j)$  
affixed to body $j$. If $\vec p$ is the point on the line $(\vec p_i, \vec d_i)$
that is closest to the line $(\vec p_j, \vec d_j)$, then the instantaneous velocity 
$\vec p'$ of $\vec p$ must be orthogonal to $\vec d_i \times \vec d_j$, the
perpendicular to both lines. I.e., $\vec p'$ must lie in the plane
$(\vec p, \vec d_i \times \vec d_j)$]{\label{fig.lineLineDistance}
\begin{minipage}[b]{0.45\linewidth}
\centering\includegraphics[width=1.15\linewidth]{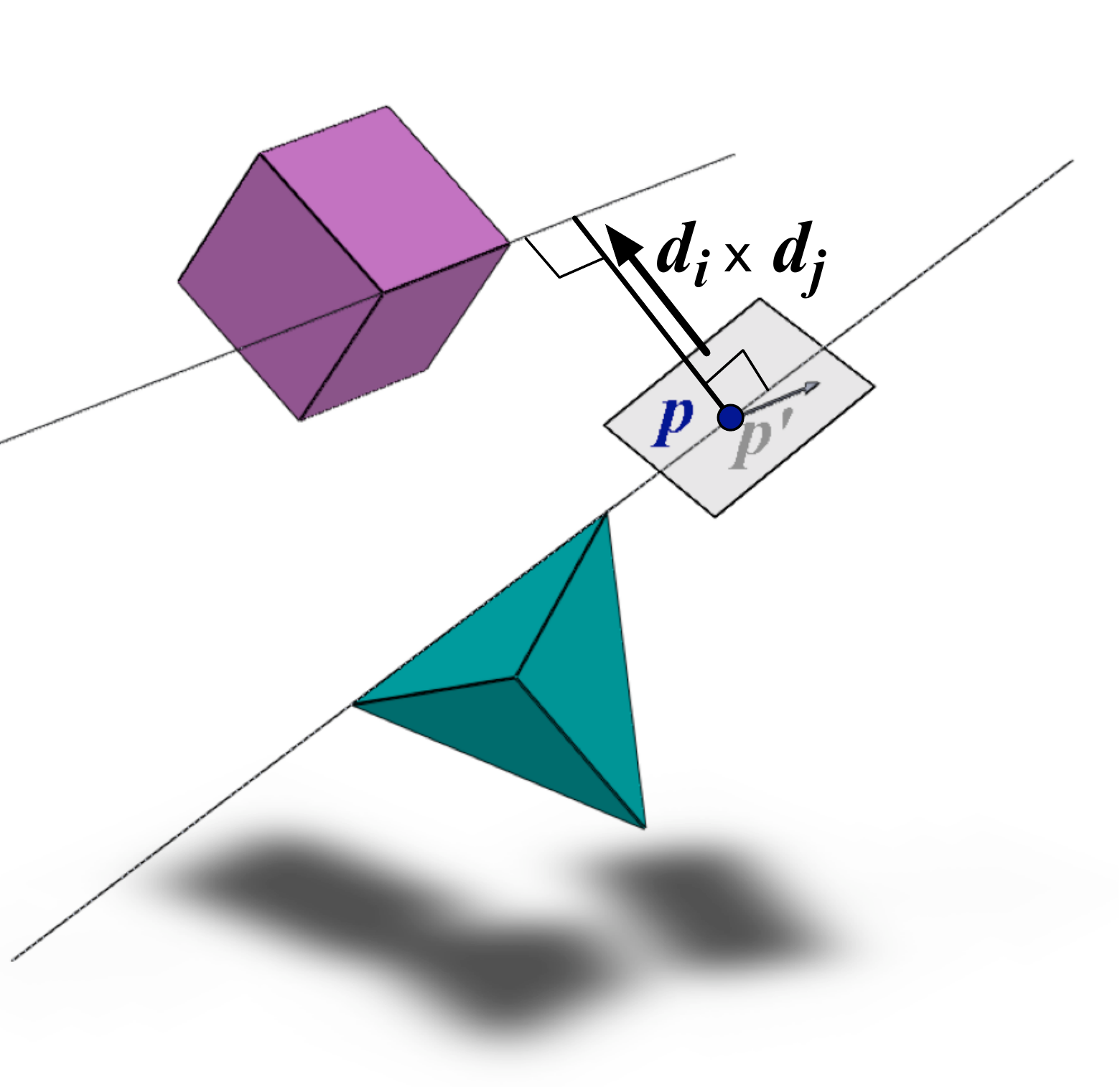}
\end{minipage}}
\caption{Line-line constraints. }
\label{fig.lineLine}
\end{figure}
\noindent{\bf Line-line coincidence} {\em (Figure \ref{fig.lineLineCoincident})}:
The function $L_{10}:E_{10} \rightarrow \Re^3 \times \Re^3$
maps an edge $e=ij$ to a pair $(\vec p, \vec d)$ so that
the line $(\vec p, \vec d)$ is constrained to be affixed to bodies $i$ and $j$ 
simultaneously.

We place a {\bf line-line 
parallel} angular constraint, resulting in 2 primitive {\em angular} constraints from Equations
\ref{eq:parallelAngleInf1} and \ref{eq:parallelAngleInf2}:
\begin{eqnarray}
	\left\langle (\vec s_i^* - \vec s_j^*), (0,0,0,-d^y,d^x, 0)\right\rangle \\
	\left\langle (\vec s_i^* - \vec s_j^*), (0,0,0,0,-d^z,d^y)\right\rangle
\end{eqnarray}
Then, to maintain
coincidence, associate 2 primitive {\em blind} constraints from
Equations \ref{eq:velSameDirAsVec1} and \ref{eq:velSameDirAsVec2}
to force the relative velocity of $\vec p$ to lie along $\vec d$:
\begin{eqnarray}
	\left\langle (\vec s_i - \vec s_j)^*, (\vec p:1) \vee (d^y, -d^x, 0, 0)\right\rangle &=&0\\
	\left\langle (\vec s_i - \vec s_j)^*, (\vec p:1) \vee (0, d^z, -d^y, 0)\right\rangle &=&0
\end{eqnarray}

These 4 equations maintain the {\bf line-line coincidence} 
constraint infinitesimally and correspond to 4 rows in the 
 rigidity matrix:
\begin{center}
	
\begin{tabular}{ccccccc}
 & \multicolumn{2}{c}{$\vec s_i^*$} & & \multicolumn{2}{c}{$\vec s_j^*$} & \\
\cline{2-3} \cline{5-6}
$\cdots$ & $\vec v_i$& $-\vc\omega_i$ &
$\cdots$ & $\vec v_j$& $-\vc\omega_j$ & $\cdots$ \\
\hline
\multicolumn{1}{|c|}{$\rmfill$} & 
	\multicolumn{1}{c|}{\cellcolor{red}$\vec 0$} & 
	\multicolumn{1}{c|}{\cellcolor{lightgray}$(-d^y, d^x, 0)$}&
	\multicolumn{1}{c|}{$\rmfill$} &
	\multicolumn{1}{c|}{\cellcolor{red}$\vec 0$} & 
	\multicolumn{1}{c|}{\cellcolor{lightgray}$(d^y, -d^x, 0)$}&	
	\multicolumn{1}{c|}{$\rmfill$} \\
\hline
\multicolumn{1}{|c|}{$\rmfill$} & 
	\multicolumn{1}{c|}{\cellcolor{red}$\vec 0$} & 
	\multicolumn{1}{c|}{\cellcolor{lightgray}$(0, -d^z, d^y)$}&
	\multicolumn{1}{c|}{$\rmfill$} &
	\multicolumn{1}{c|}{\cellcolor{red}$\vec 0$} & 
	\multicolumn{1}{c|}{\cellcolor{lightgray}$(0, d^z, -d^y)$}&	
	\multicolumn{1}{c|}{$\rmfill$} \\
\hline
\multicolumn{1}{|c|}{$\rmfill$} & 
	\multicolumn{2}{c|}{\cellcolor{lightgray}$(\vec p:1) \vee (d^y, -d^x, 0, 0)$}&
	\multicolumn{1}{c|}{$\rmfill$} &
	\multicolumn{2}{c|}{\cellcolor{lightgray}$-((\vec p:1) \vee (d^y, -d^x, 0, 0))$}&	
	\multicolumn{1}{c|}{$\rmfill$} \\
\hline
\multicolumn{1}{|c|}{$\rmfill$} & 
	\multicolumn{2}{c|}{\cellcolor{lightgray}$(\vec p:1) \vee (0, d^z, -d^y, 0)$}&
	\multicolumn{1}{c|}{$\rmfill$} &
	\multicolumn{2}{c|}{\cellcolor{lightgray}$-((\vec p:1) \vee (0, d^z, -d^y, 0))$}&	
	\multicolumn{1}{c|}{$\rmfill$} \\
\hline
\end{tabular}
\end{center}

\noindent{\bf Line-line distance} {\em (Figure \ref{fig.lineLineDistance})}:
The function $L_{11}:E_{11} \rightarrow (\Re^3 \times \Re^3) \times (\Re^3 \times \Re^3) \times \Re$
maps an edge $e=ij$ to a triple $((\vec p_i, \vec d_i), (\vec p_j, \vec d_j), a)$ so that
the lines $(\vec p_i, \vec d_i)$ and $(\vec p_j, \vec d_j)$ affixed to bodies $i$ and $j$, 
respectively, are constrained to lie a distance $a$ from each other.

Let $\vec p \in \Re^3$ be the
point on the line $(\vec p_i, \vec d_i)$ closest to the line $(\vec p_j, \vec d_j)$. 
Then the {\bf line-line distance} constraint
is infinitesimally maintained if the relative velocity of $\vec p$ 
is orthogonal to the direction perpendicular to both lines.
In other words, $\vec p'$ must lie in
the plane defined by the point $\vec p$ and normal direction
$\vec d_i \times \vec d_j$. 
By substituting $\vec p$ and $\vec d_i \times \vec d_j$
into Equation \ref{eq:velOrthogToVec}, we obtain the linear equation
\begin{equation}
	\left\langle (\vec s_i - \vec s_j)^*, (\vec p:1) \vee ((\vec d_i \times \vec d_j):0)\right\rangle = 0
\end{equation}
associating one {\em blind} primitive constraint. This corresponds to
one row in the rigidity matrix:
\begin{center}
\begin{tabular}{ccccccc}
 & \multicolumn{2}{c}{$\vec s_i^*$} & & \multicolumn{2}{c}{$\vec s_j^*$} & \\
\cline{2-3} \cline{5-6}
$\cdots$ & $\vec v_i$& $-\vc\omega_i$ &
$\cdots$ & $\vec v_j$& $-\vc\omega_j$ & $\cdots$ \\
\hline
\multicolumn{1}{|c|}{$\rmfill$} & 
	\multicolumn{2}{c|}{\cellcolor{lightgray}$(\vec p:1) \vee ((\vec d_i \times \vec d_j):0)$}&
	\multicolumn{1}{c|}{$\rmfill$} &
	\multicolumn{2}{c|}{\cellcolor{lightgray}$-((\vec p:1) \vee ((\vec d_i \times \vec d_j):0))$}&	
	\multicolumn{1}{c|}{$\rmfill$} \\
\hline
\end{tabular}
\end{center}

%
\begin{figure}[h]
\centering \subfloat[{\em Coincidence:} a line $(\vec p_i, \vec d_i)$ affixed to body $i$
must lie in a plane with normal $\vec d_j$ affixed to body $j$. Then, in addition to
a line-plane parallel constraint, the instantaneous 
velocity $\vec p_i'$ of $\vec p_i$ must lie in the plane $(\vec p_i, \vec d_j)$.] {\label{fig.linePlaneCoincident}
\begin{minipage}[b]{0.45\linewidth}
\centering
\includegraphics[width=.75\linewidth]{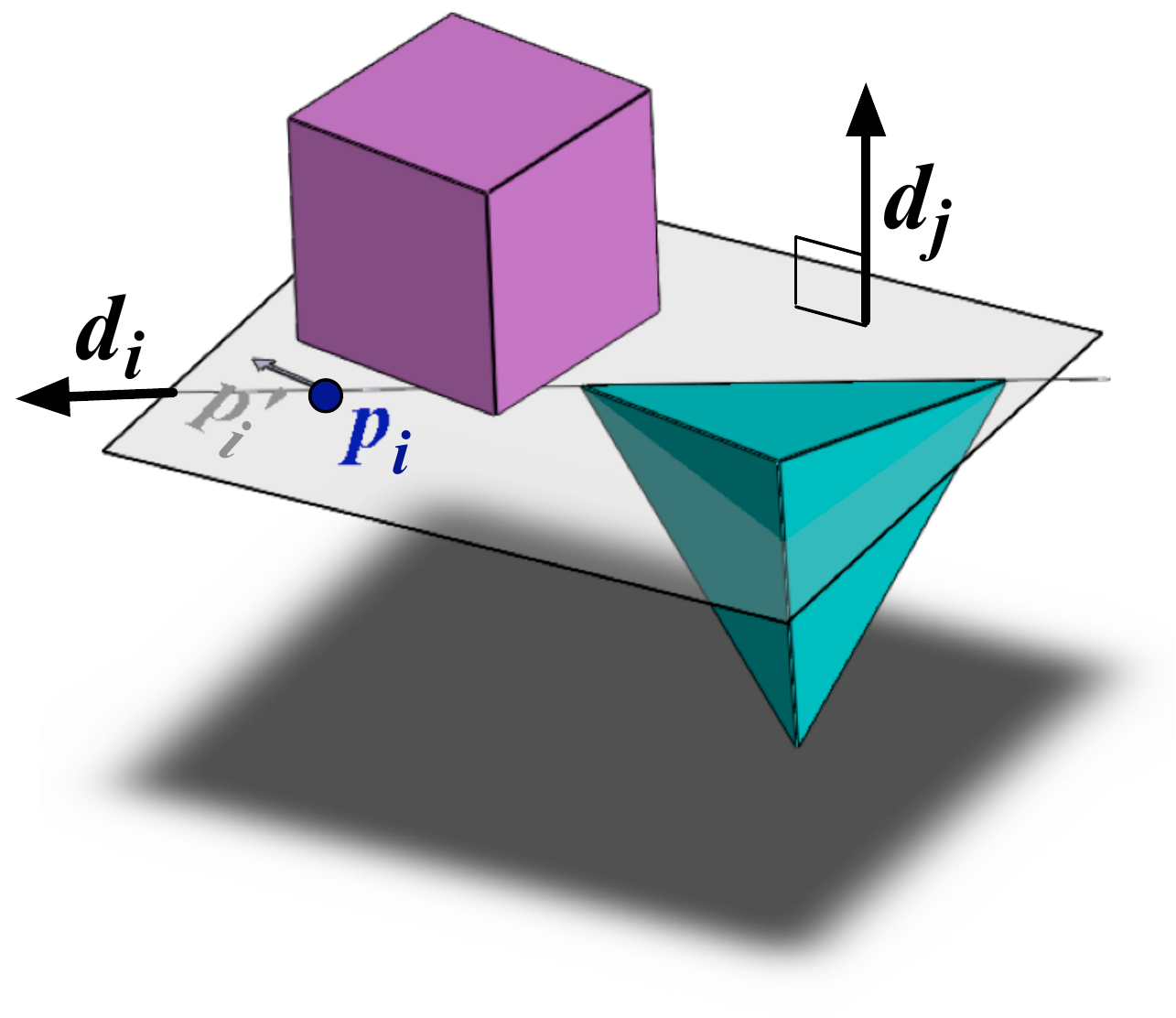}
\end{minipage}}%
\hspace{2mm}
\centering \subfloat[{\em Distance:} a line $(\vec p_i, \vec d_i)$ affixed to body $i$
must be a fixed distance from a plane with normal $\vec d_j$ affixed to body $j$. 
Then, in addition to a line-plane parallel constraint, the instantaneous 
velocity $\vec p_i'$ of $\vec p_i$ must lie in the plane 
$(\vec p_i, \vec d_j)$.]{\label{fig.linePlaneDistance}
\begin{minipage}[b]{0.45\linewidth}
\centering\includegraphics[width=.93\linewidth]{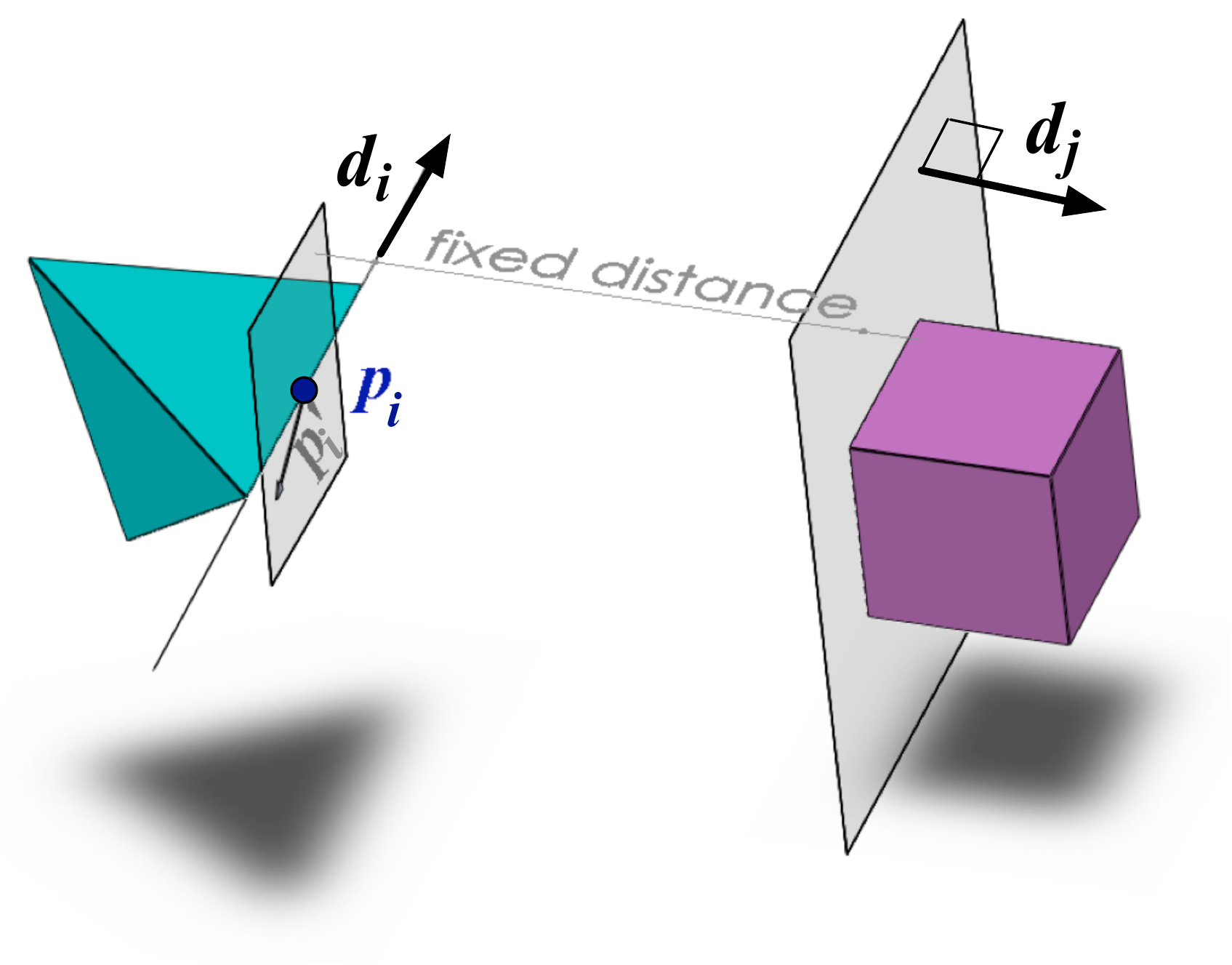}
\end{minipage}}
\caption{Line-plane constraints. }
\label{fig.linePlane}
\end{figure}
\noindent{\bf Line-plane coincidence} {\em (Figure \ref{fig.linePlaneCoincident})}:
The function $L_{15}:E_{15} \rightarrow (\Re^3 \times \Re^3) \times (\Re^3 \times \Re^3)$
maps an edge $e=ij$ to a pair $((\vec p_i, \vec d_i), (\vec p_j, \vec d_j))$ so that
the line $(\vec p_i, \vec d_i)$ affixed to body $i$ is constrained to lie
in the plane $(\vec p_j, \vec d_j)$ affixed to body $j$.

The {\bf line-plane coincidence}
constraint is infinitesimally maintained using a primitive
{\em angular} {\bf line-plane parallel} constraint from Equation \ref{eq:nonParallelAngleInf}:
\begin{equation}
	\label{eq:planeLineCoincAng}
	\left\langle (\vec s_i^* - \vec s_j^*), ((0,0,0), \vec d_j \times \vec d_i) \right\rangle = 0
\end{equation}
In addition, a primitive {\em blind} constraint from Equation \ref{eq:velOrthogToVec} 
forces the relative velocity of $\vec p_i$ to remain in the plane:
\begin{equation}
	\label{eq:planeLineCoincBlind}
	\left\langle (\vec s_i - \vec s_j)^*, (\vec p_i:1) \vee (\vec d_j:0)\right\rangle = 0
\end{equation}
Thus, a line-plane coincidence constraint corresponds to 2 rows in the rigidity matrix:
\begin{center}
\begin{tabular}{ccccccc}
 & \multicolumn{2}{c}{$\vec s_i^*$} & & \multicolumn{2}{c}{$\vec s_j^*$} & \\
\cline{2-3} \cline{5-6}
$\cdots$ & $\vec v_i$& $-\vc\omega_i$ &
$\cdots$ & $\vec v_j$& $-\vc\omega_j$ & $\cdots$ \\
\hline
\multicolumn{1}{|c|}{$\rmfill$} & 
	\multicolumn{1}{c|}{\cellcolor{red}$\vec 0$} & 
	\multicolumn{1}{c|}{\cellcolor{lightgray}$\vec d_j \times\vec  d_i$}&
	\multicolumn{1}{c|}{$\rmfill$} &
	\multicolumn{1}{c|}{\cellcolor{red}$\vec 0$} & 
	\multicolumn{1}{c|}{\cellcolor{lightgray}$\vec d_i \times\vec  d_j$}&	
	\multicolumn{1}{c|}{$\rmfill$} \\
\hline
\multicolumn{1}{|c|}{$\rmfill$} & 
	\multicolumn{2}{c|}{\cellcolor{lightgray}$ (\vec p_i:1) \vee (\vec d_j:0)$}&
	\multicolumn{1}{c|}{$\rmfill$} &
	\multicolumn{2}{c|}{\cellcolor{lightgray}$- ((\vec p_i:1) \vee (\vec d_j:0))$}&	
	\multicolumn{1}{c|}{$\rmfill$} \\
\hline
\end{tabular}
\end{center}

\noindent{\bf Line-plane distance} {\em (Figure \ref{fig.linePlaneDistance})}:
The function $L_{16}:E_{16} \rightarrow (\Re^3 \times \Re^3) \times (\Re^3 \times \Re^3) \times \Re$
maps an edge $e=ij$ to a triple $((\vec p_i, \vec d_i), (\vec p_j, \vec d_j), a)$ so that
the line $(\vec p_i, \vec d_i)$ affixed to body $i$ is constrained to lie a distance $a$
from the plane $(\vec p_j, \vec d_j)$ affixed to body $j$.

The {\bf line-plane distance}
constraint is maintained infinitesimally by using the same equations as for
the {\bf line-plane coincidence} constraint: a primitive
{\em angular} {\bf line-plane parallel} constraint from Equation \ref{eq:nonParallelAngleInf}:

\begin{equation}
	\left\langle (\vec s_i^* - \vec s_j^*), ((0,0,0), \vec d_j \times\vec  d_i) \right\rangle = 0
\end{equation}
In addition, a primitive {\em blind} constraint from Equation \ref{eq:velOrthogToVec} 
forces the relative velocity of $\vec p_i$ to remain parallel to the plane:
\begin{equation}
	\left\langle (\vec s_i - \vec s_j)^*, (\vec p_i:1) \vee (\vec d_j:0)\right\rangle = 0
\end{equation}	
Thus, a line-plane distance constraint corresponds to 2 rows in the rigidity matrix:
	\begin{center}
	\begin{tabular}{ccccccc}
	 & \multicolumn{2}{c}{$\vec s_i^*$} & & \multicolumn{2}{c}{$\vec s_j^*$} & \\
	\cline{2-3} \cline{5-6}
	$\cdots$ & $\vec v_i$& $-\vc\omega_i$ &
	$\cdots$ & $\vec v_j$& $-\vc\omega_j$ & $\cdots$ \\
	\hline
	\multicolumn{1}{|c|}{$\rmfill$} & 
		\multicolumn{1}{c|}{\cellcolor{red}$\vec 0$} & 
		\multicolumn{1}{c|}{\cellcolor{lightgray}$\vec d_j \times\vec  d_i$}&
		\multicolumn{1}{c|}{$\rmfill$} &
		\multicolumn{1}{c|}{\cellcolor{red}$\vec 0$} & 
		\multicolumn{1}{c|}{\cellcolor{lightgray}$\vec d_i \times\vec  d_j$}&	
		\multicolumn{1}{c|}{$\rmfill$} \\
	\hline
	\multicolumn{1}{|c|}{$\rmfill$} & 
		\multicolumn{2}{c|}{\cellcolor{lightgray}$ (\vec p_i:1) \vee (\vec d_j:0)$}&
		\multicolumn{1}{c|}{$\rmfill$} &
		\multicolumn{2}{c|}{\cellcolor{lightgray}$- ((\vec p_i:1) \vee (\vec d_j:0))$}&	
		\multicolumn{1}{c|}{$\rmfill$} \\
	\hline
	\end{tabular}
	\end{center}

%% file: planeConstraints.tex
\begin{figure}[h]
\centering \subfloat[{\em Coincidence:}	bodies $i$ and $j$ must
coincide on the specified plane $(\vec p, \vec d)$. Then, in addition
to a plane-plane parallel constraint, the instantaneous velocity $\vec p'$ of $\vec p$
must remain in the plane.] {\label{fig.planePlaneCoincident}
\begin{minipage}[b]{0.45\linewidth}
\centering
\includegraphics[width=.8\linewidth]{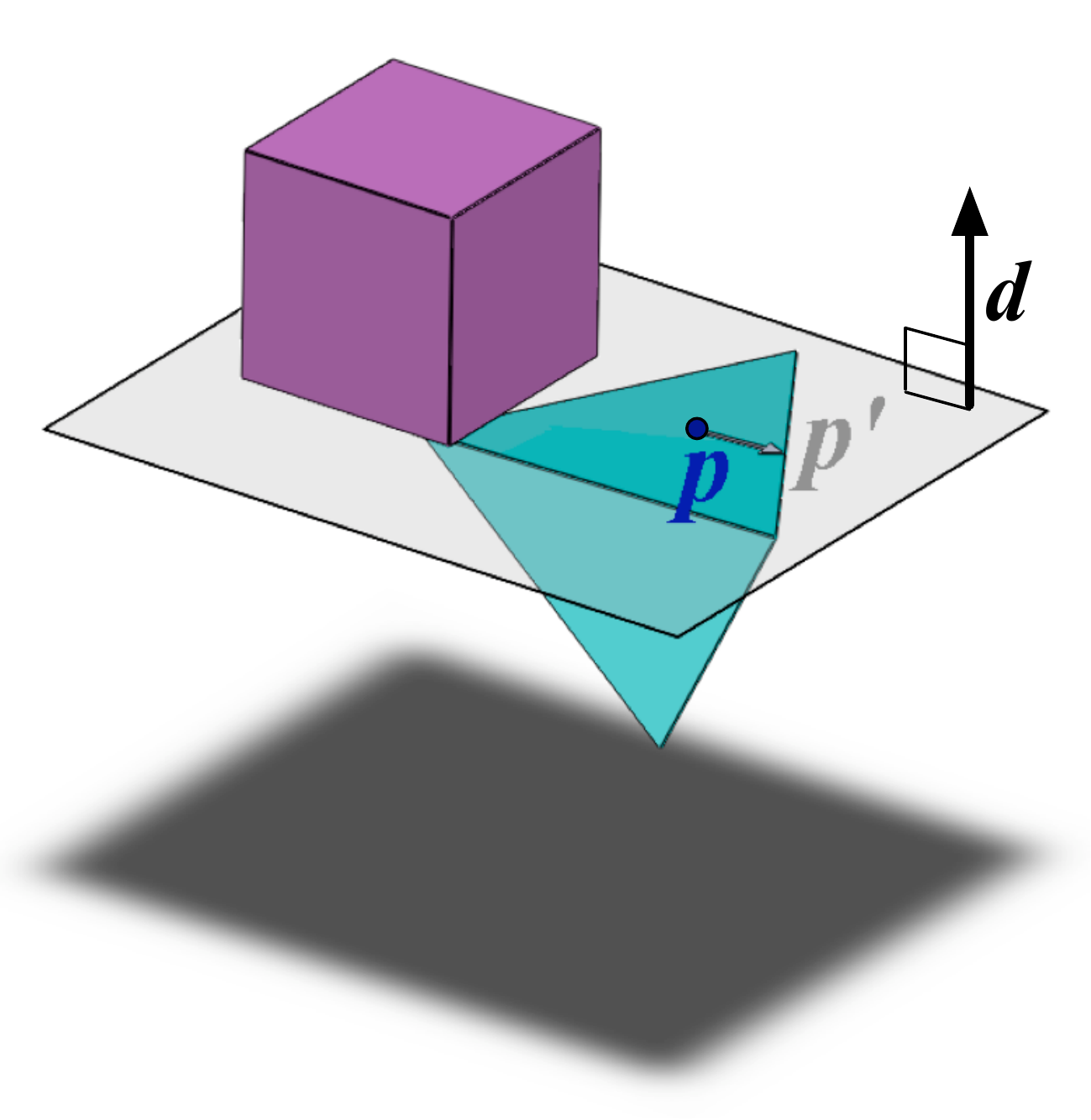}
\end{minipage}}%
\hspace{2mm}
\centering \subfloat[{\em Distance:} a plane $(\vec p_i, \vec d)$ affixed
to body $i$ must be a fixed distance from a plane $(\vec p_j, \vec d)$
affixed to body $j$. Then, in addition to a plane-plane parallel constraint,
the instantaneous velocity $\vec p_i'$ of $\vec p_i$ must remain in the plane
$(\vec p_i, \vec d)$.]{\label{fig.planePlaneDistance}
\begin{minipage}[b]{0.45\linewidth}
\centering\includegraphics[width=\linewidth]{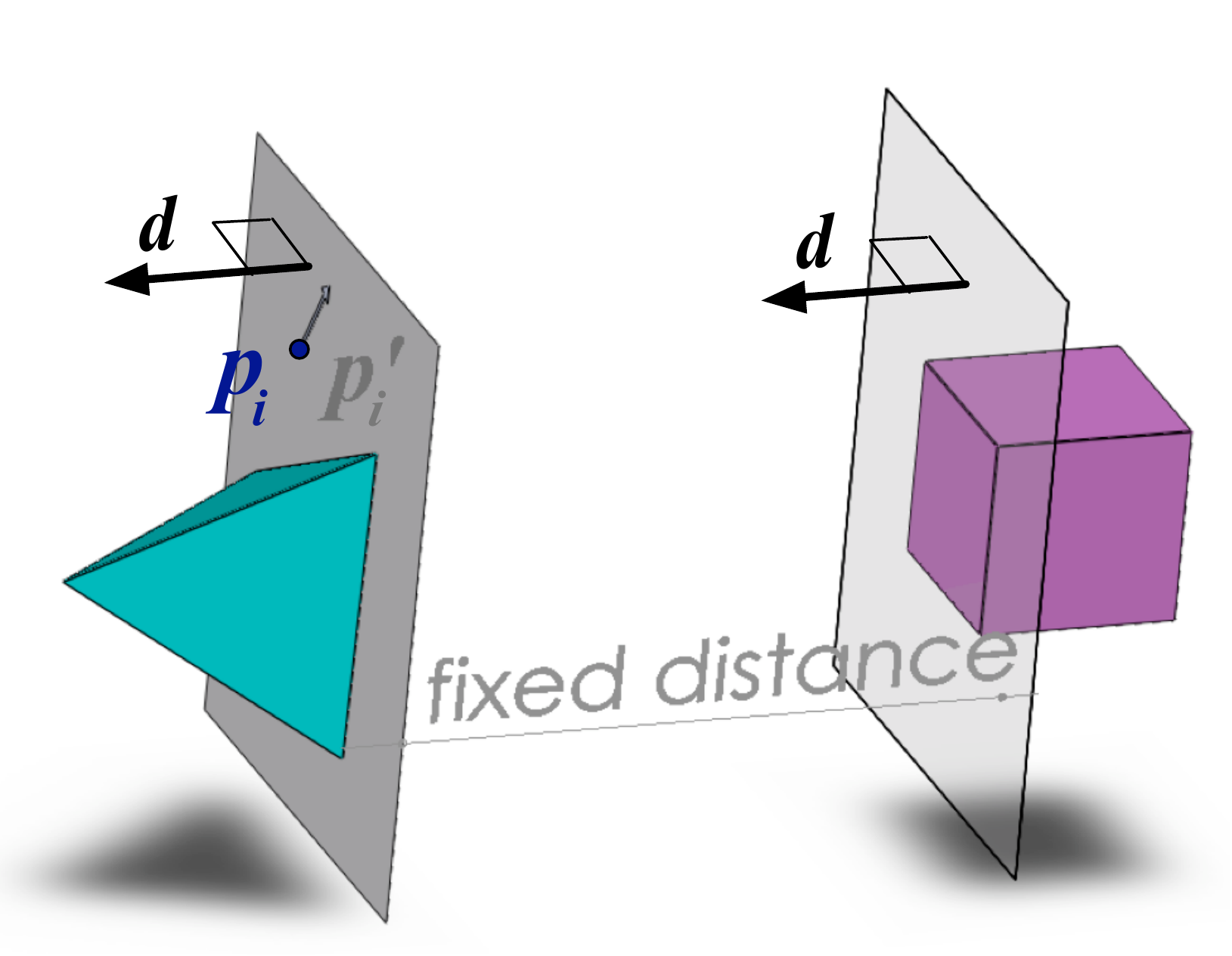}
\end{minipage}}
\caption{Plane-plane constraints. }
\label{fig.planePlane}
\end{figure}
\noindent{\bf Plane-plane coincidence} {\em (Figure \ref{fig.planePlaneCoincident})}:
The function $L_{20}:E_{20} \rightarrow \Re^3 \times \Re^3 $
maps an edge $e=ij$ to the pair $(\vec p, \vec d)$ so that
the plane $(\vec p, \vec d)$ 
is constrained to be affixed to both bodies $i$ and $j$ simultaneously.

We place a {\bf plane-plane parallel}
angular constraint, resulting in 2 primitive {\em angular} constraints
from Equations
\ref{eq:parallelAngleInf1} and \ref{eq:parallelAngleInf2}:
\begin{eqnarray}
	\left\langle (\vec s_i^* - \vec s_j^*), (0,0,0,-d^y,d^x, 0)\right\rangle \\
	\left\langle (\vec s_i^* - \vec s_j^*), (0,0,0,0,-d^z,d^y)\right\rangle
\end{eqnarray}
Then, to 
maintain coincidence, place a primitive {\em blind} constraint
by using Equation \ref{eq:velOrthogToVec} to force the
relative velocity of $\vec p$ to remain in the plane:
\begin{equation}
	\left\langle (\vec s_i - \vec s_j)^*, (\vec p:1) \vee (\vec d:0)\right\rangle = 0
\end{equation}
Thus, a plane-plane coincidence corresponds to 3 rows in the rigidity matrix:
\begin{center}
\begin{tabular}{ccccccc}
 & \multicolumn{2}{c}{$\vec s_i^*$} & & \multicolumn{2}{c}{$\vec s_j^*$} & \\
\cline{2-3} \cline{5-6}
$\cdots$ & $\vec v_i$& $-\vc\omega_i$ &
$\cdots$ & $\vec v_j$& $-\vc\omega_j$ & $\cdots$ \\
\hline
\multicolumn{1}{|c|}{$\rmfill$} & 
	\multicolumn{1}{c|}{\cellcolor{red}$\vec 0$} & 
	\multicolumn{1}{c|}{\cellcolor{lightgray}$(-d^y, d^x, 0)$}&
	\multicolumn{1}{c|}{$\rmfill$} &
	\multicolumn{1}{c|}{\cellcolor{red}$\vec 0$} & 
	\multicolumn{1}{c|}{\cellcolor{lightgray}$(d^y, -d^x, 0)$}&	
	\multicolumn{1}{c|}{$\rmfill$} \\
\hline
\multicolumn{1}{|c|}{$\rmfill$} & 
	\multicolumn{1}{c|}{\cellcolor{red}$\vec 0$} & 
	\multicolumn{1}{c|}{\cellcolor{lightgray}$(0, -d^z, d^y)$}&
	\multicolumn{1}{c|}{$\rmfill$} &
	\multicolumn{1}{c|}{\cellcolor{red}$\vec 0$} & 
	\multicolumn{1}{c|}{\cellcolor{lightgray}$(0, d^z, -d^y)$}&	
	\multicolumn{1}{c|}{$\rmfill$} \\
\hline
\multicolumn{1}{|c|}{$\rmfill$} & 
	\multicolumn{2}{c|}{\cellcolor{lightgray}$(\vec p:1) \vee (\vec d:0)$}&
	\multicolumn{1}{c|}{$\rmfill$} &
	\multicolumn{2}{c|}{\cellcolor{lightgray}$-((\vec p:1) \vee (\vec d:0))$}&	
	\multicolumn{1}{c|}{$\rmfill$} \\
\hline
\end{tabular}
\end{center}

\noindent{\bf Plane-plane distance} {\em (Figure \ref{fig.planePlaneDistance})}:
The function $L_{21}:E_{21} \rightarrow \Re^3 \times \Re^3 \times \Re^3 \times \Re$
maps an edge $e=ij$ to a quadruple $(\vec p_i, \vec p_j, \vec d, a)$ so that
the planes $(\vec p_i, \vec d)$ and $(\vec p_j, \vec d)$ affixed to bodies $i$ and $j$, respectively,
are constrained to have the distance $a$ between them.

Maintaining the plane-plane distance infinitesimally reduces to the 
same linear equations as for plane-plane coincidence. 
We place a {\bf plane-plane parallel}
angular constraint, resulting in 2 primitive
{\em angular} constraints from Equations \ref{eq:parallelAngleInf1} and \ref{eq:parallelAngleInf2}
along with a primitive {\em blind} constraint using Equation \ref{eq:velOrthogToVec} to force
the relative velocity of $\vec p_i$ to remain parallel to the plane:
\begin{eqnarray}
	\left\langle (\vec s_i^* - \vec s_j^*), (0,0,0,-d^y,d^x, 0)\right\rangle \\
	\left\langle (\vec s_i^* - \vec s_j^*), (0,0,0,0,-d^z,d^y)\right\rangle
\end{eqnarray}
\begin{equation}
	\left\langle (\vec s_i - \vec s_j)^*, (\vec p_i:1) \vee (\vec d:0)\right\rangle = 0
\end{equation}
Thus, a plane-plane distance constraint corresponds to 3 rows in the rigidity matrix:
\begin{center}
\begin{tabular}{ccccccc}
 & \multicolumn{2}{c}{$\vec s_i^*$} & & \multicolumn{2}{c}{$\vec s_j^*$} & \\
\cline{2-3} \cline{5-6}
$\cdots$ & $\vec v_i$& $-\vc\omega_i$ &
$\cdots$ & $\vec v_j$& $-\vc\omega_j$ & $\cdots$ \\
\hline
\multicolumn{1}{|c|}{$\rmfill$} & 
	\multicolumn{1}{c|}{\cellcolor{red}$\vec 0$} & 
	\multicolumn{1}{c|}{\cellcolor{lightgray}$(-d^y, d^x, 0)$}&
	\multicolumn{1}{c|}{$\rmfill$} &
	\multicolumn{1}{c|}{\cellcolor{red}$\vec 0$} & 
	\multicolumn{1}{c|}{\cellcolor{lightgray}$(d^y, -d^x, 0)$}&	
	\multicolumn{1}{c|}{$\rmfill$} \\
\hline
\multicolumn{1}{|c|}{$\rmfill$} & 
	\multicolumn{1}{c|}{\cellcolor{red}$\vec 0$} & 
	\multicolumn{1}{c|}{\cellcolor{lightgray}$(0, -d^z, d^y)$}&
	\multicolumn{1}{c|}{$\rmfill$} &
	\multicolumn{1}{c|}{\cellcolor{red}$\vec 0$} & 
	\multicolumn{1}{c|}{\cellcolor{lightgray}$(0, d^z, -d^y)$}&	
	\multicolumn{1}{c|}{$\rmfill$} \\
\hline
\multicolumn{1}{|c|}{$\rmfill$} & 
	\multicolumn{2}{c|}{\cellcolor{lightgray}$(\vec p_i:1) \vee (\vec d:0)$}&
	\multicolumn{1}{c|}{$\rmfill$} &
	\multicolumn{2}{c|}{\cellcolor{lightgray}$-((\vec p_i:1) \vee (\vec d:0))$}&	
	\multicolumn{1}{c|}{$\rmfill$} \\
\hline
\end{tabular}
\end{center}

%% file: example.tex
\subsubsection{Example}
To help the reader, we complete the formalization of the dice example depicted in 
Figures \ref{fig.dice} and \ref{fig.diceCadGraph}. We assume that the $z$-axis lies
along the base of the dice in the direction of Face $2$, with the $xy$-plane parallel to Face $3$.

Then the framework is described by the functions:
\begin{itemize}
\item $L_{17}(e_{(i)}) = ((0, 2, 0), (0,1,0), (0,1,0))$ 

\item $L_{18}(e_{(ii)}) = (((0, 2, 0), (1,0,0)), ((0,0,1), (0,0,1) ))$ 

\item $L_{16}(e_{(iii)}) = (((0, 2, 0), (0,0,1)), ((0,1,0),(0,1,0)),1)$ 

\item $L_1(e_{(iv)}) = (0, 1, 1)$ 
\end{itemize}
Since the example only uses four 
types of constraints, we omit the description of the remaining $L_i$ functions.

For each edge, we construct the associated rows in the rigidity matrix, resulting
in the following:

\begin{center}
\begin{tabular}{rcccccccccccc}
& \multicolumn{6}{c}{$\vec s_A^*$} & \multicolumn{6}{c}{$\vec s_B^*$} \\
&\multicolumn{3}{c}{$\vec v_A$} & \multicolumn{3}{c}{$-\vc\omega_A$} &
\multicolumn{3}{c}{$\vec v_B$}& \multicolumn{3}{c}{$-\vc\omega_B$} \\
\multirow{2}{*}{$e_{(i)}\left\{ \begin{array}{c}\hbox{}\\\hbox{}\end{array} \right.$} &
\multicolumn{1}{|c|}{\cellcolor{red}$0$} & 
\multicolumn{1}{c|}{\cellcolor{red}$0$} & 
\multicolumn{1}{c|}{\cellcolor{red}$0$} & 
\multicolumn{1}{c|}{\cellcolor{lightgray}$-1$}&
\multicolumn{1}{c|}{\cellcolor{lightgray}$0$}&
\multicolumn{1}{c|}{\cellcolor{lightgray}$0$}&
\multicolumn{1}{c|}{\cellcolor{red}$0$} & 
\multicolumn{1}{c|}{\cellcolor{red}$0$} & 
\multicolumn{1}{c|}{\cellcolor{red}$0$} & 
\multicolumn{1}{c|}{\cellcolor{lightgray}$1$}&
\multicolumn{1}{c|}{\cellcolor{lightgray}$0$}&
\multicolumn{1}{c|}{\cellcolor{lightgray}$0$}\\
\cline{2-13}
&\multicolumn{1}{|c|}{\cellcolor{red}$0$} & 
\multicolumn{1}{c|}{\cellcolor{red}$0$} & 
\multicolumn{1}{c|}{\cellcolor{red}$0$} & 
\multicolumn{1}{c|}{\cellcolor{lightgray}$0$}&
\multicolumn{1}{c|}{\cellcolor{lightgray}$0$}&
\multicolumn{1}{c|}{\cellcolor{lightgray}$1$}&
	\multicolumn{1}{c|}{\cellcolor{red}$0$} & 
	\multicolumn{1}{c|}{\cellcolor{red}$0$} & 
	\multicolumn{1}{c|}{\cellcolor{red}$0$} & 
	\multicolumn{1}{c|}{\cellcolor{lightgray}$0$}&
	\multicolumn{1}{c|}{\cellcolor{lightgray}$0$}&
	\multicolumn{1}{c|}{\cellcolor{lightgray}$-1$}\\
\cline{2-13}
$e_{(ii)}\left\{ \begin{array}{c}\hbox{}\end{array} \right.$ &
\multicolumn{1}{|c|}{\cellcolor{red}$0$} & 
\multicolumn{1}{c|}{\cellcolor{red}$0$} & 
\multicolumn{1}{c|}{\cellcolor{red}$0$} & 
\multicolumn{1}{c|}{\cellcolor{lightgray}$0$}&
\multicolumn{1}{c|}{\cellcolor{lightgray}$1$}&
\multicolumn{1}{c|}{\cellcolor{lightgray}$0$}&
	\multicolumn{1}{c|}{\cellcolor{red}$0$} & 
	\multicolumn{1}{c|}{\cellcolor{red}$0$} & 
	\multicolumn{1}{c|}{\cellcolor{red}$0$} & 
	\multicolumn{1}{c|}{\cellcolor{lightgray}$0$}&
	\multicolumn{1}{c|}{\cellcolor{lightgray}$-1$}&
	\multicolumn{1}{c|}{\cellcolor{lightgray}$0$}\\
\cline{2-13}
\multirow{2}{*}{$e_{(iii)}\left\{ \begin{array}{c}\hbox{}\\\hbox{}\end{array} \right.$} &
\multicolumn{1}{|c|}{\cellcolor{red}$0$} & 
\multicolumn{1}{c|}{\cellcolor{red}$0$} & 
\multicolumn{1}{c|}{\cellcolor{red}$0$} & 
\multicolumn{1}{c|}{\cellcolor{lightgray}$1$}&
\multicolumn{1}{c|}{\cellcolor{lightgray}$0$}&
\multicolumn{1}{c|}{\cellcolor{lightgray}$0$}&
	\multicolumn{1}{c|}{\cellcolor{red}$0$} & 
	\multicolumn{1}{c|}{\cellcolor{red}$0$} & 
	\multicolumn{1}{c|}{\cellcolor{red}$0$} & 
	\multicolumn{1}{c|}{\cellcolor{lightgray}$-1$}&
	\multicolumn{1}{c|}{\cellcolor{lightgray}$0$}&
	\multicolumn{1}{c|}{\cellcolor{lightgray}$0$}\\
\cline{2-13}
&\multicolumn{1}{|c|}{\cellcolor{lightgray}$0$}&
\multicolumn{1}{c|}{\cellcolor{lightgray}$-1$}&
\multicolumn{1}{c|}{\cellcolor{lightgray}$0$}&
\multicolumn{1}{c|}{\cellcolor{lightgray}$0$}&
\multicolumn{1}{c|}{\cellcolor{lightgray}$0$}&
\multicolumn{1}{c|}{\cellcolor{lightgray}$0$}&
\multicolumn{1}{c|}{\cellcolor{lightgray}$0$}&
\multicolumn{1}{c|}{\cellcolor{lightgray}$1$}&
\multicolumn{1}{c|}{\cellcolor{lightgray}$0$}&
\multicolumn{1}{c|}{\cellcolor{lightgray}$0$}&
\multicolumn{1}{c|}{\cellcolor{lightgray}$0$}&
\multicolumn{1}{c|}{\cellcolor{lightgray}$0$}\\
\cline{2-13}
\multirow{3}{*}{$e_{(iv)}\left\{ \begin{array}{c}\hbox{}\\\hbox{}\\\hbox{}\end{array} \right.$} &
\multicolumn{1}{|c|}{\cellcolor{lightgray}$1$}&
\multicolumn{1}{c|}{\cellcolor{lightgray}$0$}&
\multicolumn{1}{c|}{\cellcolor{lightgray}$0$}&
\multicolumn{1}{c|}{\cellcolor{lightgray}$0$}&
\multicolumn{1}{c|}{\cellcolor{lightgray}$-1$}&
\multicolumn{1}{c|}{\cellcolor{lightgray}$1$}&

\multicolumn{1}{c|}{\cellcolor{lightgray}$-1$}&
\multicolumn{1}{c|}{\cellcolor{lightgray}$0$}&
\multicolumn{1}{c|}{\cellcolor{lightgray}$0$}&
\multicolumn{1}{c|}{\cellcolor{lightgray}$0$}&
\multicolumn{1}{c|}{\cellcolor{lightgray}$1$}&
\multicolumn{1}{c|}{\cellcolor{lightgray}$-1$}\\
\cline{2-13}
&\multicolumn{1}{|c|}{\cellcolor{lightgray}$0$}&
\multicolumn{1}{c|}{\cellcolor{lightgray}$1$}&
\multicolumn{1}{c|}{\cellcolor{lightgray}$0$}&
\multicolumn{1}{c|}{\cellcolor{lightgray}$1$}&
\multicolumn{1}{c|}{\cellcolor{lightgray}$0$}&
\multicolumn{1}{c|}{\cellcolor{lightgray}$0$}&

\multicolumn{1}{c|}{\cellcolor{lightgray}$0$}&
\multicolumn{1}{c|}{\cellcolor{lightgray}$-1$}&
\multicolumn{1}{c|}{\cellcolor{lightgray}$0$}&
\multicolumn{1}{c|}{\cellcolor{lightgray}$-1$}&
\multicolumn{1}{c|}{\cellcolor{lightgray}$0$}&
\multicolumn{1}{c|}{\cellcolor{lightgray}$0$}\\
\cline{2-13}
&\multicolumn{1}{|c|}{\cellcolor{lightgray}$0$}&
\multicolumn{1}{c|}{\cellcolor{lightgray}$0$}&
\multicolumn{1}{c|}{\cellcolor{lightgray}$1$}&
\multicolumn{1}{c|}{\cellcolor{lightgray}$-1$}&
\multicolumn{1}{c|}{\cellcolor{lightgray}$0$}&
\multicolumn{1}{c|}{\cellcolor{lightgray}$0$}&

\multicolumn{1}{c|}{\cellcolor{lightgray}$0$}&
\multicolumn{1}{c|}{\cellcolor{lightgray}$0$}&
\multicolumn{1}{c|}{\cellcolor{lightgray}$-1$}&
\multicolumn{1}{c|}{\cellcolor{lightgray}$1$}&
\multicolumn{1}{c|}{\cellcolor{lightgray}$0$}&
\multicolumn{1}{c|}{\cellcolor{lightgray}$0$}\\
\cline{2-13}
\end{tabular}
\end{center}

%% file: combinatorics.tex
\section{Combinatorics}
\label{sec:combinatorics}
Now we address the question of combinatorially
characterizing when a body-and-cad rigidity matrix
is generically independent, i.e.,
the rank function drops only on a measure-zero set of
possible entries. The shape of the rigidity matrix
leads to a natural property that we call {\em nested sparsity}.
We show that nested sparsity is a necessary condition for 
body-and-cad rigidity 
and prove by a counterexample that it is insufficient.

\smallskip\noindent{\bf Nested sparsity.}
A graph on $n$ vertices is {\em $(k, \ell)$-sparse} if every subset of $n'$ vertices spans
at most $kn' - \ell$ edges; it is {\em tight} if, in addition, it
spans $kn - \ell$ total edges.

Let $G = (V, R \cup B)$ be a graph with its edge set colored
into red and black edges, corresponding to $R$ and $B$, respectively.
We say that $G$ is {\em $(k_1, \ell_1, k_2, \ell_2)$-nested 
sparse} if it is $(k_1,\ell_1)$-sparse and $G_R = (V, R)$ is $(k_2, \ell_2)$-sparse;
the graph is {\em $(k_1, \ell_1, k_2, \ell_2)$-nested tight} if, in addition,
$G$ is $(k_1, \ell_1)$-tight. Note that nested
sparsity only makes sense when $(k_2, \ell_2)$-sparsity is more
restrictive than $(k_1, \ell_1)$-sparsity.

\smallskip\noindent{\bf Primitive cad graphs.}
Given a cad graph $(G,c)$, we define the 
{\em primitive cad graph} $H = (V, R \cup B)$ to be the multigraph obtained
by assigning vertices to bodies and constraints to disjoint edge sets
$R$ and $B$, corresponding
respectively to primitive angular and blind constraints. 
For each edge $e$ with type $c_i$, associate 
primitive angular constraints to edges in $R$ 
and primitive blind constraints to edges in $B$ as described in Table
\ref{table:constraintAssoc}. Figure \ref{fig.dicePrimCadGraph} depicts the primitive cad graph
associated with the dice example from Figure \ref{fig.dice}, whose
cad graph is depicted in Figure \ref{fig.diceCadGraph}.
\begin{figure}[htb]
  \begin{center}
    \includegraphics[width=.7\linewidth]{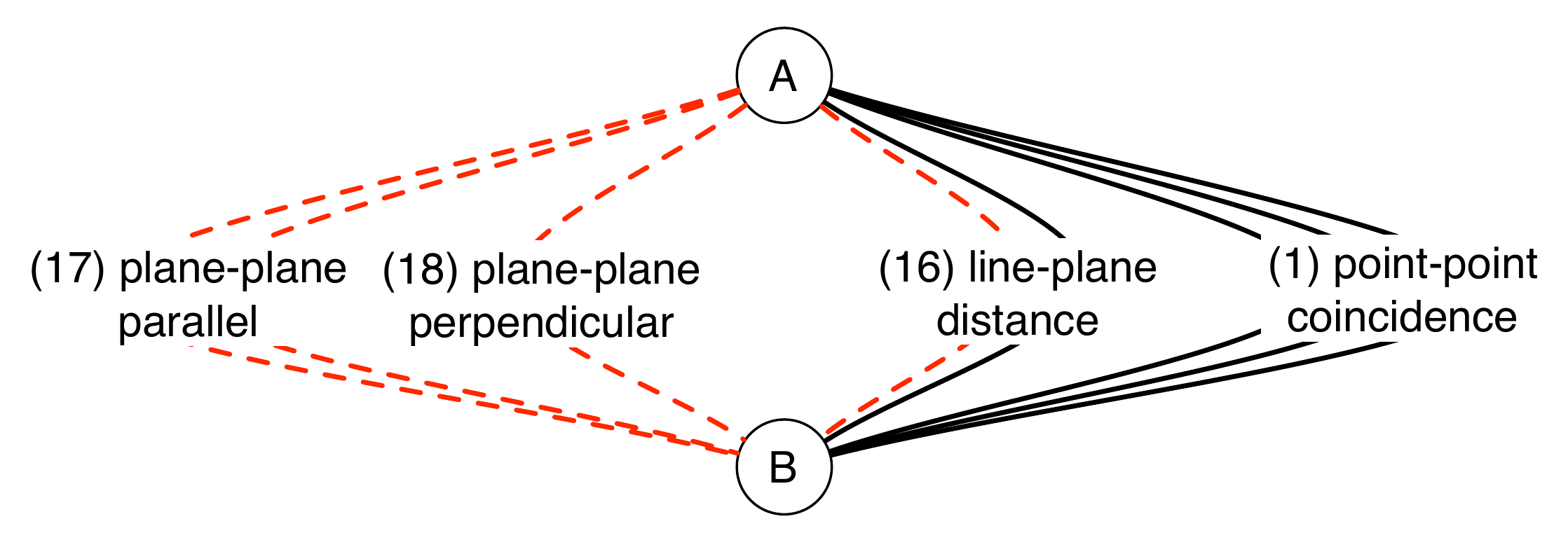}
  \end{center}
  \caption{The primitive cad graph for the example depicted in Figures \ref{fig.dice} 
and \ref{fig.diceCadGraph}; dashed edges denote $R$ edge set.}
  \label{fig.dicePrimCadGraph}
\end{figure}

\begin{thm}
	\label{thm:necessaryCADnested}
	Let $H = (V, R \cup B)$ be the primitive cad graph associated to a body-and-cad framework,
	where $R$ and $B$ correspond to primitive angular and blind constraints,
	respectively.
	Then $(6,6,3,3)$-nested tightness
	is a necessary condition for generic minimal body-and-cad rigidity.
\end{thm}
\begin{proof}
	Let $\mtx A$ be the rigidity matrix associated with $G$. Reorder the columns
	so that the first $3n$ columns correspond to the $-\vc \omega$ elements of the
	screws and the last $3n$ columns correspond to the $\vec v$ elements.
	Reorder the rows to have the $|R|$ rows corresponding to primitive angular
	constraints first; since these rows have all 0s in the last $3n$ columns, 
	we simply consider the submatrix $\mtx A_R$ defined by these $|R|$ 
	rows and the first $3n$ columns. Then
	it is clear that $(3,0)$-sparsity is necessary on $G_R = (V, R)$. To see
	that $(3,3)$-sparsity is necessary, we note that the 3-dimensional
	space of trivial motions of $\so(3)$ (infinitesimal rotations) 
	is a subspace of the kernel of $A_R$. 
	These are defined by the basis $\{\vc \rho_1, \vc \rho_2, \vc \rho_3\}$, 
	where $\vc \rho_1$ is the vector obtained by $n$ copies of $(1, 0, 0)$,
	$\vc \rho_2$ is the vector obtained by $n$ copies of $(0, 1, 0)$,
	and $\vc \rho_3$ is the vector obtained by $n$ copies of $(0, 0, 1)$. 
	Similarly, for the overall $(6,6)$-sparsity, note that we have a 
	6-dimensional space of trivial motions of $\se(3)$ (infinitesimal
	rigid body motions), defined by the basis
	$\{\hat{\vc \rho_1}, \hat{\vc \rho_2}, \hat{\vc \rho_3}, 
	\vc \tau_1, \vc \tau_2, \vc \tau_3\}$, where $\hat{\vc \rho_i}$ 
	simply appends $3n$ zeros to $\vc \rho_i$;
	 $\vc \tau_1$ is the vector obtained by $n$ copies of $(0,0,0,1, 0, 0)$,
	$\vc \tau_2$ is the vector obtained by $n$ copies of $(0,0,0,0, 1, 0)$,
	and $\vc \tau_3$ is the vector obtained by $n$ copies of $(0,0,0,0, 0, 1)$.
\end{proof}
\begin{figure}[tb]
\centering \subfloat[Structure with 3 bodies and 6 constraints.
$A$ and $B$ have 2 point-point distance constraints (denoted by the
gray pair of points and black pair of points)  
and a line-line coincidence constraint (denoted by the shared dotted line).
$A$ and $C$ have a line-line fixed angular constraint (denoted by the dashed lines)
and a plane-plane coincidence constraint (denoted by the two faces with vertical stripes).
$B$ and $C$ have a line-plane coincidence
constraint (denoted by the solid line on $B$ and the checkered face on $C$).
] {\label{fig.nestedCounterStr}
\begin{minipage}[b]{.95\linewidth}
\centering
\includegraphics[width=.85\linewidth]{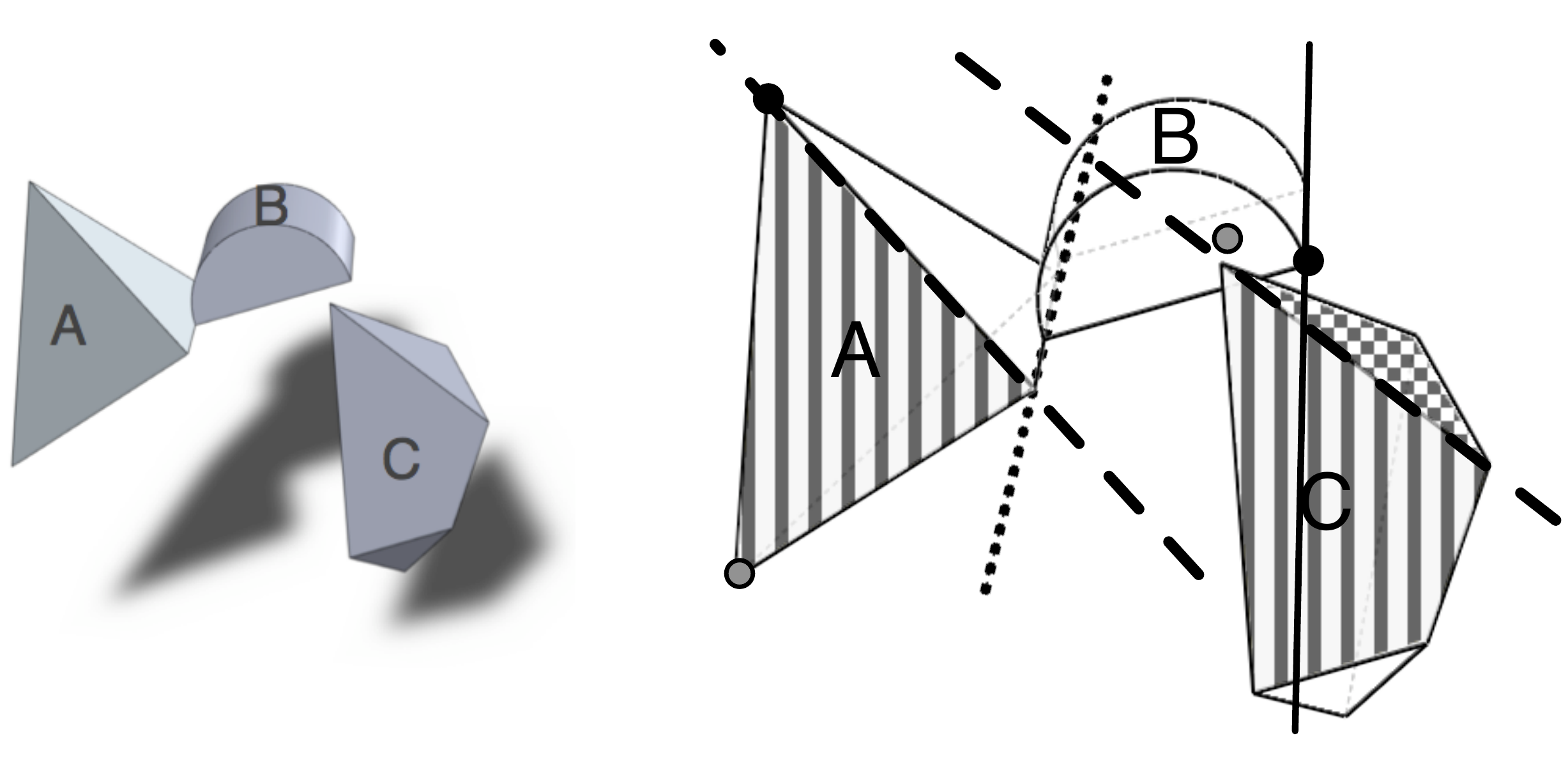}
\end{minipage}}%
\\
\centering \subfloat[The structure is flexible:
$C$ can move relative to $A$ and $B$ by translating
in the direction indicated by the arrows.]{\label{fig.nestedCounterPic}
\begin{minipage}[b]{0.45\linewidth}
\centering\includegraphics[width=.7\linewidth]{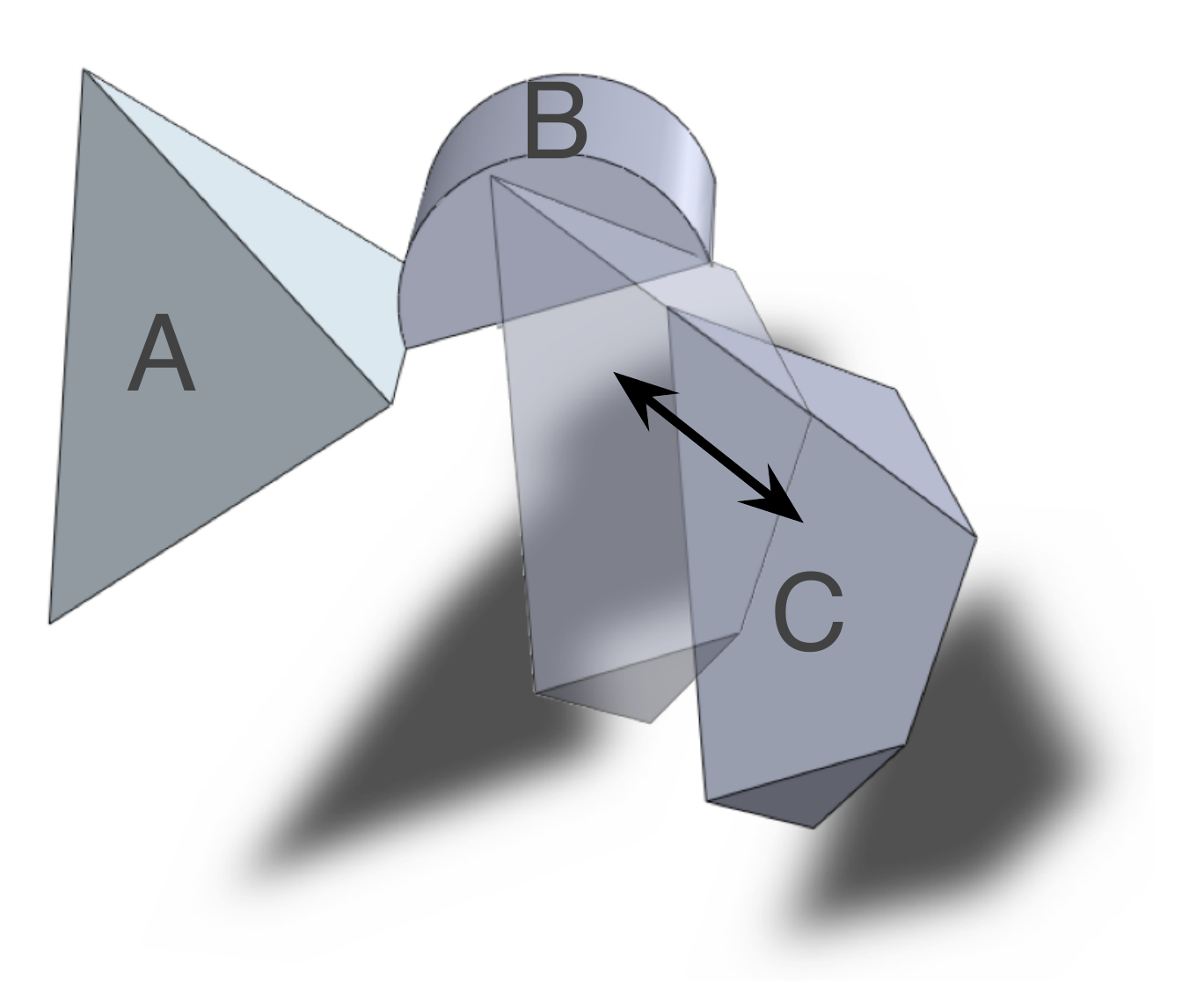}
\end{minipage}}
\hspace{2mm}
\centering \subfloat[Corresponding primitive cad
graph is $(6,6,3,3)$-nested tight; dashed edges denote $R$ edge set.]{\label{fig.nestedCounterGraph}
\begin{minipage}[b]{0.45\linewidth}
\centering\includegraphics[width=\linewidth]{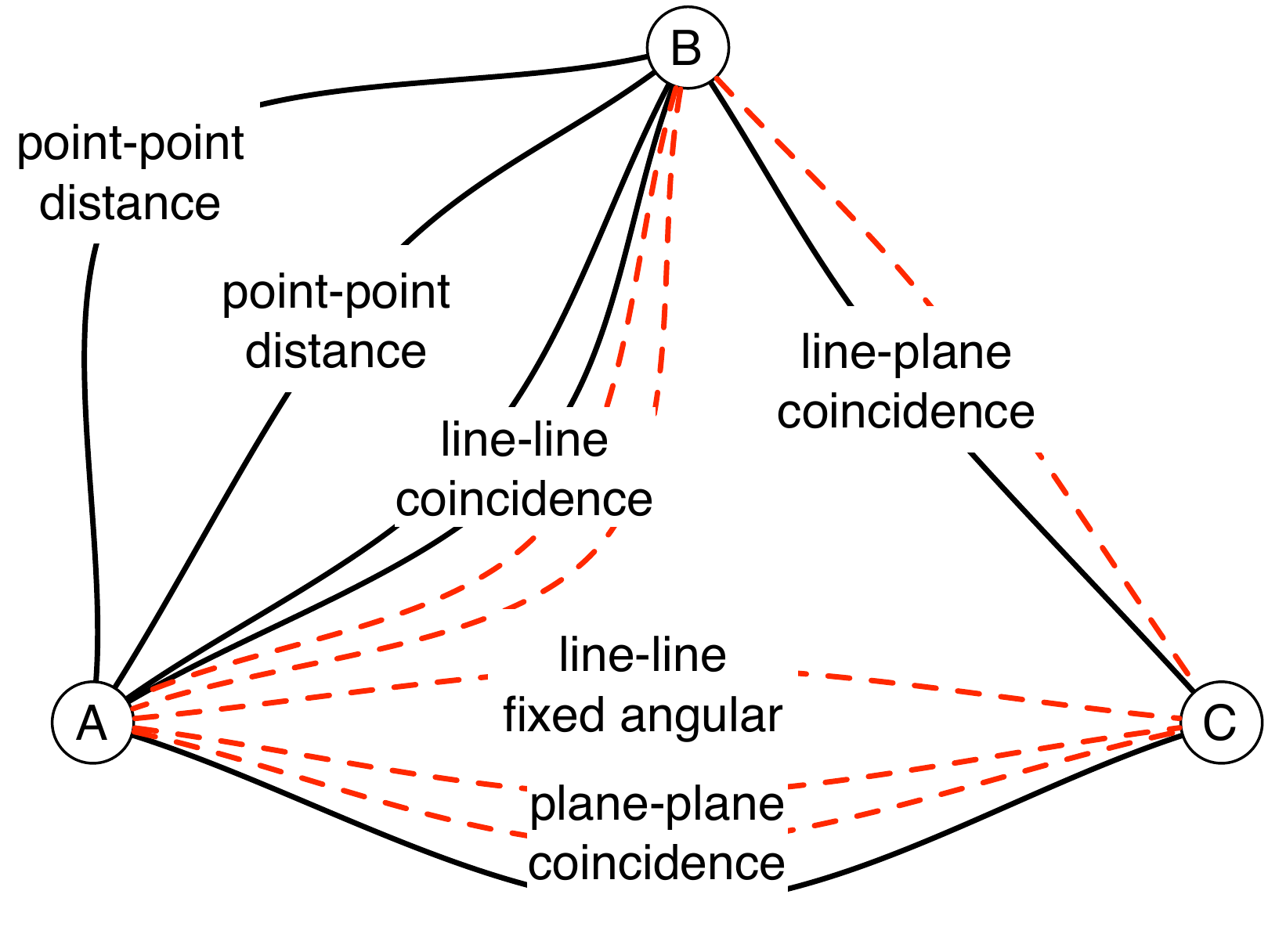}
\end{minipage}}
\caption{Counterexample shows nested sparsity condition is not sufficient. }
\label{fig.nestedCounter}
\end{figure}
\noindent{\bf Counterexample.}
We now show that $(6,6,3,3)$-nested sparsity is 
not sufficient for body-and-cad rigidity. The example in 
Figure \ref{fig.nestedCounter} depicts a {\em flexible} structure
whose associated graph is $(6,6,3,3)$-nested tight. 
It is composed of 3 bodies $A, B$ and $C$; Figure \ref{fig.nestedCounterStr}
depicts the constraints. The structure has one degree of freedom, as indicated
by the arrows in Figure \ref{fig.nestedCounterPic}. The associated primitive cad
graph is shown in Figure \ref{fig.nestedCounterGraph}; the reader may check that 
it is $(6,6,3,3)$-nested tight.

%% file: algs.tex
\section{Algorithms for nested sparsity}
\label{sec:algs}
In the previous section, we defined nested sparsity, proving
that is a necessary, but not sufficient, condition for body-and-cad
rigidity. We now examine the algorithmic aspects of nested sparsity.
\begin{figure}[th]
\centering \subfloat[Original graph is not $(2,2,1,1)$-nested sparse; 
dashed edges denote $R$ edge set.]{\label{fig.counter2211_all}
\begin{minipage}[b]{0.3\linewidth}
\centering
\includegraphics[width=1.6in]{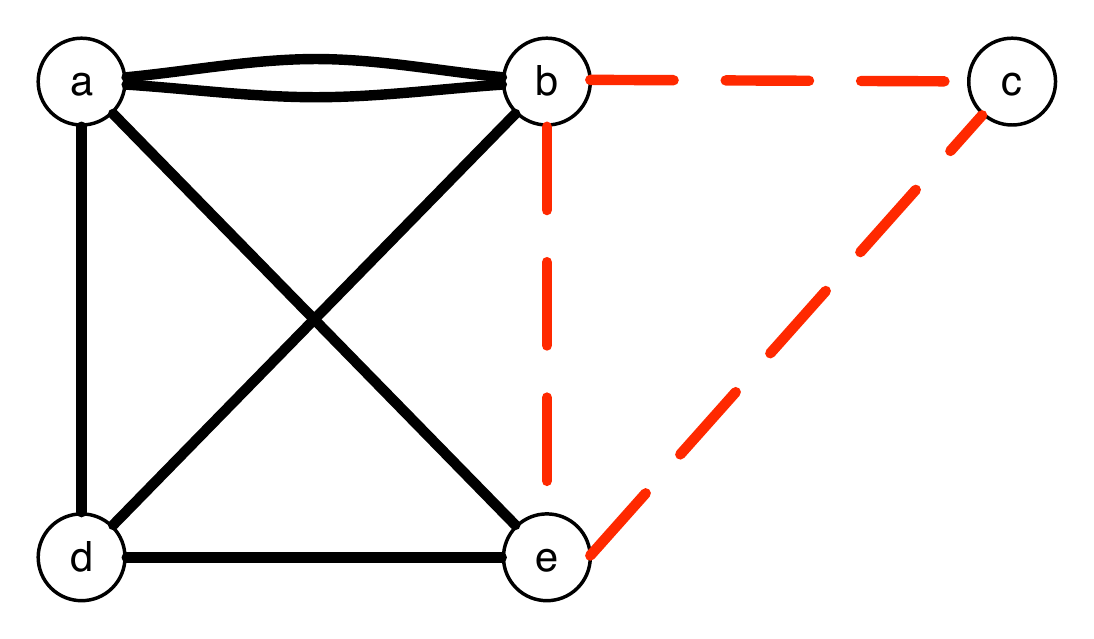}
\end{minipage}}
\hspace{2mm}
\subfloat[Subgraph that is maximally $(2,2,1,1)$-nested sparse with 7
edges.]{\label{fig.counter2211_maximal1}
\begin{minipage}[b]{0.3\linewidth}
\centering
\includegraphics[width=1.6in]{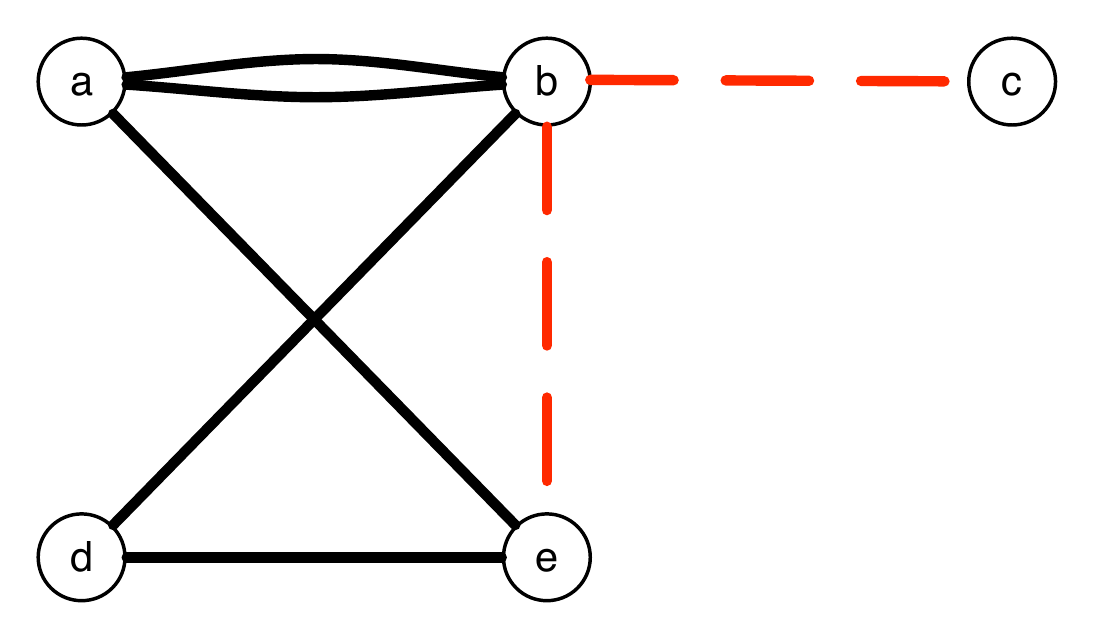}
\end{minipage}}
\hspace{2mm}
\subfloat[Subgraph that is maximally $(2,2,1,1)$-nested sparse with 8
edges.]{\label{fig.counter2211_maximal2}
\begin{minipage}[b]{0.3\linewidth}
\centering
\includegraphics[width=1.6in]{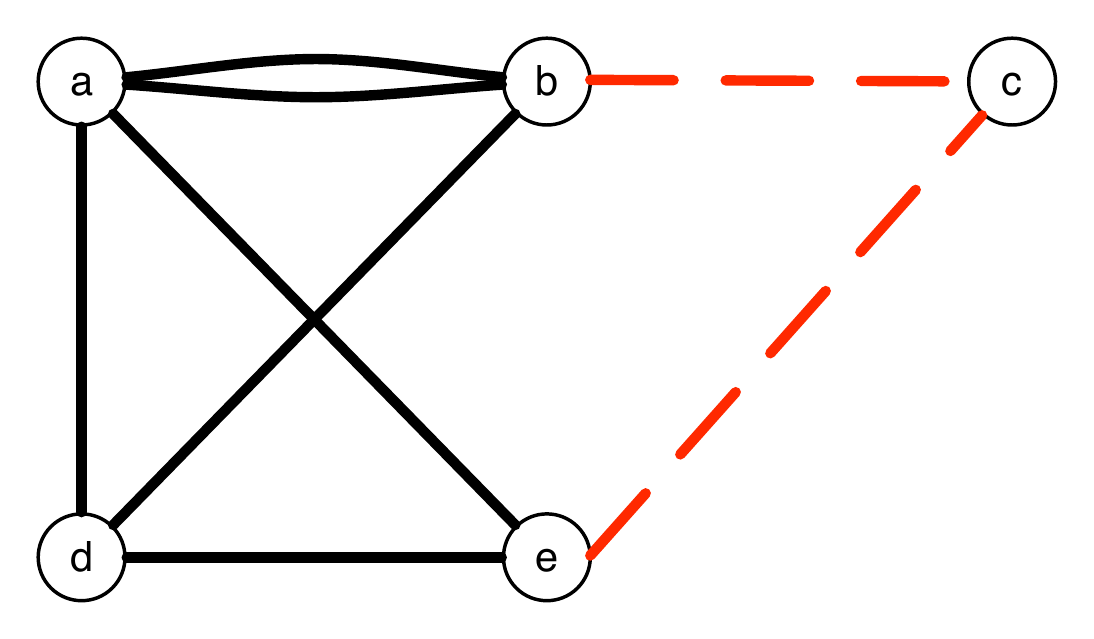}
\end{minipage}}
\caption{Example showing why $(2,2,1,1)$-sparsity is not matroidal;
two maximal subgraphs do not have the same size.}
\label{fig.counter2211}
\end{figure}

We first observe that nested sparsity is not matroidal, as seen by the
example for $(2,2,1,1)$-nested sparsity in Figure
\ref{fig.counter2211}.
However, for certain values of $k_1, \ell_1, k_2$
and $\ell_2$, nested sparsity is the intersection of
two matroids.

\begin{thm}\cite{TheranPersonalCom07}
	\label{thm:nestedIntersection}
	When $0 \leq \ell_i < 2 k_i$, for $i=1,2$, nested sparsity is the intersection
	of two matroids.
\end{thm}

	\begin{proof}	
Define the ground set $E$ to be the complete graph
$K_n^{(2k_1 - \ell_1) + (2k_2 - \ell_2)}$; the edges of the
graph are colored red and black, with $2k_1 - \ell_1$ black
edge multiplicity and $2k_2 - \ell_2$ red edge multiplicity. Let
$M(k,\ell)$ be the bases of the $(k,\ell)$-sparsity matroid; then
$(k_1,\ell_1,k_2,\ell_2)$-nested sparsity is the intersection of 
the following two matroids, defined by their bases:
\begin{enumerate}
	\item $M_1 = \{E' \subseteq E | E' \in M(k_1, \ell_1)\}$, bases in the
	$M(k_1, \ell_1)$-sparsity matroid when edge color is disregarded.
	\item $M_2 = M(k_2,\ell_2) \cup K_n^{2k_1 - \ell_1}$, bases in
	the red $(k_2, \ell_2)$-sparsity matroid padded with full
	edge multiplicity of the black edges.
\end{enumerate}
\end{proof}

As a consequence, when $0 \leq \ell_i < 2 k_i$, the matroid intersection algorithm of
Edmonds \cite{EdmondsMatroidIntersection} can be used to solve the {\bf Decision}
{\em (is a graph 
nested sparse?)}, 
{\bf Extraction} 
{\em (given an input graph, extract a maximum-sized 
nested sparse
subgraph)} and {\bf Components} 
{\em (given an input graph, extract its maximal vertex sets that span 
nested tight subgraphs)}
problems. 

Edmonds' algorithm  outputs a maximum-sized set of edges that are independent
in both matroids and requires an {\em oracle} to test for independence
in each matroid. For the oracles, we use the pebble games algorithms of \cite{streinu:lee:pebbleGames:2008}, a family of algorithms
parametrized by two constants $k$ and $\ell$; the $(k, \ell)$-pebble game characterizes
$(k,\ell)$-sparsity. In particular, the $(k,\ell)$-pebble game takes a graph as
input and can be run in two modes:
the {\bf Decision} mode returns ``yes'' if the input graph is $(k, \ell)$-sparse, and 
the {\bf Components} mode returns the maximal vertex sets that span $(k, \ell)$-tight subgraphs.
Algorithm \ref{alg:nestedSparsity} gives a more detailed
description of how Edmonds' matroid intersection algorithm is used to solve
problems for nested sparsity. 

\begin{figure}[ht]
\noindent \framebox[.95\textwidth][c]{ 
\parbox{.9\textwidth}
{
\begin{myAlg}
	\label{alg:nestedSparsity}
	
\noindent{\bf $(k_1, \ell_1, k_2, \ell_2)$-nested sparsity: Decision, Extraction and Components}\\

\noindent{\em Input:} A graph $G = (V, E = R \cup B)$ and constants $k_1, \ell_1, k_2, \ell_2$, where
$0 \leq \ell_i < 2 k_i$ for $i=1,2$.

\noindent{\em Method:}
\begin{enumerate}
	\item Run Edmonds' matroid intersection algorithm \cite{EdmondsMatroidIntersection} on
		$G$ for the two matroids $M_1$ and $M_2$,
		as defined in the proof of Theorem \ref{thm:nestedIntersection}. 
		When the algorithm performs independence queries on a set of edges $I \subseteq E$,
		\begin{enumerate}
			\item For the matroid $M_1$,
			play the $(k_1, \ell_1)$-pebble game on the input $(V, I)$ in {\bf Decision} mode, which 
			returns ``yes'' if it is $(k_1,\ell_1)$-sparse
			\item For the matroid $M_2$,
			play the $(k_2, \ell_2)$-pebble game on the input $(V, I \cap R)$ in {\bf Decision} mode, which 
			returns ``yes'' if it is $(k_2,\ell_2)$-sparse
		\end{enumerate}
	\item Edmonds' algorithm returns a set $I \subseteq E$ that is of maximum size, where
	 $(V,I)$ is $(k_1, \ell_1, k_2, \ell_2)$-nested sparse.
\item \noindent {\em Output:}
	\begin{itemize}
		\item For {\bf Decision}, {\em ``yes''} if $I = E$ and {\em no} otherwise.
		\item For {\bf Extraction}, $(V,I)$.
		\item For {\bf Components}, play the $(k_1, \ell_1)$-pebble game  in {\bf Components} mode
		on $(V, I)$ and
			{\em output} the components returned by the pebble game.
	\end{itemize}
\end{enumerate}
\end{myAlg}
}} 
\caption{Algorithm for nested sparsity.}
\label{fig:alg:nestedSparsity}
\end{figure}

\noindent{\bf Complexity analysis.} 
Edmonds' algorithm queries the oracles $O(mr^2)$ times, where
$m$ is the number of elements in the ground set and $r$ is the smaller rank of 
the two matroids. For nested sparsity, on a graph $G$ with $n$ vertices and
$m=O(n^2)$ edges, both matroids have rank $O(n)$. 
Therefore, Edmonds' algorithm requires $O(n^4)$ oracle queries. The pebble game
algorithms require $O(n^2)$ time, resulting in $O(n^6)$
 total complexity for Algorithm \ref{alg:nestedSparsity}.
We note that, using the recent matroid intersection algorithm of Harvey \cite{DBLP:journals/siamcomp/Harvey09},
a more efficient running time of $O(mr^{\omega-1})=O(n^{\omega + 1})$, where $\omega$ is the matrix
multiplication exponent, can be obtained for nested sparsity.

Since $(6,6,3,3)$-nested sparsity meets the conditions for Theorem \ref{thm:nestedIntersection}, 
we can apply Algorithm \ref{alg:nestedSparsity} to address the necessary 
condition for body-and-cad rigidity in polynomial time.

%% file: conclusions.tex
\section{Conclusions}
\label{sec:conclusions}
Constraint-based CAD software contains a rich
set of geometric constraints.
Motivated by their applications, 
we have initiated the study of body-and-cad rigidity 
by identifying a class of constraints amenable to 
rigidity-theoretical investigation and developing
their infinitesimal theory.
The shape of the rigidity matrix naturally
led to the study of $(6,6,3,3)$-nested sparsity, a necessary,
but not sufficient, condition for body-and-cad rigidity. 
The polynomial time algorithm we presented for testing this condition
may have practical applications as a filter for finding rigid components 
in a CAD environment, providing informative feedback to the user.

\input{apps.tex}

\smallskip
\noindent{\bf Future directions.}
It remains an open problem to find a combinatorial characterization for
generic body-and-cad rigidity.
We anticipate the study of some of the 
constraints introduced here may prove more tractable than 
classical 3D bar-and-joint rigidity. 
A full combinatorial characterization for angular constraints appears in
\cite{stjohn:streinu:angularRigidityCCCG:2009,lee:PhDThesis:2008}.
However, we observe that finding a complete
characterization may require overcoming well-known obstacles
such as detecting dependencies in 3D bar-and-joint,
2D points-and-angles, 2D circles-and-angles, and 2D point-line
incidence constraint systems. 

Analogous body-and-cad structures for 2D consist of rigid bodies
with  pairwise coincidence (point-point, point-line and line-line), 
angular (line-line) and distance (point-point, point-line and line-line) constraints identified
between points and lines rigidly attached to bodies. 
The development of the rigidity matrix is a straightforward
extension of this work. The interesting question, which remains future work, is a full combinatorial characterization.

\smallskip
\noindent {\bf Acknowledgements.} 
The authors would like to thank
Louis Theran for the observation and proof of Theorem \ref{thm:nestedIntersection}. 
The figures were made using the SolidWorks CAD software \cite{solidworks}.

%% file: apps.tex

\smallskip
\noindent{\bf Applications.}
The results presented can be applied to a larger
set of CAD constraints via simple reductions. In particular,
it is easy to establish the following reductions:

\begin{itemize}
	\item {\bf Sphere-sphere tangency:} Reduces to {\bf point-point distance} using
	the sphere centers and the sum of the radii.
	\item {\bf Sphere-plane tangency:} Reduces to {\bf point-plane distance} using
	the sphere center, the plane and the sphere radius.
	\item {\bf Sphere-line tangency:} Reduces to {\bf point-line distance} using
	the sphere center, the line and the sphere radius.
	\item {\bf Sphere-point coincidence:} Reduces to {\bf point-point distance} using
	the sphere center, the point and the sphere radius.
\end{itemize}
Analogous reductions can be applied when considering cylinders instead of spheres
by substituting the cylinder's center axis for the sphere's center point.